\documentclass[final]{siamonline171218}

\usepackage{lipsum}
\usepackage{amsfonts}
\usepackage{graphicx}
\usepackage{epstopdf}
\usepackage{algorithmicx}
\usepackage{amsmath,amssymb}
\usepackage{algpseudocode}

  \usepackage{paralist}
  \usepackage{graphics} 
  \usepackage{epsfig} 
\usepackage{graphicx}
\graphicspath{ {images/} }
\usepackage{float}
\usepackage{color}

\usepackage[caption=false]{subfig}
 \usepackage{epstopdf}
\usepackage{dsfont}
 \usepackage{multirow}
 \hypersetup{urlcolor=blue, citecolor=red}

\ifpdf
  \DeclareGraphicsExtensions{.eps,.pdf,.png,.jpg}
\else
  \DeclareGraphicsExtensions{.eps}
\fi
\DeclareMathOperator*{\argmin}{\mathbf{argmin}}

\numberwithin{theorem}{section}

\newcommand{\TheTitle}{Non-convex optimization for 3D point source localization using a rotating point spread function}
\newcommand{\TheAuthors}{C. Wang, R. Chan, M. Nikolova, R. Plemmmons and S. Prasad}
\newcommand{\ShortTitle}{non-convex optimization for 3D localization}

\headers{\ShortTitle}{\TheAuthors}

\title{{\TheTitle}}
\title{{\TheTitle}\thanks{Dedicated to the memory of Professor Mila Nikolova who passed away on June 20th, 2018. 
Submitted to the editors on April 1st, 2018. Revision submitted on September 27th, 2018.
\funding{The research of the first, fourth and fifth author were supported by the US Air Force Office of Scientific Research under grant FA9550-15-1-0286. The work of the second author was supported by HKRGC Grants No. CUHK14306316, HKRGC CRF Grant C1007-15G, HKRGC AoE Grant AoE/M-05/12, CUHK DAG No. 4053211, and CUHK FIS Grant No. 1907303. The work of the third author was partially funded by the French Research Agency (ANR)
under grant No ANR-14-CE27-001 (MIRIAM)
and by the Isaac Newton Institute for Mathematical Sciences for
support and hospitality during the programme Variational Methods and Effective Algorithms for Imaging and Vision,
EPSRC  grant  no  EP/K032208/1.
The work of the fourth author was also supported by HKRGC Grants No. CUHK14306316.  }}}
\author{Chao Wang\thanks{Department of Mathematics, The Chinese University of Hong Kong,
Shatin, Hong Kong (\email{chaowang.hk@gmail.com} and \email{rchan@math.cuhk.edu.hk}). }
\and
Raymond Chan\footnotemark[2]
\and Mila Nikolova\thanks{CMLA, CNRS, ENS Cachan, Universit\'{e} Paris-Saclay,
94235 Cachan Cedex, France (deceased).
}
\and Robert Plemmons\thanks{Departments of Computer Science and Mathematics, Wake Forest University, Winston-Salem, NC 27109, USA (\email{plemmons@wfu.edu}). }
\and Sudhakar Prasad\thanks{Department of Physics and Astronomy, The University of New Mexico, Albuquerque, NM 87131, USA  (\email{sprasad@unm.edu}).  }}

\usepackage{amsopn}

\ifpdf
\hypersetup{
  pdftitle={\TheTitle},
  pdfauthor={\TheAuthors}
}
\fi

\begin{document}

\maketitle

\begin{abstract}
 We consider the high-resolution imaging problem of 3D point source image recovery from 2D data using a method based on point spread function (PSF) engineering. The method involves  a new technique, recently proposed by S.~Prasad, based on the use of a rotating PSF with a single lobe to obtain depth from defocus.   The amount of rotation of the PSF encodes the depth position of the point source. Applications include high-resolution single molecule localization microscopy as well as the problem addressed in this paper on localization of space debris using a space-based telescope. The localization problem is discretized on a cubical lattice where the coordinates of nonzero entries represent the 3D locations and the values of these entries the fluxes of the point sources. Finding the locations and fluxes of the point sources is a large-scale sparse 3D inverse problem. A new non-convex regularization method with a data-fitting term based on Kullback-Leibler (KL) divergence is proposed for 3D localization for the Poisson noise model. In addition, we propose a new scheme of estimation of the source fluxes from the KL data-fitting term. Numerical experiments illustrate the efficiency and stability of the algorithms  that are trained on a random subset of image data before being applied to other images. Our 3D localization algorithms can be readily applied to other kinds of depth-encoding PSFs as well.
 \end{abstract}


\section{Introduction} \label{sec:introduction}
Imaging and localizing point sources with high accuracy in a 3D volume is an important and challenging task. An example is super-resolution 3D single-molecule localization \cite{Book_SR_micro2017,3dsml2017review}, which is an area of intense interest in biology (cell imaging, folding, membrane behavior, etc.), in chemistry (spectral diffusion, molecular distortions, etc.), and in physics (structures of materials, quantum optics, etc.).
A different problem of interest in the space-surveillance community is that of localization of space debris using a space-based telescope.  Since the optical wavelength is much shorter than the radio wavelength, optical detection and localization is expected to attain far greater precision than the more commonly employed radar systems. However, the shorter field depth of optical imaging systems may limit their performance to a shorter range of distances. An integrated system consisting of a radar system for performing radio detection, localization, and ranging of space debris at larger distances, which cues in an optical system when debris reach shorter distances, may ultimately provide optimal performance for detecting and tracking debris at distances ranging from tens of kilometers down to hundreds of meters.

A stand-alone optical system based on the use of a light-sheet illumination and scattering concept \cite{englert2014optical}
for spotting debris within meters of a spacecraft has been proposed.
A second system can localize all three coordinates of an unresolved, scattering debris \cite{hampf2015optical,wagner2016detection} by
utilizing either parallex between two observatories or a pulsed laser ranging system or a hybrid system. For parallex, two observatories receive debris scattered optical signal simultaneously. For the pulsed laser, the ranging system is coupled to a single imaging observatory. The hybrid system utilizes both approaches in which the laser pulse transmitted from one of the two observatories is received at time-gated single-photon detectors with good parallax information at both the observatories.
 However, to the best of our knowledge there is no other proposal for a full 3D debris localization and tracking optical or optical-radar system working in the range of tens to hundreds of meters.
 Prasad \cite{Internal_report2016prasad} has proposed
the use of an optical imager that exploits off-center image rotation to encode in a single image snapshot both the range $z$ and transverse ($x,y$) coordinates of a swarm of unresolved sources such as
small, sub-centimeter class space debris, which when actively illuminated can scatter a fraction of laser irradiance back into the imaging sensor.

Here we develop a promising non-convex optimization algorithm and compare its performance with other recent algorithms that can reconstruct the 3D positions and fluxes of a random collection of many point sources within the focal depth of a rotating-PSF-based imaging system from a single image dataset. We shall assume the data to be corrupted only by signal-dependent shot noise that is well described by a Poisson statistical model, as would be roughly characteristic of signal acquisition by an EM-CCD sensor in the photon counting (PC) mode \cite{Daigle2010}.

\subsection{Previous related work}
The area of 3D object localization and imaging is getting increased attention in recent years. One method is to scan the 2D slices of  information at different depths and then reconstruct the 3D image \cite{scan2D2004simultaneous}. Owing to the inefficiency and other limitations of standard PSFs,
several optical modifications \cite{Book_SR_micro2017,3dsml2017review} have been developed for these problems. There are three typical classes of methods: multifocus method, interferometric detection and point spread function (PSF) engineering. The multifocus method \cite{Biplane2008juette,prabhat2004simultaneous} uses more than one focal plane simultaneously. The interferometric method \cite{shtengel2009interferometric} extracts the locations of point sources from the interference fringes of a source signal coherently propagating in two opposed paths in the imaging instrument. PSF engineering \cite{lew2011corkscrew,DH2008pavani,DH2009pavani,prasad2013rotating,tetrapods2014shechtman,tetrapods2015shechtman}, which encodes the depth coordinate of point sources in a 3D scene into a single 2D snapshot, is based on choosing a phase pattern that makes the defocused image of a point source depth-dependent without blurring it excessively.  For example, this can be achieved by inserting a cylindrical lens into an optical system to get approximately elliptical Gaussian, astigmatic images whose axial orientation and size contain unique information about the source locations relative to the plane of best focus \cite{huang2008cylendricalstorm}.
A more general way to create the requisite phase aberrations is to use phase masks. There are many different kinds of phase masks, which give rise to  the double-helix  PSF \cite{DH2008pavani,DH2009pavani}, the corkscrew PSF \cite{lew2011corkscrew}, the Tetrapod PSFs \cite{tetrapods2014shechtman,tetrapods2015shechtman} and the single-lobe PSF \cite{kumar2013psf,prasad2013rotating}. Since the amount of rotation of the PSF encodes the depth information for the double-helix PSF, the corkscrew and the single-lobe PSF, we call them rotating PSFs.
 Tetrapod PSFs \cite{tetrapods2014shechtman,tetrapods2015shechtman} represent another approach that is rather analogous to astigmatic imaging in that the source depth is encoded in the shape of the PSF. 
 
In this paper, we will only focus on Prasad's single-lobe rotating PSF \cite{prasad2013rotating}.
Several algorithms have been developed for {3D} point source localization, for example in the area of single molecule localization in microscopy 3D DAOSTORM \cite{3Ddao2012high} and 3D FALCON \cite{3D_FALCON20143d}. However these two methods cannot be applied for localization using rotating PSFs,  since they are  based on fitting Gaussian functions, which works well only for astigmatism-based PSFs.  Easy-DHPSF \cite{easy_DHPSF2013} is a popular approach for solving the double-helix PSF problem but it requires no overlap for PSF images which is often impossible to satisfy in reality. In \cite{Rice2016generalized}, the authors proposed a generalized method for various rotating PSFs by using 3D deconvolution with a regularization method. The optimization model is for an additive Poisson noise model that is only approximately correct under highly limited brightness conditions. {Their additive Poisson noise model is formed by adding a noise simulated from the Poisson distribution to the image}, {making their noise model unrealistically data-independent}. Owing to sparsity, a hybrid algorithm with matching pursuit and convex optimization was described in \cite{MP_CO2014three}, {but no theoretical analysis was presented}.
{It is therefore} important to consider an algorithm
for the high-density case {that makes no assumptions about source brightness}. With this aim, we propose an optimization model for single-lobe rotating PSF \cite{prasad2013rotating} based 3D localization {of point sources}. We emphasize, however, that our model is
 broadly applicable to a variety of PSFs.
 
\subsection{Problem development}
In image formation modeling, the point spread function (PSF) is the imaging system's response to light from a point source. Our purpose is to extend the work of Prasad on PSF engineering \cite{prasad2013rotating}, which proposed the use of depth-dependent image rotation, by developing a non-convex optimization algorithm for 3D point source localization using such an imaging system.  See Section  \ref{sec:Physics} for details.


By imposing spiral phase retaration with a phase winding number that changes in regular integer steps from one annular zone to the next of an aperture-based phase mask, one can create an image of a point source that has an approximate rotational shape invariance, provided the zone radii have a square root dependence on their indices. Specifically, when the distance of the source from the aperture of such an imaging system changes, the off-center, shape-preserving PSF merely rotates by an amount roughly proportional to the source misfocus from the plane of best focus.
The following general model based on the rotating PSF image describes the spatial distribution of image brightness for $M$ point sources. The observed 2D image without noise is
\begin{equation}
	I_0 (x,y) = \sum_{i=1}^M \mathcal{H}_i(x-x_i,y-y_i)f_i + b,
	\label{equ:forward_model}
\end{equation}
where $b$ is the uniform background and
$\mathcal{H}_i(x-x_i,y-y_i)$ is the rotating PSF for the $i$-th point source of flux $f_i$ and 3D position coordinates $(x_i, y_i, z_i)$ with the depth information coded in each $\mathcal{H}_i$, and $(x,y)$ is  the position in image plane.

Here, we  build a  forward model for the problem based in part  on the approach developed in  \cite{Rice2016generalized}.
In order to estimate the 3D locations of the point sources, {we discretize the distribution of point sources {on a} lattice  $\mathcal{X}\in \mathds{R}^{m \times n \times d}$}. 
The indices of the nonzero entries of $\mathcal{X}$ are the 3-dimensional locations of the point sources and the values at these entries   correspond to the fluxes, i.e., the energy emitted by the illuminated point source. The 2D observed image $G \in \mathds{R}^{m \times n}$ can be {approximated as}
\begin{equation*}
	{G \approx \mathcal{P}\left(\mathcal{T}(\mathcal{A} \ast \mathcal{X}  ) + b  \right),}
\end{equation*}
where $b $ is background signal and   $\mathcal{A} \ast \mathcal{X}$ is the convolution of $\mathcal{X}$ with the 3D PSF $\mathcal{A}$. This 3D PSF is a cube which is constructed by a sequence of images with respect to different depths of the points. Each slice is the image corresponding to a point source at the origin in the $(x,y)$ plane and at depth $z$.  Here $\mathcal{T}$ is an operator for extracting the last slice of the cube $\mathcal{A} \ast \mathcal{X}$ since the observed information is a snapshot, and $\mathcal{P}$ is the Poisson noise operator. 
\Cref{fig:illustration} shows the forward model for a specific PSF.
 In order to recover $\mathcal{X}$,  we need to solve a large-scale sparse 3D inverse problem given as follows:

\begin{equation}\label{model}
        \begin{split}
        	\min\limits_{\mathcal{X}} \   &  \quad \quad \quad  \|\mathcal{X}\|_0 \\
        \text{s.t.}	& \quad \quad  \mathcal{D}\left(\mathcal{X}\right) < \epsilon
        \end{split}
\end{equation}
where $\|\mathcal{X}\|_{0}$ is
the counting function
which gives the number of nonzero entries in $\mathcal{X}$ and is also called the {$\ell_0$ pseudo-norm}.  Here $\epsilon$ provides the noise level and $\mathcal{D}$ is a certain data-fitting term based on the noise model. This problem is $\mathrm{NP}$-hard \cite{NPhard1995natarajan} and there are a number of algorithms and methods \cite{CS2013foucart} being developed to find sparse representation or a good approximate solution.  Examples of such algorithms include greedy algorithms like orthogonal matching pursuit (OMP) \cite{OMP1993pati},  
iterative hard thresholding algorithm (IHT) \cite{IHT2008iterative}, which requires the knowledge of the exact number of sparse entries
in the ground truth,
and convex/non-convex relaxation methods  \cite{L12001atomic,mila2010nonconvex,mila2013nonconvex,mila2008nonconvex,CEL02015soubies,mila2015nonconvex}
where the $\ell_0$ term is approximated and \eqref{model} is replaced by a regularization model.
Here, we develop a convex/non-convex relaxation approach {for} solving (\ref{model}) for the Poisson noise model. In comparison to  state-of-the-art methods, our approach is more efficient even for high density cases. We do not need to know the exact number of the ground truth point sources and, furthermore, we have observed less overfitting than other methods. That is to say that our estimated result recovers fewer false point sources.

\begin{figure}[htbp]
\centering
\includegraphics[width=0.7\textwidth]{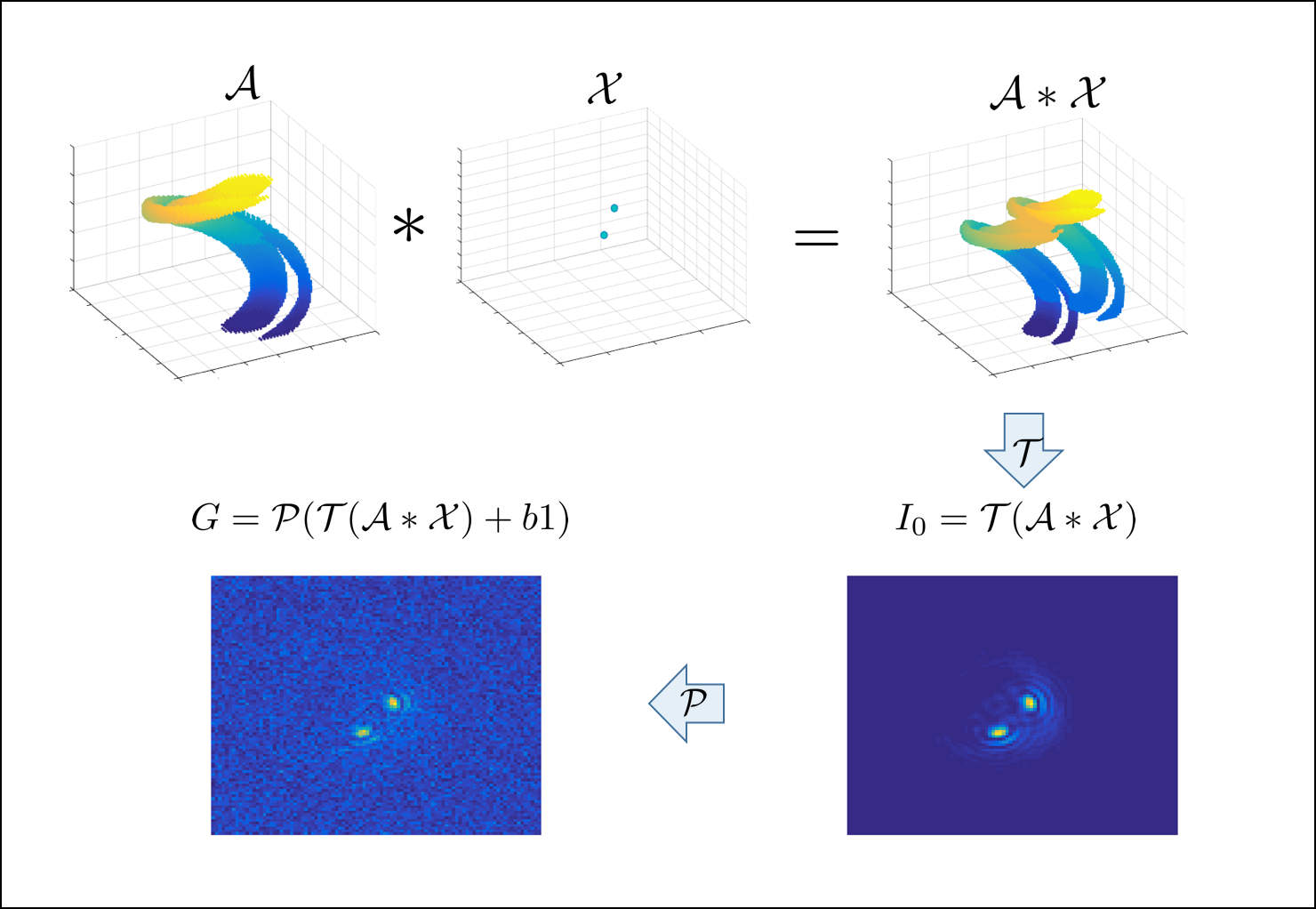}
\caption{Illustration of the discretized forward model for the single lobe PSF. }
\label{fig:illustration}
\end{figure}

\subsection{Outline of the paper}

The rest of the paper is organized as follows. In Section \ref{sec:Physics}, we describe the physics model for Prasad's single lobe rotating PSF. In Section \ref{sec:non-convex}, we propose a non-convex optimization model to solve the point source localization problem. In Section \ref{sec:alg}, our non-convex optimization algorithm is developed, with a post-processing step for eliminating clustered false positive point sources. In addition,  a new iterative scheme for estimating the flux values is also proposed in this section. Numerical experiments, including comparisons with other optimization models, are discussed in Section \ref{sec:numerical}.  Some concluding remarks are made in Section \ref{sec:Conclusions}.

\section{Physics model for single lobe rotating PSFs } \label{sec:Physics}
Here, we consider a new technique, recently patented by Prasad in \cite{kumar2013psf,prasad2013rotating}, for applying a rotating PSF with a single lobe to obtain depth from defocus. Comparing to the traditional double-helix rotating PSF \cite{DH2009pavani}, this single lobe PSF uses photons from targets more efficiently and yields images that are less cluttered and confounding in high-density swarm scenarios than other techniques.

The amount of rotation of Prasad's PSF encodes the depth position of the point source. Denote $A_{\zeta}$ as the PSF matrix for a point source with flux $f = 1$, transverse location $\mathbf{r}_0 =(x_0, y_0) =\mathbf{0}$, and defocus parameter $\zeta$. Note that $\zeta$ is proportional to
the object-distance $\delta z$ from the in-focus object plane, according to the relation
  \begin{equation}\label{zeta_delta_z}
	\zeta = -\frac{\pi \delta z R^2}{\lambda l_0(l_0+ \delta z)},
\end{equation}
where $l_0$ is the distance between the lens and {the plane of Gaussian focus} and $R$ is the radius of the pupil. Here the imaging wavelength is denoted by $\lambda$. The two-dimensional image data matrix $G\in \mathds{R}^{m\times n }$ with Poisson noise is {approximated as:}
\begin{equation}
	{G \approx \mathcal{P}\left(\sum_{i=1}^{M} \left( A_{\zeta_{i}}\ast \delta{(x_i,y_i)}\right)f_i  + b\right),}
	\label{equ:rpsf_total}
\end{equation}
where $b$ is the spatially uniform background, $\ast $ is the convolution operation, and $\delta{(x_i,y_i)}$ is the {transverse} location matrix for the $i$th source, with elements,
$$ \left(\delta{(x_i,y_i)}\right)_{uv} = \begin{cases}
	1, &(u,v) =(x_i,y_i),\\
	0, & {\rm otherwise}.
\end{cases}$$
The center of rotation of the image in the transverse plane, denoted by $\mathbf{r}_I = (x,y)$, is related to the source transverse location, $\mathbf{r}_0$, in the object plane {via linear magnification,}
$$\mathbf{r}_I = -\frac{z_I \mathbf{r}_0 }{ l_0+ \delta z }, $$
where $z_I$ is the distance between the image plane and the lens.

According to the Fourier optics model \cite{Goodman17}, the  incoherent PSF for a clear aperture containing a phase mask with optical phase 
retardation, $\psi(\mathbf{s})$, is given by 
{
\begin{equation}
	A_{\zeta}(\mathbf{s} ) = {1\over \pi}\left|\int P(\mathbf{u} )\mathrm{exp} \left[ \iota( 2\pi \mathbf{u}\cdot\mathbf{s} +  \zeta u^2 - \psi(\mathbf{u}))  \right] d \mathbf{u} \right|^2,
	\label{equ:A}
\end{equation}
where $\iota = \sqrt{-1} $, 
$P(\mathbf{u} )$ denotes the indicator function for the pupil of unit radius, and $\mathbf{s} = (s, \phi_{\mathbf{s}} ) $ is a scaled version of the image-plane position vector, $\mathbf{r}$, namely $\mathbf{s} =  \frac{\mathbf{r}}{\lambda z_I/R} $.
Here $\mathbf{r}$ is measured from the center of the geometric image point located at $\mathbf{r}_I$. The pupil-plane position vector $\boldsymbol{\rho}$ is normalized by the pupil radius, $\mathbf{u} = \frac{\boldsymbol{\rho}}{R}$. 
For the single-lobe rotating PSF, $\psi(\mathbf{u}) $ is chosen to be the spiral phase
 profile defined as $$\psi(\mathbf{u}) = l\phi_{\mathbf{u}}, \ \ \text{for } \sqrt{\frac{l-1}{L}}\leq u\leq \sqrt{\frac{l}{L}}, \ l = 1,\cdot\cdot\cdot, L, $$  in which $L$ is the number of concentric annular zones in the phase mask. We evaluate (\ref{equ:A}) by using the fast Fourier transform.}
\vspace{0.3cm}\\
{\bf Remark:} For the discretized forward model (see \Cref{fig:illustration}), we  get the dictionary $\mathcal{A}$ by sampling depths at regular intervals in the range, $\zeta_i\in [-\pi L, \ \pi L]$, over which the PSF performs one complete rotation about the geometric image center before it begins to break apart. The $i$-th slice of dictionary is denoted as $A_{\zeta_i}$.

\section{Non-convex optimization model }\label{sec:non-convex}
 In general, for inverse problems the objective function includes a data-fitting term and a regularization term. The regularization term, representing prior information, is based on the desired properties of the solution, while the data-fitting term measures the discrepancy between the estimated and the observed images according to the noise model. Here we consider the Poisson noise model for which the data fitting term is the $I$-divergence. The model with a regularization term, $\mu  {\mathcal R}(\mathcal{X})$ is
\begin{equation*}
	\min\limits_{\mathcal{X}\geq 0 }  D_{KL}(\mathcal{T}(\mathcal{A} \ast \mathcal{X})+b\,1, G) +\mu\mathcal{R}(\mathcal{X}),
	 \end{equation*}
	 where the uniform background is denoted as $b$. Here we also use it to denote the matrix $b\,1$, $1\in \mathds{R}^{m\times n}$, i.e., $b\,1$ is the ${m\times n}$ matrix with all entries $b$.
The data-fidelity term is the negative log-likelihood of the data. For Poisson noise, this term is the $I$-divergence which is also known as Kullback-Leibler (KL) divergence (see \cite{poisson_formula})
	   $$D_{KL}(z,g) = \langle g, \ln \frac{g}{z} \rangle + \langle 1, z-g \rangle. $$
	
This  is a large-scale sparse 3D inverse problem. For the regularization term, we consider a non-convex 
  function (see \cite{mila2010nonconvex,mila2013nonconvex,mila2008nonconvex,mila2015nonconvex}), using specifically
\begin{equation*}
	\mathcal{R}(\mathcal{X}):= \sum_{i,j,k = 1}^{m,n,d} \theta(\mathcal{X}_{ijk}) = \sum_{i,j,k = 1}^{m,n,d} \frac{|\mathcal{X}_{ijk}|}{a+|\mathcal{X}_{ijk}|},
\end{equation*}
where $a$ is fixed and
determines the degree of non-convexity, see \Cref{fig:curve_a_epsilon}(a).
Therefore, the minimization problem amounts to
	 \begin{equation}
	 	\min\limits_{\mathcal{X}\geq 0 }\left\{  \left\langle 1,  \mathcal{T}(\mathcal{A} \ast \mathcal{X})- G \ln(\mathcal{T}(\mathcal{A} \ast \mathcal{X})+b \,1) \right\rangle + \mu\sum_{i,j,k = 1}^{m,n,d} \frac{|\mathcal{X}_{ijk}|}{a+|\mathcal{X}_{ijk}|}\right\}. \label{equ:min_fun}
	 \end{equation}
Here, $\theta(t) = \lim\limits_{\epsilon \to 1} \theta_\epsilon(t), $ where $\theta_\epsilon(t) = \frac{|t|}{a+ \epsilon |t|}. $ Since $\theta_\epsilon(t)$ represents the {$\ell_1$-norm} when $\epsilon = 0$, we see the process of increasing non-convexity as $\epsilon$ increases from 0 to 1; see \Cref{fig:curve_a_epsilon}(b).
Moreover, we easily get the derivative of
$\theta (t)$:
\begin{equation*}
	\theta^\prime (t) = \frac{a \,\mathrm{sign}(t) }{(a+|t|)^2}, \quad \text{when } t\neq0, \text{and} \quad  \theta^\prime(0^+) = \frac{1}{a},
\end{equation*}
where $\theta^\prime(0^+)$ is the right
derivative at zero.


\begin{figure}[htbp]
\centering
\subfloat[$\theta(t;a)$]{\includegraphics[width=0.4\textwidth]{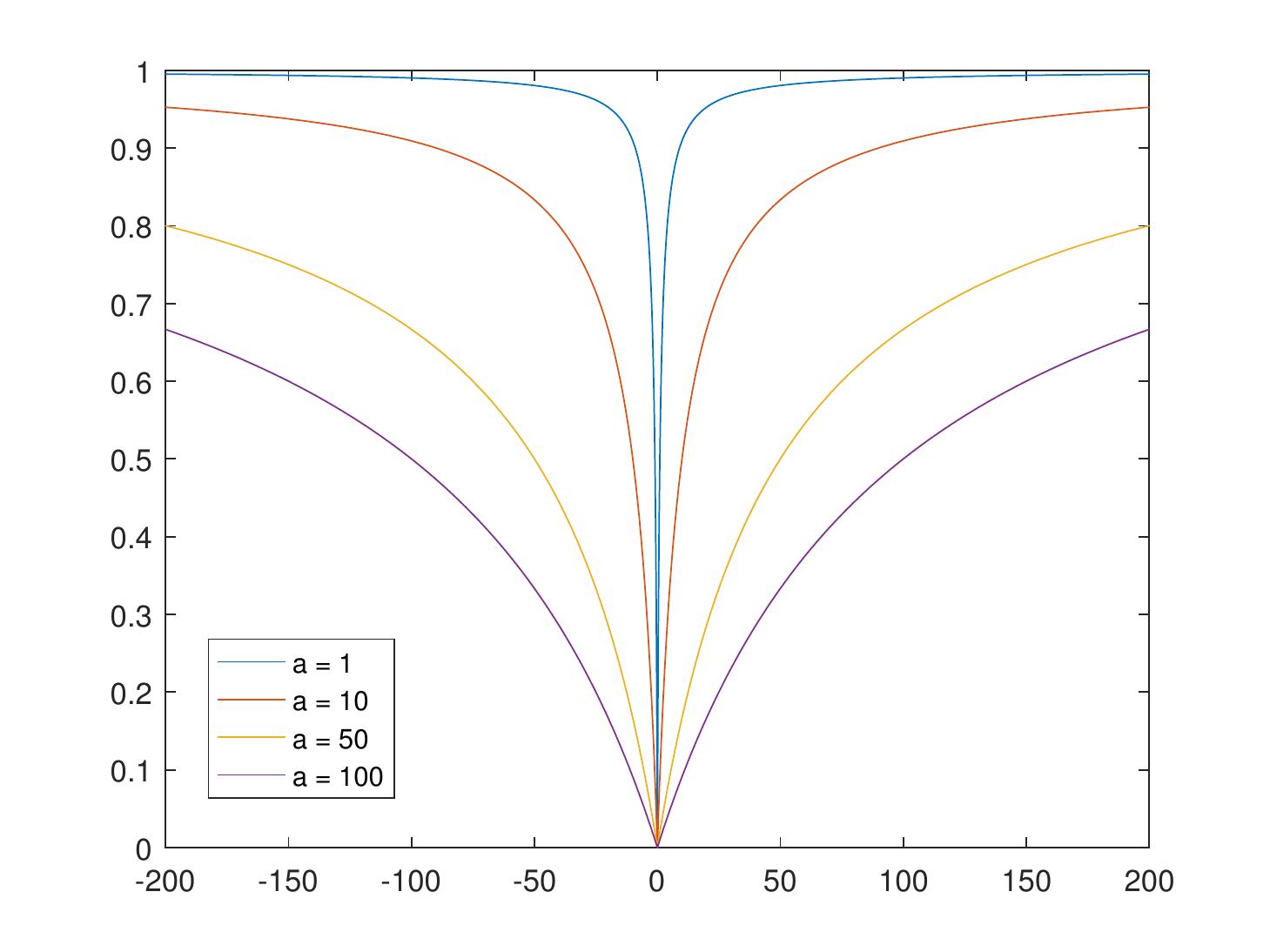}} \hspace{0.01mm}
\subfloat[$\theta_\epsilon(t;a)$]{\includegraphics[width=0.4\textwidth]{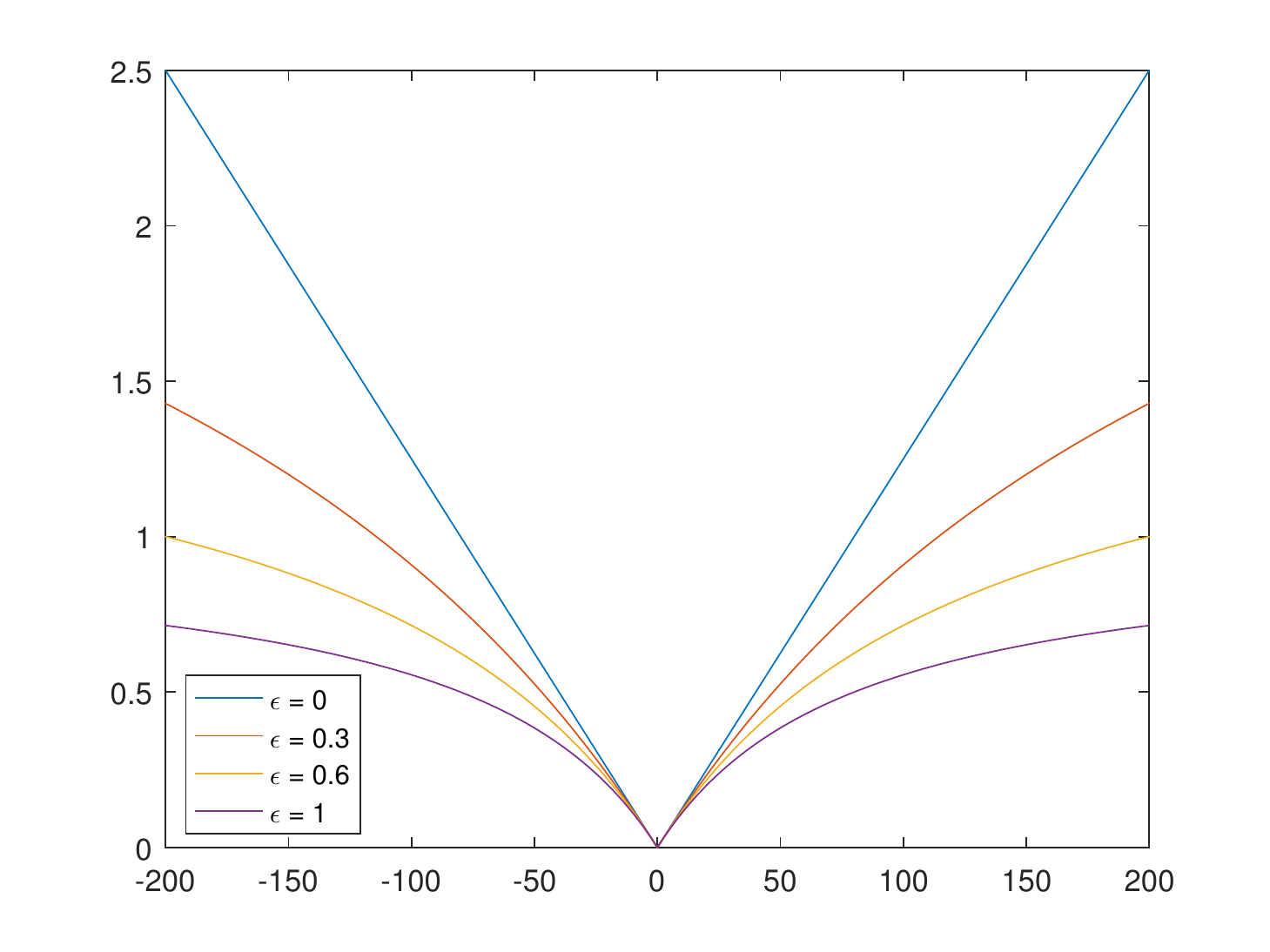}}\hspace{0.01mm}
\caption{Plot of  $\theta(t;a)$ and $\theta_\epsilon(t;a)$. In (a), we choose different values of $a$; in (b) we choose different values of $\epsilon$ while fixing $a = 80$. }
\label{fig:curve_a_epsilon}
\end{figure}

\section{Algorithm development}\label{sec:alg}
The minimization problem \eqref{equ:min_fun} involves non-convex 
 optimization. Here we use an iterative reweighted $\ell_1$ algorithm (IRL1) \cite{IRL1_2015} to solve this optimization problem. It is a majorization-minimization method which solves a series of convex optimization problem with a weighted-$\ell_1$ regularization term.
 It considers the problem (see Algorithm 3, in \cite{IRL1_2015})
\begin{equation*}
	\min_{x\in X} F(x):=F_1(x)+ F_2(G(x)),
\end{equation*}
 where $X$ is the constraint set. $F$ is a lower semicontinuous (lsc) function, extended, real-valued, proper, while $F_1$ is proper, lower-semicontinous, and convex and $F_2$ is coordinatewise nondecreasing, i.e. $F_2(x)\leq F_2(x+t e_i)$ with $x, x+t e_i\in G(X)$ and $t >0,$ where  $e_i$ is the $i$-th canonical basis unit vector. The function $F_2$ is concave on $G(X)$.   The IRL1 iterative scheme \cite[Algorithm 3]{IRL1_2015} is
 \begin{equation*}
 	\begin{cases}
 		w^l = \partial F_2(y), \ y = G(x^l), \\
 		x^{l+1} = \argmin\limits_{x\in X} \left\{F_1(x) + \langle w^l, G(x)\rangle \right\},
 	\end{cases}
 \end{equation*}
where $\partial$ stands for subdifferential.
For our problem \eqref{equ:min_fun}, we can choose:
\begin{equation*}
	\begin{split}
		F_1(\mathcal{X}) = & \ \langle 1, \  \mathcal{T}(\mathcal{A} \ast \mathcal{X})- G \log(\mathcal{T}(\mathcal{A} \ast \mathcal{X})+b \,1) \rangle;  \\
		F_2(\mathcal{X}) = & \  \mu \sum_{i,j,k = 1}^{m,n,d} \frac{\mathcal{X}_{ijk}}{a+\mathcal{X}_{ijk}}; \\
		G(\mathcal{X}) = & \  |\mathcal{X}|; \\
		X = & \ \{\mathcal{X} \ | \ \mathcal{X}_{ijk} \geq 0 \text{ for all } i,j,k  \}.
	\end{split}
\end{equation*}
Therefore, we compute the partial derivative of $w^l$ and get  $w^l_{ijk} =\frac{a \mu } { \left(a + \hat{\mathcal{X}}_{ijk}^l  \right)^2 }, \quad \forall i,j,k$. Here $w^l\neq0$ is finite, since $a, \mu \neq 0$, and all $\mathcal{X}_{ijk}\geq 0$ owning to the constraint $X$.
According to \cite{poisson_marcia,mila2008nonconvex}, these terms satisfy the requirements of the algorithm. \Cref{alg:outer} gives the IRL1 for solving \eqref{equ:min_fun}.
\begin{algorithm}[h]
\begin{algorithmic}[1]
\Require $\mathcal{X}^0\in \mathds{R}^{m\times n \times d}$ and $G\in \mathds{R}^{m \times n}$. Set $a$ and $ \mu$.
\Ensure The solution $\mathcal{X}^\ast$ which is the minimizer in the last outer iteration.
\Repeat
    \State Compute $w^l_{ijk} =\frac{a \mu } { \left(a + \hat{\mathcal{X}}_{ijk}^l  \right)^2 }, \quad \forall i,j,k$;
    \State Given $G$,  $w^l_{ijk}$,  obtain $\hat{\mathcal{X}}^l$ by solving
    \begin{equation}
        \hat{\mathcal{X}}^l = \argmin_{\mathcal{X}\geq 0  } \left\{ \left\langle 1, \  \mathcal{T}(\mathcal{A} \ast \mathcal{X})- G \log(\mathcal{T}(\mathcal{A} \ast \mathcal{X})+b \,1) \right\rangle + \sum_{i,j,k=1}^{m,n,d} w^l_{ijk}  |\mathcal{X}_{ijk} | \right\}.\label{equ:subproblem}
    \end{equation}
   \Until{convergence}
\end{algorithmic}
\caption{Iterative reweighted $\ell_1$ algorithm (IRL1) for the rotating PSF problem}
\label{alg:outer}
\end{algorithm}

\subsection{Subproblem of IRL1}
In \Cref{alg:outer}, we need to solve the subproblem \eqref{equ:subproblem}, which can be solved by the alternating direction method of multipliers (ADMM) \cite{ADMM2011boyd}. For this,
we introduce two auxiliary variables, namely $\mathcal{U}_0$ and $\mathcal{U}_1, $ and transform \eqref{equ:subproblem} into
\begin{equation*}
	\min\limits_{\mathcal{U}_1\geq 0, \ \mathcal{U}_0, \ \mathcal{X}} \left\{  \langle 1, \mathcal{T} \mathcal{U}_0 - G \log (\mathcal{T}\mathcal{U}_0 + b \,1) \rangle + \sum_{i,j,k=1}^{m,n,d} w^l_{ijk}  |(\mathcal{U}_1)_{ijk} | \ : \ \mathcal{U}_0 = \mathcal{A}\ast \mathcal{X}, \ \mathcal{U}_1 = \mathcal{X} \right\}.
\end{equation*}
The nonnegative constraint can be rewritten as an indicator function such that the objective function in this subproblem becomes
\begin{equation}
	\label{equ:aug}
 \langle 1, \mathcal{T} \mathcal{U}_0 - G \log (\mathcal{T}\mathcal{U}_0 + b \,1) \rangle + \sum_{i,j,k=1}^{m,n,d} w^l_{ijk}  |(\mathcal{U}_1)_{ijk} |
+ \mathbf{1}_{+}(\mathcal{U}_1),
\end{equation}
 where $\mathbf{1}_+(\mathcal{X})$ is the indicator function of the nonnegative constraint
    \begin{equation*}
    	\mathbf{1}_+(\mathcal{X})=
    	\begin{cases}
    		0, & \mathcal{X} \geq 0,\\
 \infty, & \text{otherwise}.     	
 \end{cases}
    \end{equation*}
Therefore, the augmented Lagrangrian function $\mathcal{L}(\mathcal{U}_0, \mathcal{U}_1, \mathcal{X}, \eta_0, \eta_1 )$  \cite{nocedal2006numerical} for  subproblem \eqref{equ:subproblem} is
\begin{equation*}
\begin{split}
	\mathcal{L}(\mathcal{U}_0, \mathcal{U}_1, \mathcal{X}, \eta_0, \eta_1 ):= & \langle 1, \mathcal{T} \mathcal{U}_0 - G \log (\mathcal{T}\mathcal{U}_0 + b \,1) \rangle + \sum_{i,j,k=1}^{m,n,d} w^l_{ijk}  |(\mathcal{U}_1)_{ijk} | \\
	& + \frac{\beta_0}{2}\|\mathcal{U}_0-\mathcal{A} \ast \mathcal{X} - \eta_0 \|_F^2 + \frac{\beta_1 }{2} \|\mathcal{U}_1-\mathcal{X}-\eta_1\|_F^2+ \mathbf{1}_+(\mathcal{U}_1 ),
	\end{split}
\end{equation*}
where $\eta_0, \eta_1 \in \mathds{R}^{m\times n \times d} $ are the Lagrange multipliers and $\beta_0, \beta_1 >0$. Here $\|\mathcal{X}\|_F$ is the Frobenius norm of $\mathcal{X}$, which is equal to the {$\ell_2$-norm} of the vectorized $\mathcal{X}$.   Starting at $\mathcal{X} =\mathcal{X}^t, \  \mathcal{U}_1 = \mathcal{U}_1^t$ and $\mathcal{U}_2 = \mathcal{U}_2^t, $ applying ADMM in \cite{ADMM2011boyd} yields the iterative scheme
\begin{subequations}
	\begin{align}
		\mathcal{U}_0^{t+1} &= \argmin\limits_{\mathcal{U}_0} \left\{\langle 1, \mathcal{T} \mathcal{U}_0 - G \log (\mathcal{T}\mathcal{U}_0 + b \,1) \rangle  + \frac{\beta_0 }{2} \|\mathcal{U}_0- \mathcal{A} \ast \mathcal{X}^{t} -\eta_0^{t} \|_F^2 \right\} \label{equ:U_0} \\
		\mathcal{U}_1^{t+1} &= \argmin\limits_{\mathcal{U}_1 \geq 0} \left\{  \sum_{i,j,k=1}^{m,n,d} w^l_{ijk}  |(\mathcal{U}_1)_{ijk} |  + \frac{\beta_1 }{2} \|\mathcal{U}_1-\mathcal{X}^t-\eta_1^t\|_F^2\right\} \label{equ:U_1} \\
		\mathcal{X}^{t+1} &= \argmin\limits_{\mathcal{X}} \left\{  \frac{\beta_0}{2}\left\|\mathcal{U}_0^{t+1}-\mathcal{A} \ast \mathcal{X} - \eta_0^{t} \right\|_F^2 + \frac{\beta_1 }{2} \left\|\mathcal{U}_1^{t+1}-\mathcal{X}-\eta_1^{t}\right\|_F^2\right\} \label{equ:X} \\
		\eta_0^{t+1} & = \eta_0^{t} - {\rho (\mathcal{U}_0^{t+1} -\mathcal{A} \ast \mathcal{X}^{t+1})} \\
		\eta_1^{t+1} & = \eta_1^{t} - {\rho( \mathcal{U}_1^{t+1} - \mathcal{X}^{t+1})}
	\end{align}
\end{subequations}
{Here $\rho\in \left(0, \frac{1+\sqrt{5}}{2}\right)$ is the dual steplength. We choose $\rho=1.618$ in our numerical tests.}
The $\mathcal{U}_1$-subproblem \eqref{equ:U_1} is solved by soft-thresholding under nonnegative constraint. So  the closed-form solution is given by 
\begin{equation*}
	\left(\mathcal{U}_1^{t+1}\right)_{ijk} = \max \left\{\mathcal{X}^t+\eta_1^t - w^l_{ijk}/\beta_1,\  0 \right\}.
\end{equation*}
The $\mathcal{X}$-subproblem \eqref{equ:X} is a least squares problem. We rewrite the convolution into  componentwise multiplication by using the Fourier transform and then the {objective function of} the subproblem becomes
$$ \frac{\beta_0}{2}\left\|\mathcal{F}\{\mathcal{A}\} \cdot \mathcal{F}\{\mathcal{X}\} - \mathcal{F}\left\{\mathcal{U}_0^{t+1}-\eta_0^{t}\right\} \right\|_F^2 + \frac{\beta_1 }{2} \left\|\mathcal{F}\{\mathcal{X}\}-\mathcal{F}\{\mathcal{U}_1^{t+1}-\eta_1^{t}\}\right\|_F^2. $$
Its closed-form solution reads as
\begin{equation}
	\mathcal{X}^{t+1} = \mathcal{F}^{-1}
	\left\{\left(\left|\mathcal{F}\{\mathcal{A}\right\}|^2+\frac{\beta_1}{\beta_0} \right)^{-1}\left(\overline{\mathcal{F}\{\mathcal{A}\}} \cdot \mathcal{F}\{\mathcal{U}_0^{t+1}-\eta_0^{t}\} +\frac{\beta_1}{\beta_0} \mathcal{F}\{\mathcal{U}_1^{t+1}-\eta_1^{t}\}\right)  \right\},
	\label{equ:fft}
	\end{equation}
	where $\left|\mathcal{X}\right|^2$ and $\overline{\mathcal{X}}$ are the componentwise operations of the square of absolute value of $\mathcal{X}$, and complex conjugate of $\mathcal{X}$, respectively.
	By assuming the boundary condition to be periodic, we use the 3D fast Fourier transformation (3D FFT) to compute \eqref{equ:fft} efficiently.
	
	 For the solution of $\mathcal{U}_0$-subproblem \eqref{equ:U_0}, we need \Cref{them:closed_form_soln}.  The range of the following proposition goes beyond the solution of \eqref{equ:U_0}.
It gives a closed-form solution to a highly nonlinear and useful functional,
namely the Thikhonov regularized KL divergence involving high-dimensional linear operators.
It can also be seen as the closed-form solution of the proximal operator of this general KL  divergence.
{
\begin{proposition}
\label{them:closed_form_soln}
Given $\beta, b\in \mathds{R}^1, \eta\in \mathds{R}^{m\times n \times d}$ and $G\in \mathds{R}^{m \times n}$,  consider the minimization problem
\begin{equation}
\min\limits_{\mathcal{U}} \left\{\langle 1, \mathcal{T} \mathcal{U} - G \log (\mathcal{T}\mathcal{U} + b\, 1) \rangle  + \frac{\beta }{2} \|\mathcal{U}- \xi \|_F^2 \right\}.
	\label{equ:thm}
\end{equation}
Then \eqref{equ:thm} has a closed-form solution
	\begin{equation}
    	\mathcal{U}^\ast_{ijk} = 
    	\begin{cases} 
    		\frac{-(1 + \beta b - \beta \xi_{ijk}) + \sqrt{(1 - \beta  b - \beta \xi_{ijk})^2+4 \beta  G_{ij}   }}{2 \beta  }, & \text{if } k =  d, \\
    		\xi_{ijk}, & \text{otherwise}.
    	\end{cases}\label{equ:u_closed}
    \end{equation}
\end{proposition}
\begin{proof}
Denote $\mathcal{J}(\mathcal{U}) := \langle 1, \mathcal{T} \mathcal{U} - G \log (\mathcal{T}\mathcal{U} + b \,1) \rangle  + \frac{\beta }{2} \|\mathcal{U}- \xi \|_F^2. $ The minimizer of $\mathcal{J}$ with respect to $\mathcal{U}$ satisfies $\nabla \mathcal{J}(\mathcal{U}^\ast) = 0$. 
Here  the minimization of $\mathcal{J}(\mathcal{U})$  can be considered as a sequence sub-minimization problems separately. 
 Then the minimizer can be achieved by considering partial derivative separately. 
In addition, $\langle 1, \mathcal{T}\mathcal{U}-G \log (\mathcal{T}\mathcal{U} + b \,1) \rangle$ can be rewritten as $\sum\limits_{i,j=1}^{m,n}\left(\mathcal{U}_{ijd}-G_{ij} \log (\mathcal{U}_{ijd} + b )\right) $ which does not involve $\mathcal{U}_{ijk}$ with $k \neq d$. Hence the partial derivative of $\langle 1, \mathcal{T}\mathcal{U}-G \log (\mathcal{T}\mathcal{U} + b \,1) \rangle$  with respect to the entries $\{(i,j,k), k\neq d \}$ are zero.   By $\nabla \frac{\beta}{2} \|\mathcal{U}-\xi\|^2_F = \beta(\mathcal{U}-\xi)$  we get
\begin{equation}
	\label{equ:U}
	\mathcal{U}^\ast_{ijk} = \xi_{ijk} \text{ if } k \neq d,
\end{equation}
For those entries in the last slice of $\mathcal{U}^\ast$, we rewrite the minimization problem in these entries as 
\begin{equation}
\min\limits_{\mathcal{U}_{ijd}} \left\{\langle 1,  \mathcal{U}_{ijd} - G_{ij} \log (\mathcal{U}_{ijd}  + b) \rangle  + \frac{\beta }{2} (\mathcal{U}_{ijd} - \xi_{ijd})^2 \right\}.
	\label{equ:min_u}
\end{equation}
Denoting $y = \mathcal{U}_{ijd}  + b$, then \eqref{equ:min_u} becomes a standard one-dimensional KL divergence problem $$\min\limits_y \left\{D_{KL}(y,G_{ij}) + \frac{\beta}{2}(y-\xi_{ijd}-b)\right\}. $$ 
By   \cite[Lemma 2.2]{Idivergence2013minimization}, the explicit form of the minimizer is  
	\begin{equation}
		y^\ast = \frac{-(1 - \beta b - \beta \xi_{ijk}) + \sqrt{(1 - \beta  b - \beta \xi_{ijk})^2+4 \beta  G_{ij}   }}{2 \beta  }. 
		\label{eqn:closed_form}
	\end{equation}
Combining \eqref{equ:U} and \eqref{eqn:closed_form}, we obtain \eqref{equ:u_closed}. 
\end{proof}
Based on \Cref{them:closed_form_soln},   we can obtain a closed-form solution for the $\mathcal{U}_0$-subproblem by setting $\xi = \mathcal{A} \ast \mathcal{X}^{t} +\eta_0^{t}$ and $\beta = \beta_0$
}

{
{\bf Remark:}
The convergence  of \eqref{equ:subproblem} is ensured by \cite[Theorem 4.2]{Idivergence2013minimization}, since the discrete forward model $\mathcal{T}(\mathcal{A}\ast \mathcal{X})$ can be written as matrix-vector multiplication. And,  we can regarded $L$ in \cite[Theorem 4.2]{Idivergence2013minimization} as an identify matrix. }
\subsection{Removing false positives using a centroid method}\label{subsec;centroid}
For real data, point sources may be not on the grid, which means the discrete model is not accurate. In order to avoid missing sources, the regularization parameter $\mu$ is kept small, which leads to the over-fitting effect. That is to say that there are more point sources in our estimated result than the ground truth,  which we must attempt to mitigate.
Our optimization solution generally contains tightly clustered point sources, so we need to regard any such cluster of point sources as a single point source.
{The same phenomenon has been observed in \cite{FALCON2014,Rice2016generalized,clustered2012fasterstorm}.
To this end, we apply a post-processing approach following
\cite{clustered2012fasterstorm}. } The method is based on a well-defined tolerance distance for recognizing clustered neighbors denoted as $C$. The criterion is to find nonzero entries that are within a certain distance from the given point whose pixel value is higher than or equal to other nonzero points in the neighborhood.
For recognizing the cluster of point sources, we need to start from the entry with the highest intensity. By the above criterion, one clustered neighborhood is recognized.
Then we compute the centroid of each cluster. That is to say, to get the centroid location $(x,y,z)$ for one cluster $C$ as

\begin{equation*}
		x = \frac{\sum_{(i,j,k)\in C}i\mathcal{X}_{ijk}}{\sum_{(i,j,k)\in C}\mathcal{X}_{ijk}}; \quad  \
		y =  \frac{\sum_{(i,j,k)\in C}j\mathcal{X}_{ijk}}{\sum_{(i,j,k)\in C}\mathcal{X}_{ijk}}; \quad  \
		z =  \frac{\sum_{(i,j,k)\in C}k\mathcal{X}_{ijk}}{\sum_{(i,j,k)\in C}\mathcal{X}_{ijk}}.
\end{equation*}
The flux for this representative point source in cluster $C$ is $\sum_{(i,j,k)\in C}\mathcal{X}_{ijk}$.
We set the value of these entries in the recognized neighborhood as 0 and then we proceed with the searching  process recursively.  An estimated point source that cannot find out its corresponding ground truth is called false positive.
 To some false positives which may be due to the effect of the periodic boundary condition,  we set a threshold, say 5\% of highest intensity \cite{Rice2016generalized}. Those entries whose pixel value is lower than the threshold will be regarded as false positives so that we directly delete them in the searching  process.  The whole post-processing method is summarized in \Cref{alg:post-procesing}.
\begin{algorithm}[h]
\caption{Removing cluster point sources by computing centroids (Centroid FP). }
\label{alg:post-procesing}
	\begin{algorithmic}[1]
\State Set $\mathcal{U}$ as a zero 3D tensor of the same size as $\mathcal{X}$;
	\State Find the maximum value of the solution $\mathcal{X}$;
	\State Recognize the neighbour $C$ based on the above criterion and recursively check the other nonzero points;
	\State Compute the centroid of $C$ and set the value of this centroid entry in $\mathcal{U}$ as the summation of all the  pixel values in $C$;
	\State Set $\mathcal{X}(C) = 0$;
\State If there is any nonzero entry in $\mathcal{X}$, go to Step 2, otherwise, go to the  next step;
\State Set pixel value of these entries in $\mathcal{U}$ whose value is lower than $5\%$ of the  highest intensity as 0.
	\end{algorithmic}
\end{algorithm}

\subsection{Estimating the flux values}
In \Cref{sec:alg}, we provided an algorithm for estimating the locations and fluxes for the point sources. However, numerical results show that the flux values are generally underestimated. In \cite{FALCON2014,Rice2016generalized},  least squares fitting is used for improving the resolution as well as updating the corresponding fluxes. However, our problem is not Gaussian noise or additive Poisson noise as used in their paper. Our Poisson noise is data-dependent, which cannot force the regenerated image $\mathcal{T}(\mathcal{A}\ast \mathcal{X^\ast})$ to match the observed data with least squares, where $\mathcal{X}^\ast$ is the result after solving the non-convex optimization problem and post-processing.
Our aim is thus to estimate the source fluxes from the KL data-fitting term appropriate to the Poisson noise model
when the source 3D positions have already been accurately estimated.

Let the PSF corresponding to the $i$-th source
be arranged as the column vector $\mathbf{h}_{i}$. The stacking of the $M$ column vectors in the same sequence as
the source labels for the $M$ sources then defines a system PSF matrix $H$, with $H= \left[\mathbf{h}_1, \mathbf{h}_2, \cdots, \mathbf{h}_M\right]\in \mathds{R}^{K\times M}, $ where $K$ is the total number of pixels in the vectorized data array, so $K=mn. $
The vectorized observed image is denoted by  $\mathbf{g} \in \mathds{R}^{ K \times 1}$.  The uniform background is denoted as  the vector $b\mathbf{1}$ with  $\mathbf{1}\in \mathds{R}^{K\times 1}$.
The flux vector is denoted as $\mathbf{f}\in \mathds{R}^{M\times 1}$. Here the problem is overdetermined meaning that the number of point sources $M$ is much smaller than the number of available data $K = mn$.   Therefore we need to do some refinement of the estimates by minimizing directly data fitting term.
 Since the negative log-likelihood function for the Poisson model, up to certain data dependent terms,
is simply the KL divergence function,
\begin{equation*}
	D_{KL}(H \mathbf{f} +b \mathbf{1}, \mathbf{g}) = \left\langle \mathbf{1}, H \mathbf{f} - \mathbf{g} \log(H \mathbf{f} + b \mathbf{1})  \right\rangle,
\end{equation*}
its minimization with respect to  the flux vector $\mathbf{f}$, performed
by setting the gradient of $D_{KL}$  with respect  to $\mathbf{f}$ (see \cite{poisson_marcia}) zero, yields the nonlinear relation
\begin{equation}
\label{e2}
\begin{split}
\nabla D_{KL}(H \mathbf{f} +b \mathbf{1}, \mathbf{g}) = & H^T \mathbf{1} - \sum\limits_{i = 1}^K \frac{\mathbf{g}_i}{\mathbf{e}_i^T (H \mathbf{f} + b \mathbf{1})} H^T \mathbf{e}_i \\
= & \sum\limits_{i = 1}^K \frac{\mathbf{e}_i^T\left( H \mathbf{f} + b \mathbf{1} -\mathbf{g}\right)}{\mathbf{e}_i^T (H \mathbf{f} + b \mathbf{1})} H^T \mathbf{e}_i = 0,
\end{split}
\end{equation}
{where $\mathbf{e}_i$ is the $i$-th canonical basis unit vector. }




Consider now an iterative solution of \eqref{e2}. Since $b$ is proportional to the vector of ones, $\mathbf{1}$,
we write \eqref{e2} as

\begin{equation*}
	\begin{split}
		0 = & \ b D_{KL}(H \mathbf{f} +b \mathbf{1}, \mathbf{g})  \\
		= & \ \sum\limits_{i = 1}^K \frac{\mathbf{e}_i^T\left( H \mathbf{f} + b\mathbf{1} -\mathbf{g} \right)\mathbf{e}_i^T (b \mathbf{1})}{\mathbf{e}_i^T (H \mathbf{f} + b \mathbf{1})} H^T \mathbf{e}_i \\
		= & \ \sum\limits_{i = 1}^K \frac{\mathbf{e}_i^T\left( H \mathbf{f} + b \mathbf{1} -\mathbf{g} \right)\mathbf{e}_i^T\left( H \mathbf{f} + b \mathbf{1} -H \mathbf{f} \right)}{\mathbf{e}_i^T (H \mathbf{f} + b \mathbf{1})} H^T \mathbf{e}_i \\
		= & \ \sum\limits_{i=1}^K  \mathbf{e}_i^T\left( H \mathbf{f} + b \mathbf{1}-\mathbf{g} \right) H^T \mathbf{e}_i  -  \sum\limits_{i = 1}^K\frac{\mathbf{e}_i^T\left( H \mathbf{f} + b \mathbf{1}-\mathbf{g} \right)\mathbf{e}_i^TH \mathbf{f}}{\mathbf{e}_i^T( H \mathbf{f} + b \mathbf{1})} H^T \mathbf{e}_i \\
		= & \ H^T H \mathbf{f} + H^T \left(b \mathbf{1} -\mathbf{g} \right)  - \sum\limits_{i = 1}^K\frac{\mathbf{e}_i^T\left( H \mathbf{f} + b \mathbf{1}-\mathbf{g} \right)\mathbf{e}_i^TH \mathbf{f}}{\mathbf{e}_i^T (H \mathbf{f} + b \mathbf{1})} H^T \mathbf{e}_i.
	\end{split}
\end{equation*}
In this form, we easily get
\begin{equation*}
	H^T H \mathbf{f} =  H^T \left( \mathbf{g}-b \mathbf{1} \right)  + \sum\limits_{i = 1}^K\frac{\mathbf{e}_i^T\left( H \mathbf{f} + b \mathbf{1}-\mathbf{g} \right)\mathbf{e}_i^TH \mathbf{f}}{\mathbf{e}_i^T (H \mathbf{f} + b \mathbf{1})} H^T \mathbf{e}_i,
\end{equation*}

By multiplying by $(H^T H)^{-1}$ the two sides of this equation,  we directly get
\begin{equation}
	\label{equ:formula_f}
	\mathbf{f} = \mathbf{f}_G +  \sum\limits_{i = 1}^K\frac{\mathbf{e}_i^T\left( H \mathbf{f} + b \mathbf{1}-\mathbf{g} \right)\mathbf{e}_i^TH \mathbf{f}}{\mathbf{e}_i^T (H \mathbf{f} + b \mathbf{1})} H^+ \mathbf{e}_i,
\end{equation}
where $H^{+} = (H^T H)^{-1}H^T$ and $\mathbf{f}_G = H^{+}(\mathbf{g}-b \mathbf{1} )$ is the solution corresponding to the Gaussian noise model.
This  suggests the following iterative algorithm:

\begin{equation}
	\mathbf{f}^{n+1} = \mathbf{f}_G + \mathcal{K}(\mathbf{f}^{n}),  \quad n = 1, 2, \cdots
	\label{iterative_scheme_flux}
\end{equation}
where $$\mathcal{K}(\mathbf{f}) = \sum\limits_{i = 1}^K\frac{\mathbf{e}_i^T\left( H \mathbf{f} + b \, 1-\mathbf{g} \right)\mathbf{e}_i^TH \mathbf{f}}{\mathbf{e}_i^T (H \mathbf{f} + b \mathbf{1})} H^+ \mathbf{e}_i. $$
	To emphasize that our non-convex optimization model is based on the use of a KL data fitting (KL) term and a non-convex (NC) regularization term, we designate our approach as KL-NC that has been enhanced with post-processing and a refined estimation of flux. We summarize the different steps of our proposed method for the rotating PSF problem in \Cref{alg:whole_process}.

\begin{algorithm}[t]
\begin{algorithmic}[1]
\Require $\mathcal{X}^0\in \mathds{R}^{m\times n \times d}$ and $G\in  \mathds{R}^{m \times n}$.
\Ensure The locations of point sources with the corresponding flux  $\left\{(x_i^\ast, y_i^\ast, z_i^\ast, \mathbf{f}_i^\ast)\right\}$.
    \State Solve the non-convex optimization problem by \Cref{alg:outer} and get the minimizer $\mathcal{X}^*$;
    \State Do the post-processing by \Cref{alg:post-procesing} to get the locations of estimated point sources   $\left\{(x_i^\ast, y_i^\ast, z_i^\ast)\right\}$;
    \State Estimate the fluxes by the iterative scheme \eqref{iterative_scheme_flux} and get $\mathbf{f}^\ast$.
\end{algorithmic}
\caption{KL-NC with post-processing and estimation of the flux. }
\label{alg:whole_process}
\end{algorithm}

\section{Numerical experiments}\label{sec:numerical}

In this section we apply our optimization  approach to solving  simulated rotating PSF problems for point source localization and compare it to some other optimization methods. The codes of our algorithm and the others with which we compared our method were written in $\mathrm{MATLAB \ 9.0 \ (R2016a)}, $ and all the numerical experiments were conducted on a typical personal computer with a standard CPU (Intel i7-6700, 3.4GHz).

The fidelity of localization is assessed in terms of the {\bf recall rate}, defined as { \it  the ratio of the number of identified true positive point sources over the number of true positive point sources,} and the {\bf precision rate}, defined as {\it the ratio of the number of identified true positive point sources over the number of all point sources obtained by the algorithm}; see \cite{Book_SR_micro2017}.

To distinguish  true positives from false positives from the estimated point sources, we need to determine the minimum total distance between the estimated point sources and true point sources. Here all 2D simulated observed images are described by 96-by-96 matrices. We set the number of zones of the spiral phase mask responsible for the rotating PSF at $L=7$  and the aperture-plane side length as 4 which sets the pixel resolution in the 2D image (FFT) plane as 1/4 in units of $\lambda z_I/R$.
The dictionary corresponding to our discretized 3D space contains 21 slices in the axial direction, with the corresponding values of the defocus parameter, $\zeta$, distributed uniformly over the range, $[-21, \  21]$.
According to the Abbe-Rayleigh resolution criterion,  two point sources that are within $(1/2)\lambda z_I/R$ of each other and lying in the same transverse plane cannot be separated in the limit of low intensities.
In view of this criterion and our choice of the aperture-plane side length and if we assume conservatively that our algorithm does not yield any significant superresolution, we must regard two point sources that  are within 2 image pixel units of each other as a single point source.
Analogously, two sources along the same line of sight ({\it i.e.,} with the same $x,y$ coordinates) that are axially separated from each other within a single unit of $\zeta$  must also be regarded as a single point source.

As for real problems, our simulation does not assume that the point sources are on the grid points. Rather, a number of point sources are randomly generated in a 3D continuous image space with certain fluxes. We consider a variety of source densities, from 5 point sources to 40 point sources in the same size space.  For each density, we randomly generate 20 observed images and use them for training the parameters in our algorithm, and then test 50 simulated images with the well-selected parameters. The number of photons emitted by each point source follows a Poisson distribution with mean of 2000 photons. Instead of adding Poisson noise as additive noise as in \cite{Rice2016generalized}, we apply  data-dependent Poisson noise by using the MATLAB command
  \begin{equation*}
  	\verb|G = poissrnd(I0+b)|,
  \end{equation*}
where $\verb|I0|$ is the 2D original image formed by  adding all the images of the point sources, and  $\verb|b|$ is the background noise which we set to a typical value 5. Here, \verb|poissrnd| is the $\mathrm{MATLAB}$ command whose input is the mean of the Poisson distribution.

\subsection{3D localizations for low and high density cases}
In this subsection, we will show some 3D localizations results in low and high density cases (15 point sources and 30 point sources); see \Cref{fig:15- 30}.

\begin{figure}[htbp]
\centering
\subfloat[Observed image]{\includegraphics[width=0.46\textwidth]{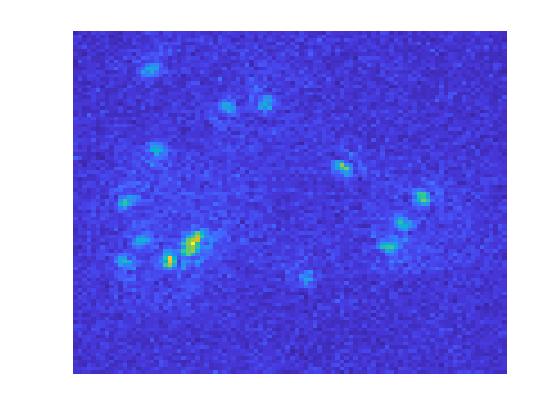}}
\subfloat[Observed image]{\includegraphics[width=0.46\textwidth]{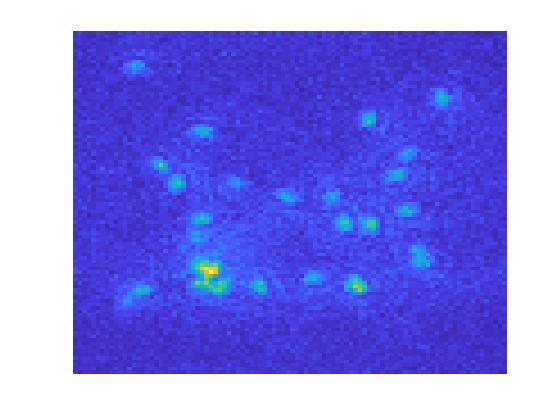}}\\
\subfloat[Estimated locations in 2D]{\includegraphics[width=0.46\textwidth]{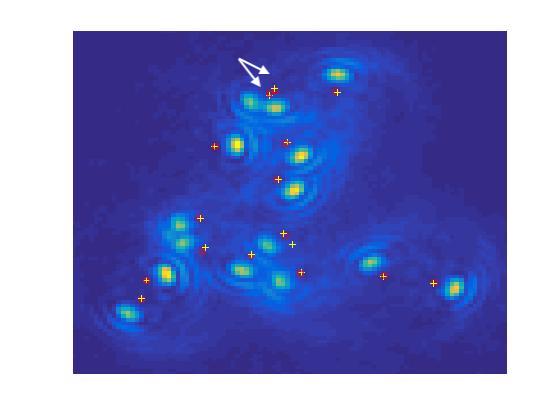}}
\subfloat[Estimated locations in 2D]{\includegraphics[width=0.46\textwidth]{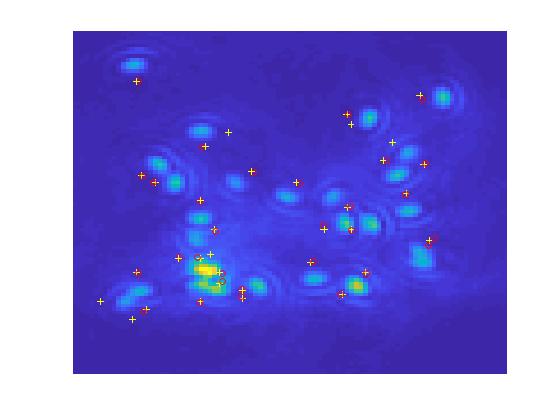}} \\
\subfloat[Estimated locations in 3D]{\includegraphics[width=0.46\textwidth]{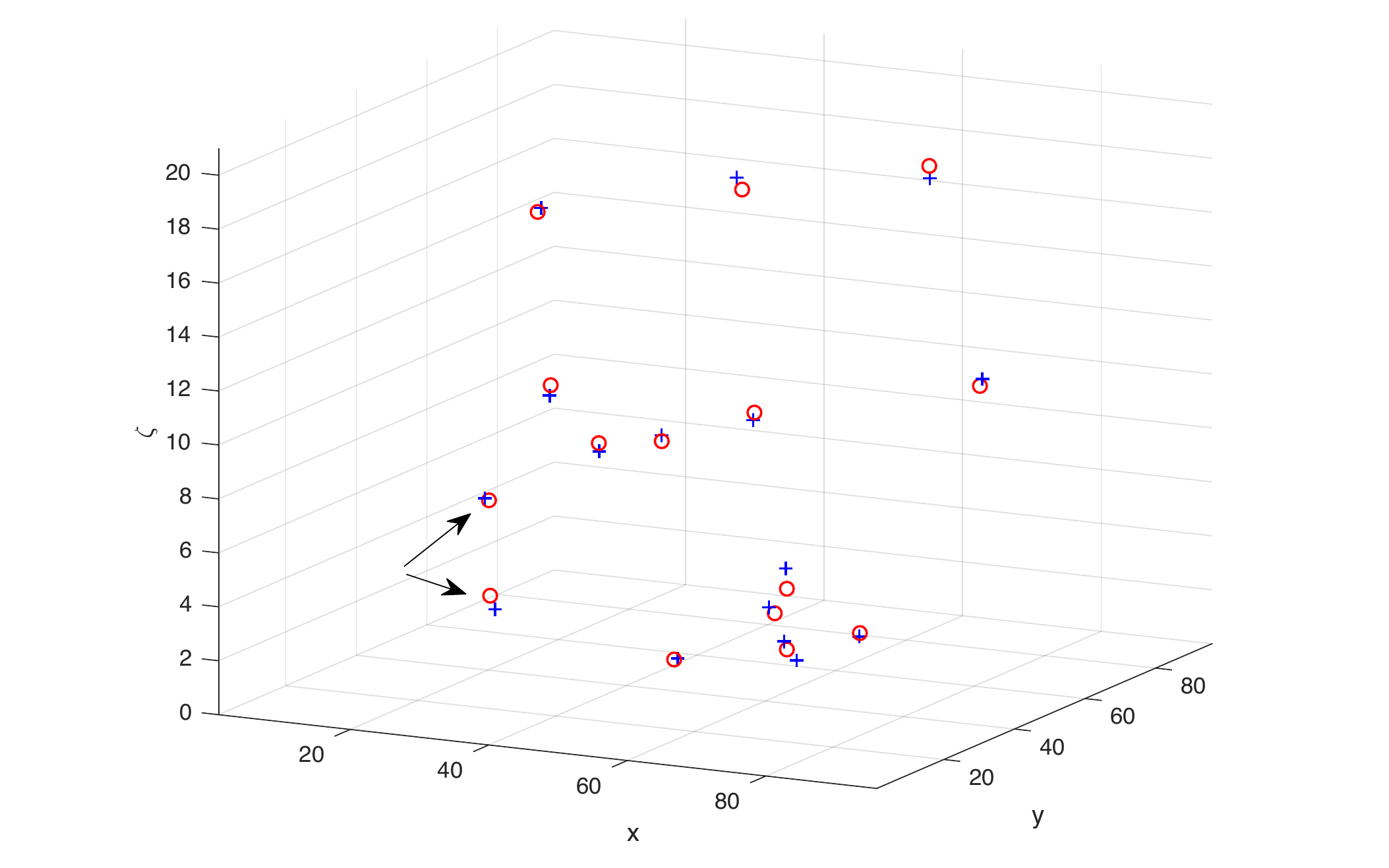}}
\subfloat[Estimated locations in 3D]{\includegraphics[width=0.46\textwidth]{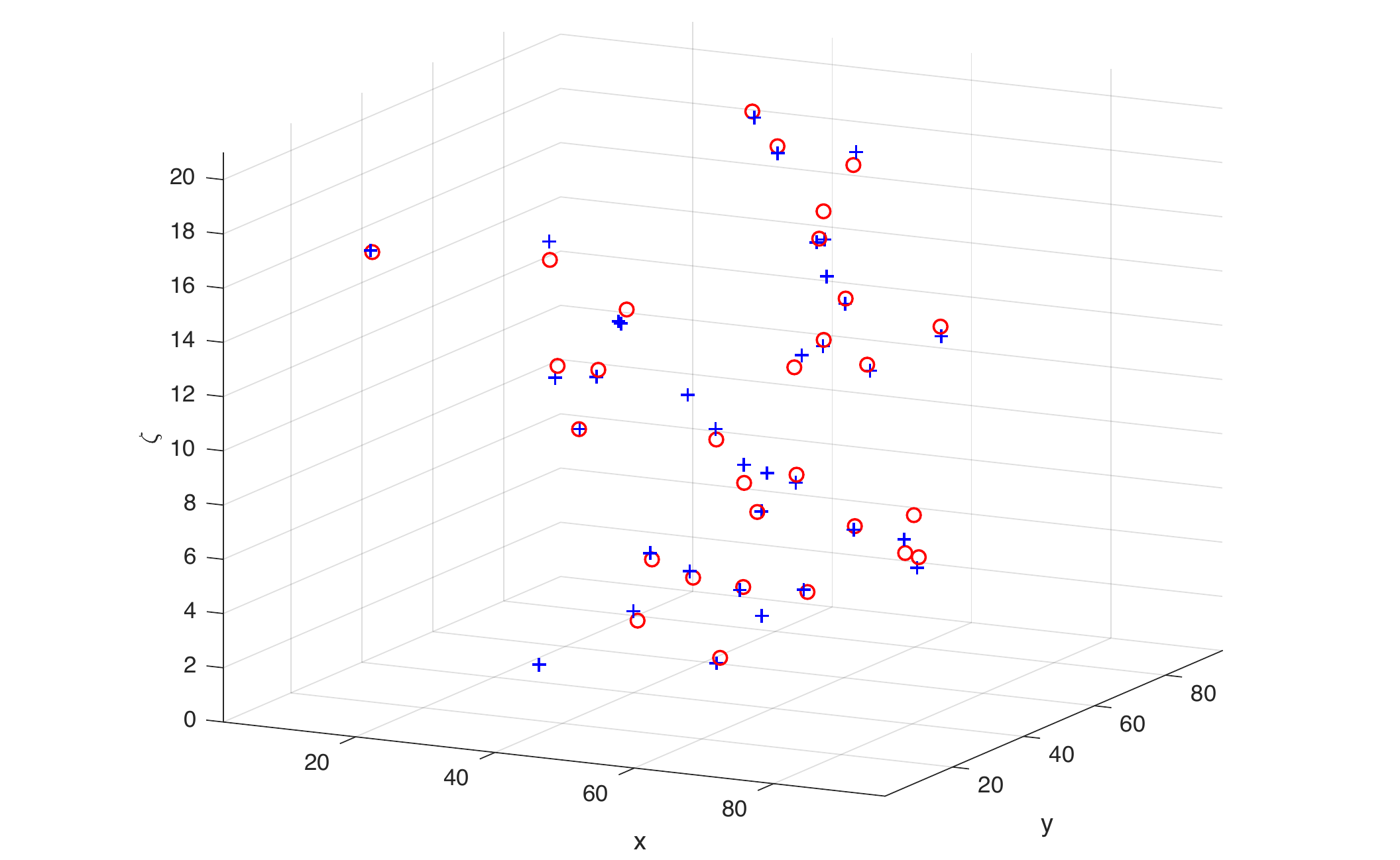}}
\caption{3D localizations for low and high density cases: (a), (c) and  (e) are the results for the 15 point sources case; (b), (d) and  (f) are the results for the 30 point sources case. ``o'' is the ground truth point source and ``+'' is the estimated point source. }\label{fig:15- 30}
\end{figure}

Note that our 3D localization is estimated very well for 15 point sources. In this example, all the true positives are identified and there is only one false positive. From \Cref{fig:15- 30}(c) and (e), we see that there are two true point sources that are close in their transverse coordinates,
but well-separated in the axial direction; {see the arrows in \Cref{fig:15- 30}(c),(e)}. In the 2D original image (without noise) in \Cref{fig:15- 30}(c), {the rotating PSF images of those two point sources are non-overlapping} and our approach can well identify both two point sources.

\Cref{fig:15- 30}(b), (d) and (f) show the high density (30 point sources) case with {many overlapping rotating PSF images of the respective point sources}. Such image overlap in the presence of Poisson noise makes the problem difficult. The number of point sources is not easily obtained by observation. In this specific case, our algorithm still identifies all the true point sources correctly, but produces 6 false positives. From \Cref{fig:15- 30}(d), we can see that these false positives come from {rather substantial} PSF overlapping.

Recall that the solution of our non-convex optimization problem involves outer and inner iterations. {The stopping criteria for inner iteration is set as the relative error $\frac{\|\mathcal{U}_1^{t+1}-\mathcal{U}_1^{t}\|_F}{\|\mathcal{U}_1^{t}\|_F}$ less than a tolerant small number. We also set the  maximum number of outer iterations as 2 and maximum number of inner iterations as 400.} From \Cref{fig:out_iteration,fig:inner_iteration}, we can see how and why we chose these two maximum iteration numbers. Let us denote the number of maximum outer iterations and inner iterations as $\mathrm{Max_{out}}$ and $\mathrm{Max_{in}}$ respectively. In \Cref{fig:out_iteration}, we fixed $\mathrm{Max_{in}}$ as 400 and tried different $\mathrm{Max_{out}}$. As we see, the recall rate, whether with or without post-processing, changes little. The precision rate increases, however, when we change $\mathrm{Max_{out}}$ from 1 to 2, but it changes considerably less when changing $\mathrm{Max_{out}}$ from 2 to 3, while the computational time increases linearly with $\mathrm{Max_{out}}$. In \Cref{fig:inner_iteration}, on the other hand, we fixed $\mathrm{Max_{out}}$ as 2 and changed $\mathrm{Max_{in}}$. Here $\mathrm{Max_{in}}=400$  gives the best results with a relatively low time cost.  In both cases, the precision rate can be improved significantly with post-processing.

\begin{figure}[htp]
\centering
	\resizebox{0.8\textwidth}{0.4\textheight}{\includegraphics{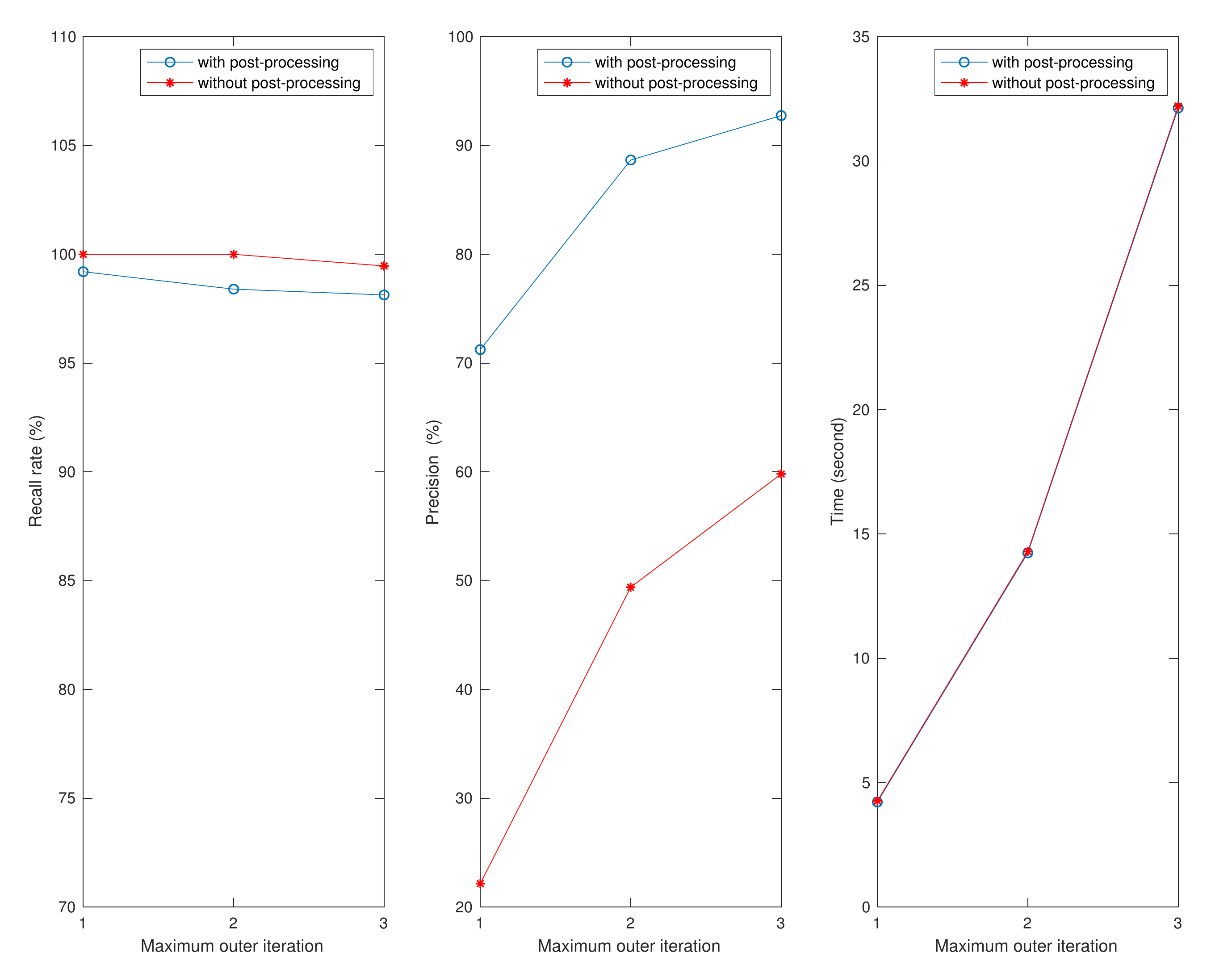}}
\caption{Effects of the maximum number of outer iteration.  }\label{fig:out_iteration}
\end{figure}
\begin{figure}[htp]
\centering
	\resizebox{0.8\textwidth}{0.4\textheight}{\includegraphics{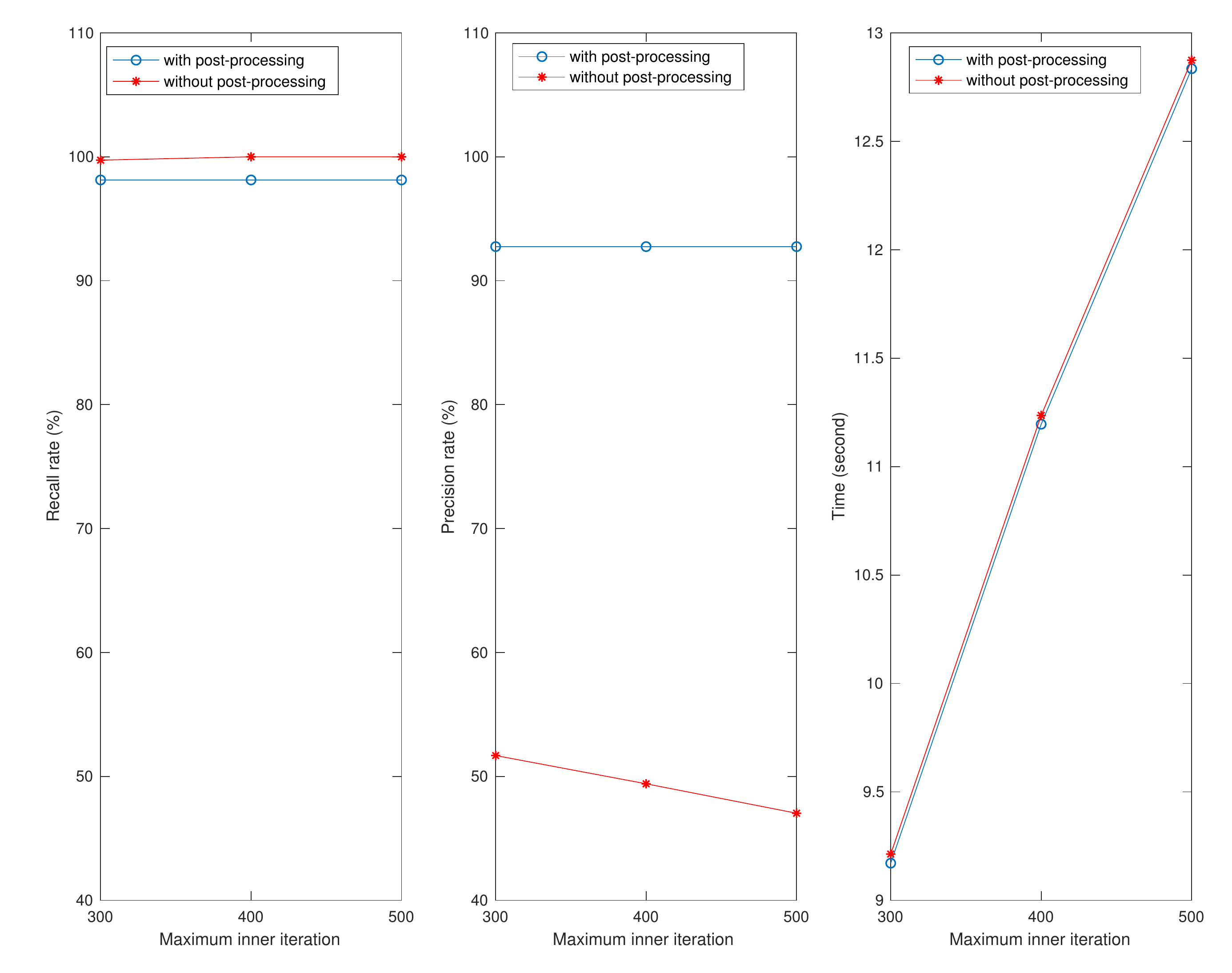}}
\caption{Effects of the maximum number of inner iteration.  }\label{fig:inner_iteration}
\end{figure}

Here, we consider the practical convergence of our algorithm.
  That is to check
  the differences between the auxiliary variables ($\mathcal{U}_0$ and $\mathcal{U}_1$) and their corresponding substituted variables ($\mathcal{A}\ast \mathcal{X}$ and $\mathcal{X}$), i.e., the value of $\|\mathcal{U}_0-\mathcal{A}\ast \mathcal{X}\|_F$ and
$\|\mathcal{U}_1-\mathcal{X}\|_F$  in each iteration. Noted that the maximum number of inner and outer iterations are 400 and 2 respectively.  We plot the values of these two terms for each iteration in \Cref{fig:converage_u0,fig:converage_u1}.  {Both terms decrease  iteratively. Note that the values of $\|\mathcal{U}_0\|_F$ and $\|\mathcal{A}\ast \mathcal{X}\|_F$ are about 2300 and the values of $\|\mathcal{U}_1\|_F$ and $\|\mathcal{X}\|_F$ are about 260.} Therefore, the differences between auxiliary variables and their substituted variables are small, which shows the convergence numerically.

\begin{figure}[htbp]
\centering
\subfloat[In the 1st outer iteration]{\includegraphics[width=0.49\textwidth, height=0.28\textheight]{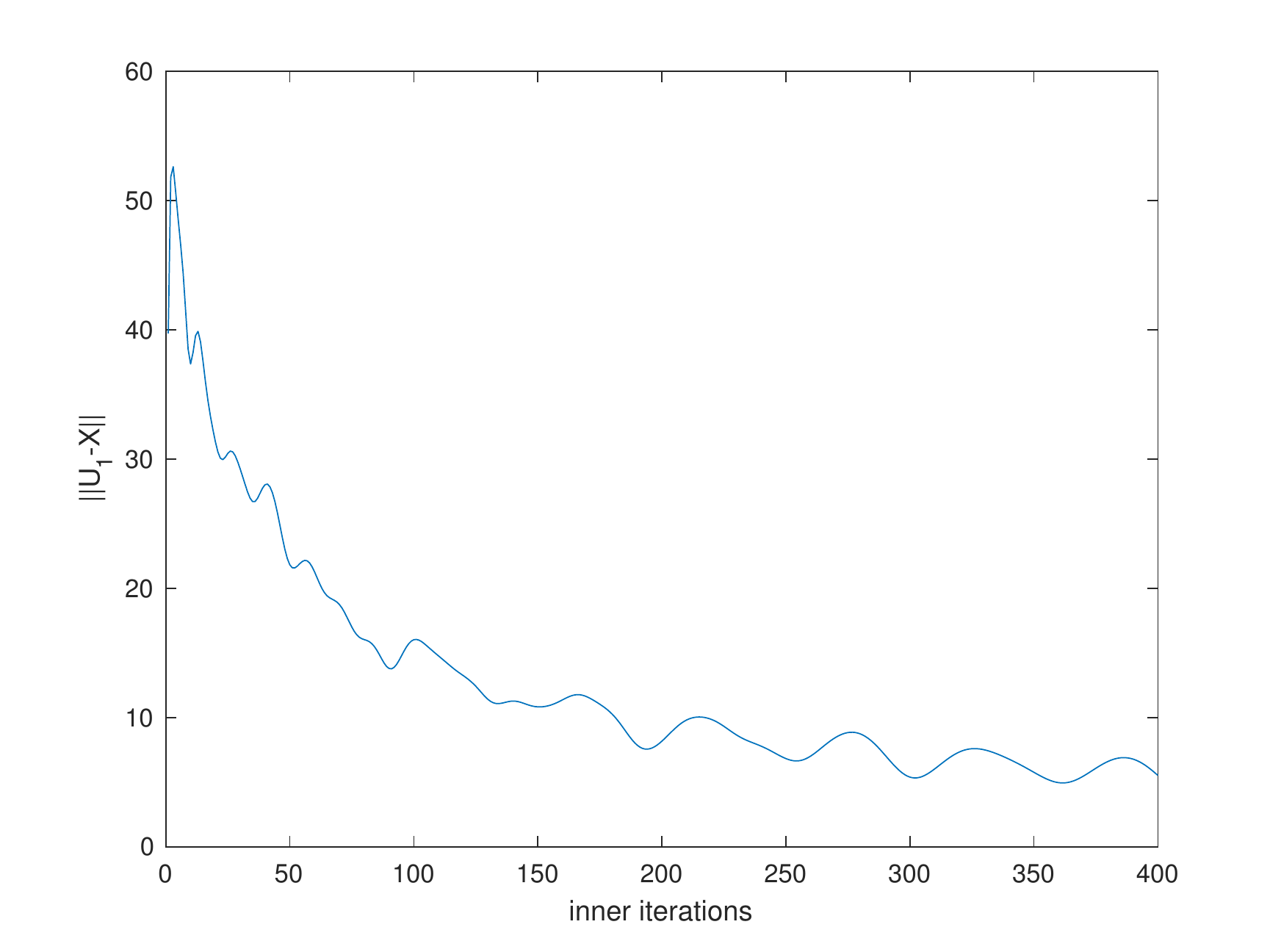}}
\subfloat[In the 2nd outer iteration]{\includegraphics[width=0.49\textwidth, height=0.28\textheight]{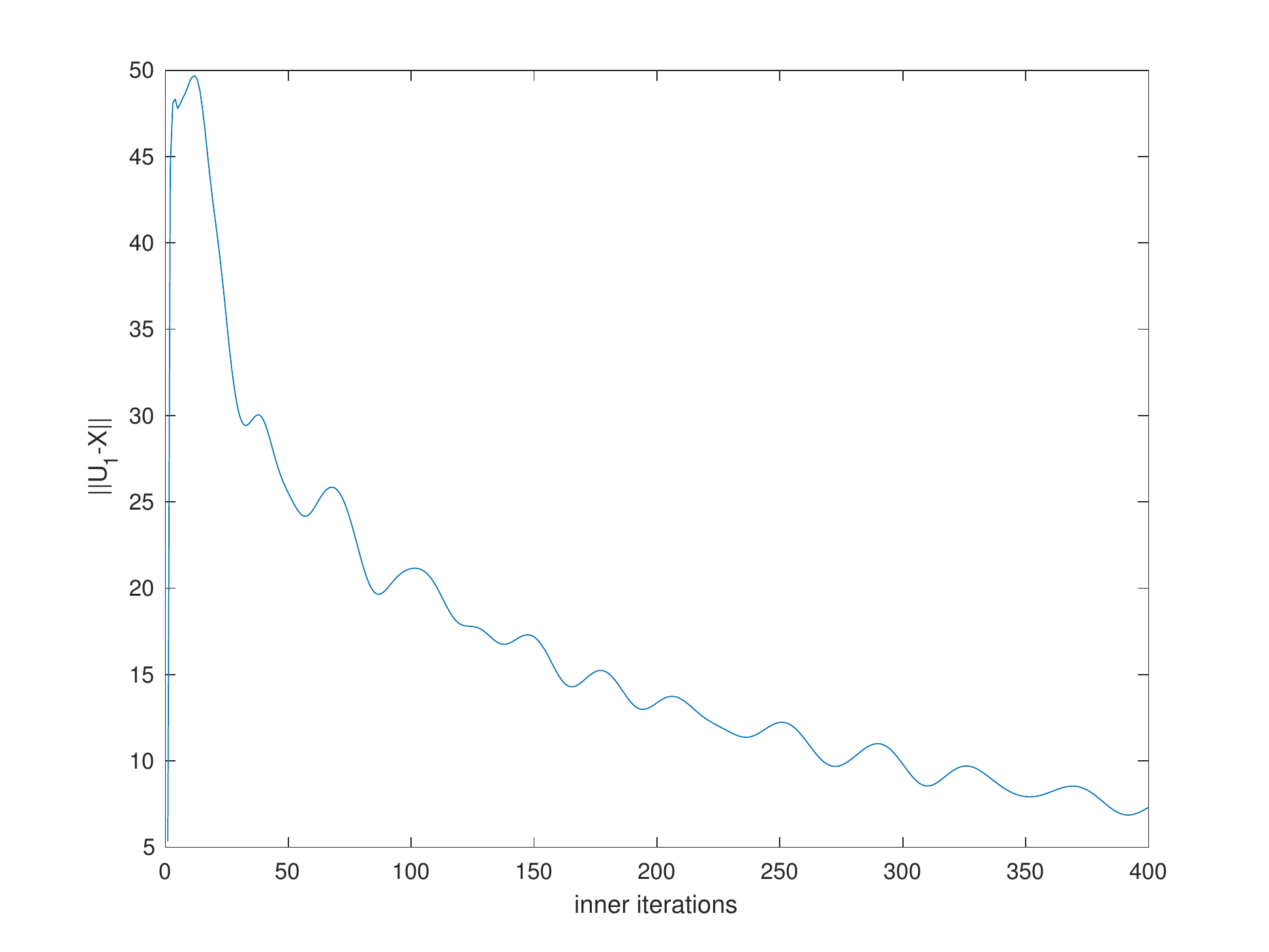}}\\
\caption{$\|\mathcal{U}_1-\mathcal{X}\|_F$ in each iteration.  }
\label{fig:converage_u1}
\end{figure}

\begin{figure}[htbp]
\centering
\subfloat[In the 1st outer iteration]{\includegraphics[width=0.49\textwidth, height=0.28\textheight]{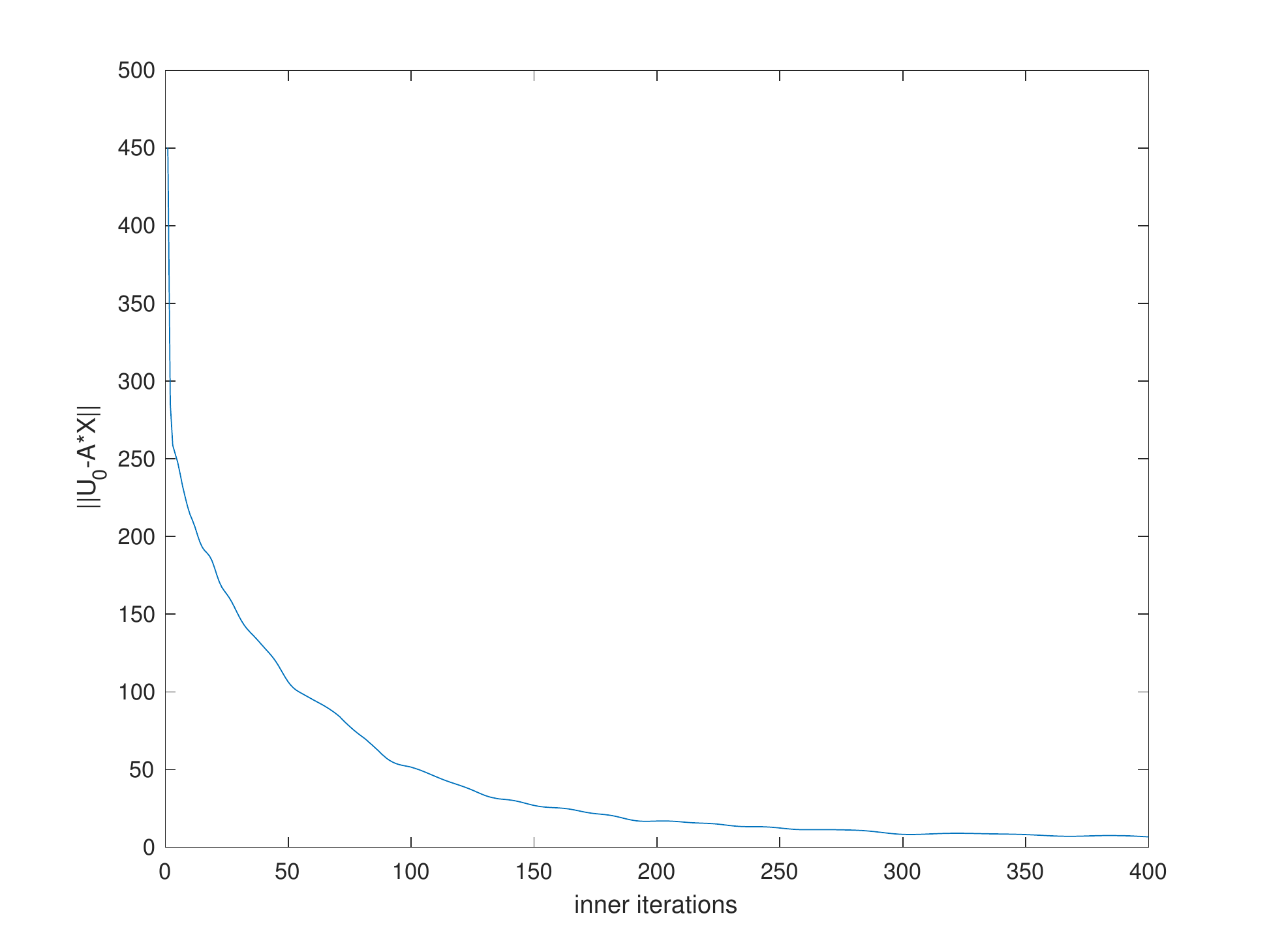}}
\subfloat[In the 2nd outer iteration]{\includegraphics[width=0.49\textwidth, height=0.28\textheight]{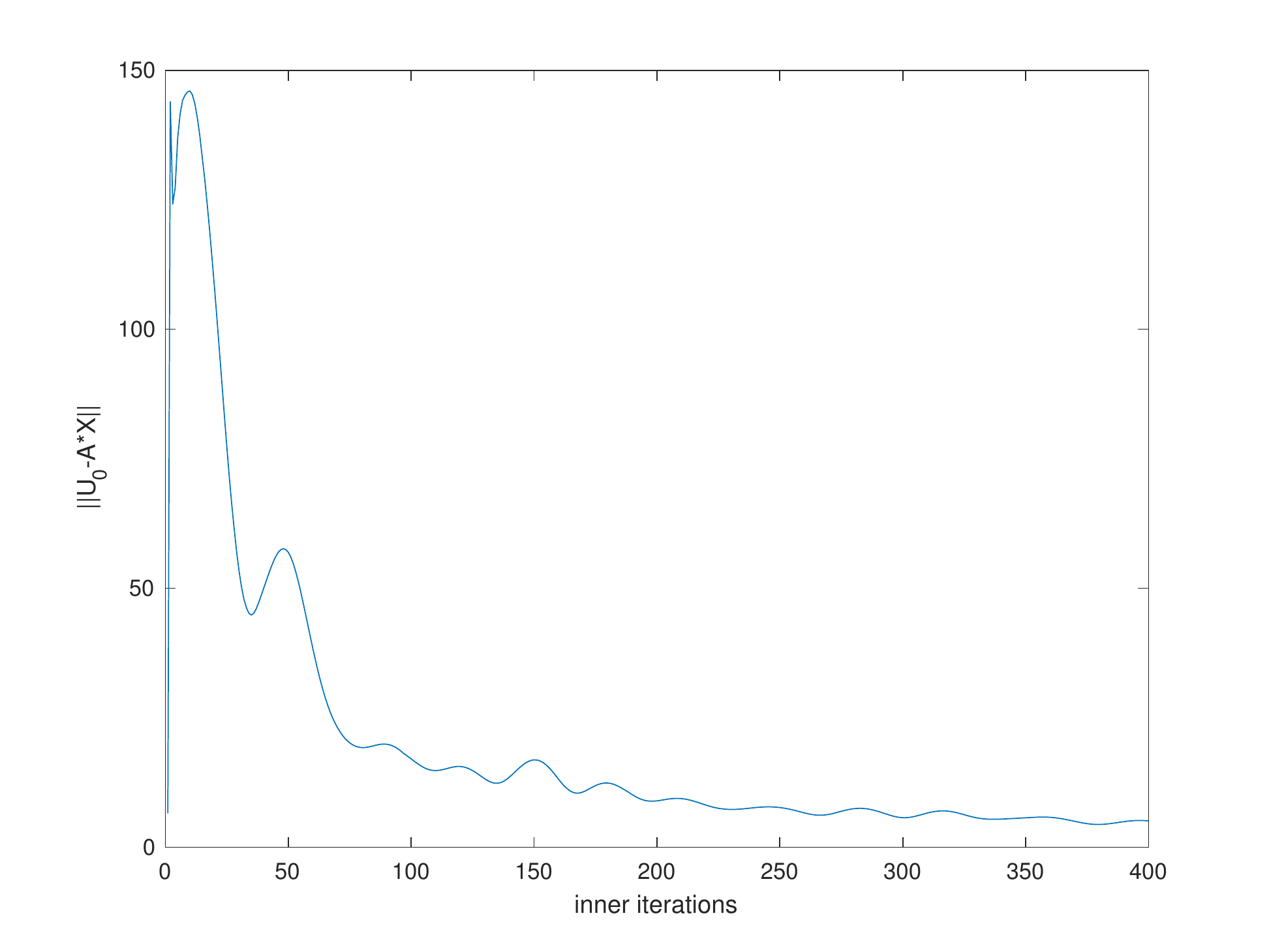}}\\
\caption{$\|\mathcal{U}_0-\mathcal{A}\ast \mathcal{X}\|_F$ in each iteration.  }
\label{fig:converage_u0}
\end{figure}

{
At the end of this subsection, we consider {the stability of our algorithm under PSF errors and bias. As an example of such an error, which might arise from imperfections of the phase-mask fabrication process, we numerically add a Gaussian phase randomness at the level of a fraction of a wave to the spiral phase mask.} That is to say, $A_{\zeta}(\mathbf{s} )$ in \eqref{equ:A} is changed to 
\begin{equation*}
	\hat{A}_{\zeta}(\mathbf{s}) = {1\over \pi}\left|\int P(\mathbf{u} )\mathrm{exp} \left[ \iota( 2\pi \mathbf{u}\cdot\mathbf{s} +  \zeta u^2 - \psi(\mathbf{u}) + \sigma \mathcal{N})  \right] d \mathbf{u} \right|^2,
\end{equation*}
where $\sigma\mathcal{N}$ {represents a Gaussian noise term} with standard deviation $\sigma$ chosen from 1/40th to 1/10th of a wave in the spiral phase, {\it i.e.} $\sigma$ is set between $ \frac{2\pi}{40}$ and $\frac{2\pi}{10}$ {radians}.
 We apply the Gaussian noise by using the MATLAB command ``randn'' and  consider the case of 15 point sources by  calculating the average of recall and precision rate for 50 observed images randomly generated from the distribution of point sources.  {We compare the results with and without any PSF error. See \Cref{fig:stability}. The recall rate is stable for all these noise levels.  When the noise level, $\sigma$, is less than 0.4 radian, ({\it i.e.}, about 1/25th wave of spiral phase), the precision rate is within 3\% of the noise-less case. But when $\sigma$ is  larger than 0.4 radian, the precision rate drops by about 8\%.  These results show good stability of our algorithm under modest PSF phase errors.} 
\begin{figure}[htbp]
\centering
\includegraphics[width=0.7\textwidth]{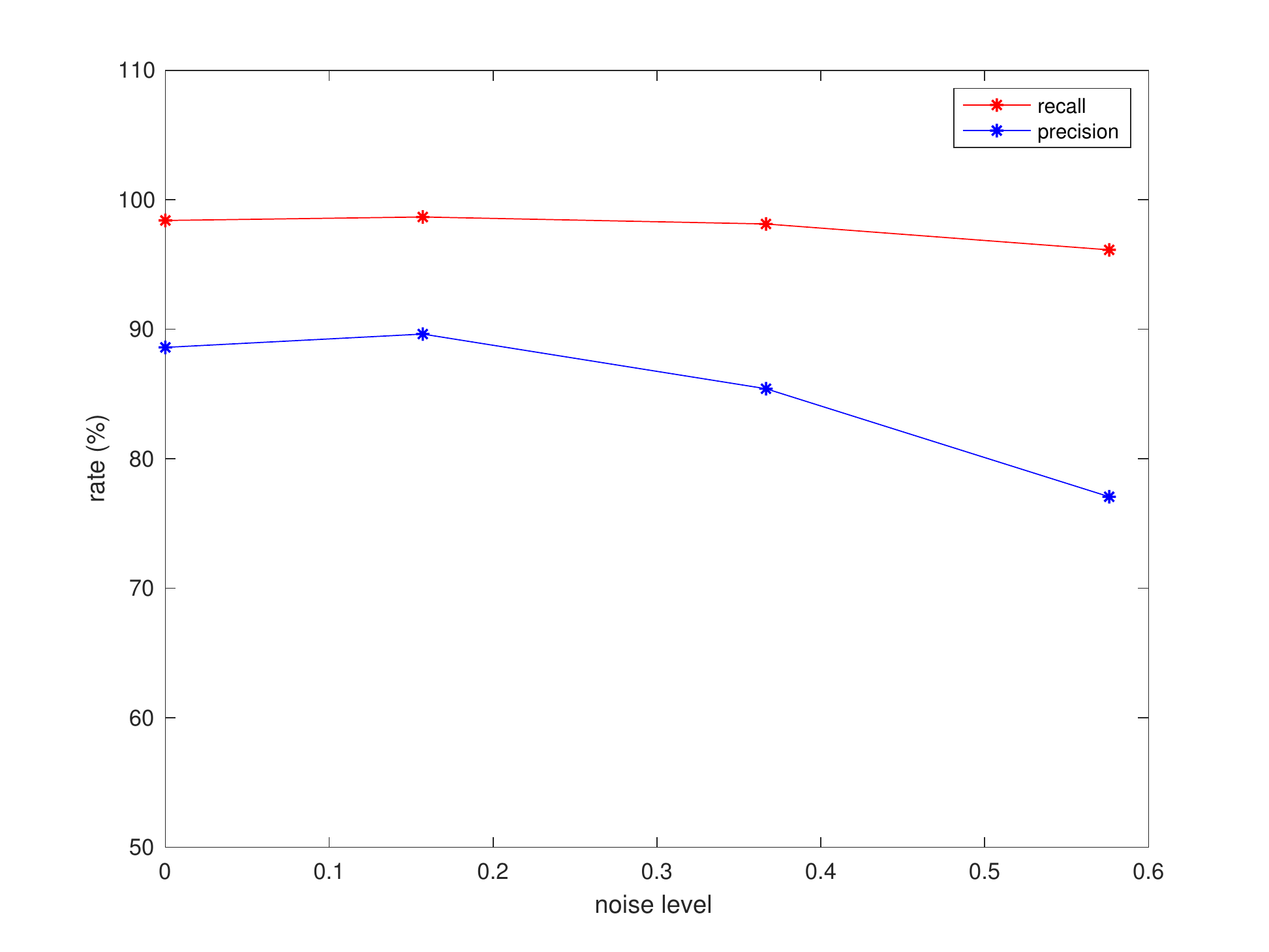}
\caption{Recall and precision rates under different levels of errors in the degraded PSF. }
\label{fig:stability}
\end{figure} 
}

\subsection{Comparison of Algorithm KL-NC with other algorithms}
In this subsection we compare our approach with three other optimization algorithms: KL-$\ell_1$ (KL data fitting with $\ell_1$ regularization model);
 $\ell_2$-$\ell_1$ (least squares fitting term with $\ell_1$ regularization model) and $\ell_2$-NC (least squares fitting term with non-convex regularization model). For all these models, we do the same post-processing and estimation of flux values after solving the corresponding optimization problem.

 In KL-$\ell_1$, we need to solve the following optimization problem
 \begin{equation*}
 	\min\limits_{\mathcal{X}\geq 0 } \left\{  D_{KL}(\mathcal{T}(\mathcal{A} \ast \mathcal{X})+b \,1, G) +\mu\|\mathcal{X}\|_1 \right\}.
 \end{equation*}
 We use ADMM by introducing two auxiliary variables $\mathcal{U}_0 = \mathcal{A}\ast\mathcal{X}$  and $\mathcal{U}_1 = \mathcal{X}$. We  get a similar augmented Lagrangrian function with $w_{ijk}^l$ replaced by $\mu$ in \eqref{equ:aug}. Therefore, we use \Cref{them:closed_form_soln} to solve this problem.

 In $\ell_2$-$\ell_1$, we need to solve the following optimization problem
 \begin{equation}
 \label{equ:l2-l1}
 	\min\limits_{\mathcal{X}\geq 0 }\left\{ \frac{1}{2} \|\mathcal{T}(\mathcal{A} \ast \mathcal{X})+b \,1- G\|_F^2 +\mu\|\mathcal{X}\|_1\right\}.
 \end{equation}
The problem is similar to the optimization problem in \cite{Rice2016generalized}, except that we include
the background 
which is estimated easily by  signal processing. We can therefore still use ADMM for solving \eqref{equ:l2-l1}.

In $\ell_2$-NC, we need to solve the following optimization problem
 \begin{equation}
 \label{equ:l2-nc}
 	\min\limits_{\mathcal{X}\geq 0 } \left\{  \frac{1}{2} \left\|\mathcal{T}(\mathcal{A} \ast \mathcal{X})+b \,1- G \right\|_F^2  +\mu\mathcal{R}(\mathcal{X})\right\},
 \end{equation}
where $\mathcal{R}(\mathcal{X})$ is the same as KL-NC method. We see that  IRL1 can also work well for solving \eqref{equ:l2-nc} with the $\ell_1$ weighted optimization as a subproblem which we solve by ADMM since it is similar to $\ell_2$-$\ell_1$.

 Both the initial guesses of $\mathcal{X}$ and $\mathcal{U}_0$ are set as 0 for all these methods. In order to do the comparison, we plot the solutions for the four models as well as the ground truth in the same space; see \Cref{fig:compare_distr_15,fig:compare_distr_30} which correspond to the cases of 15 and 30 point sources, respectively.

\begin{figure}[htbp]
\centering
\subfloat[Without post-processing ($\ell_1$ models)]{\includegraphics[width=0.49\textwidth, height=0.32\textheight]{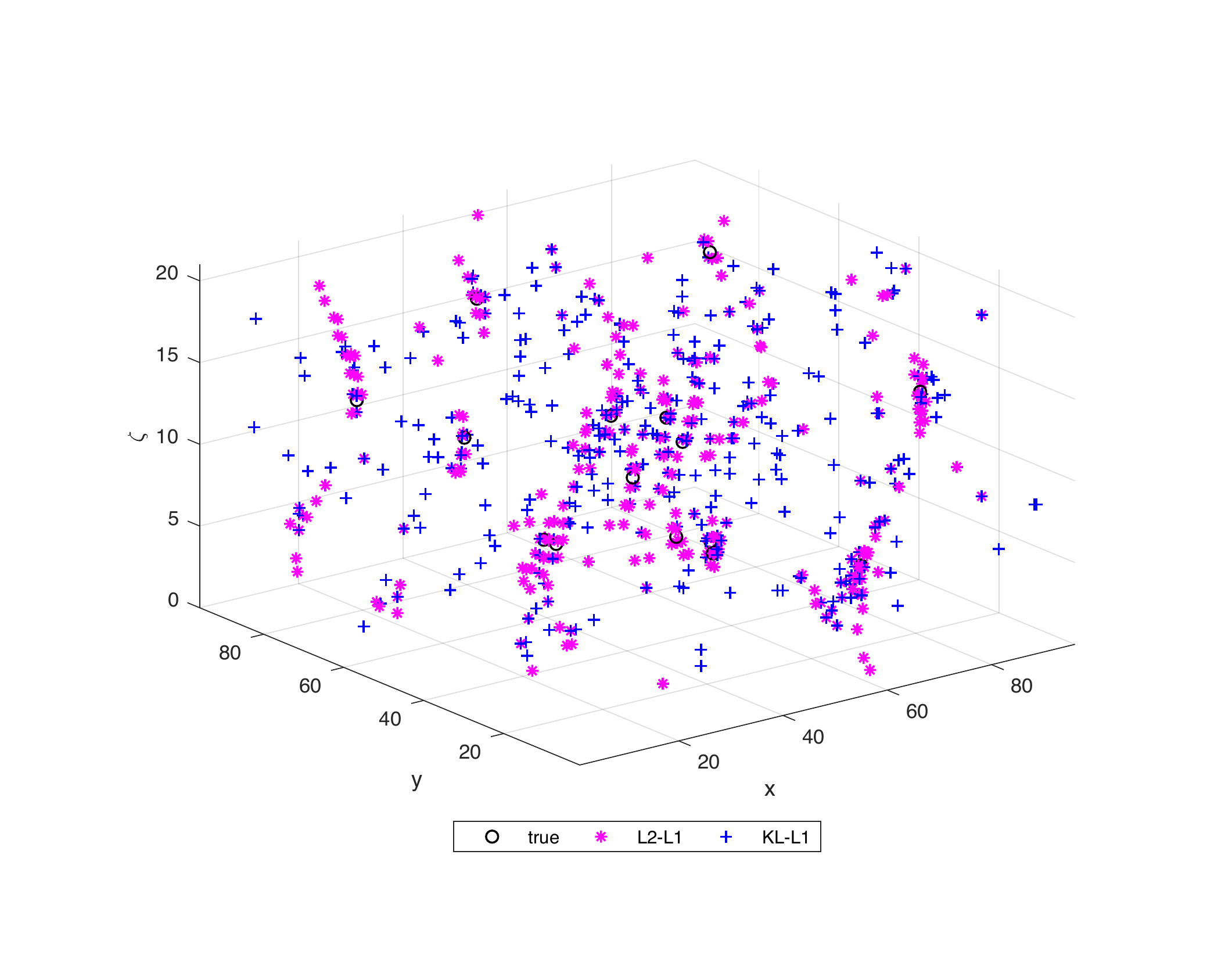}}
\subfloat[Without post-processing (non-convex model)]{\includegraphics[width=0.49\textwidth, height=0.32\textheight]{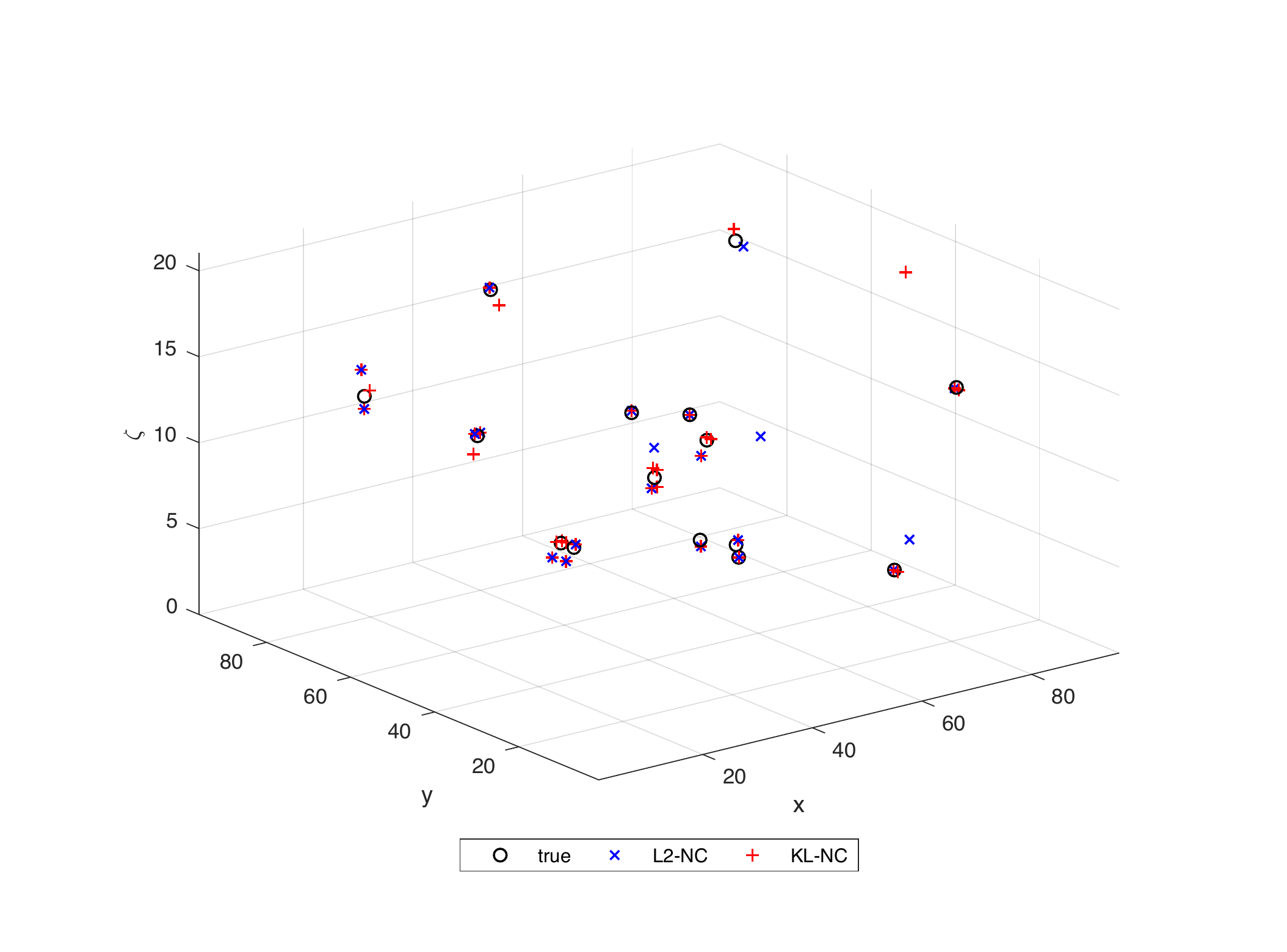}}\\
\subfloat[With post-processing ($\ell_1$ models)]{\includegraphics[width=0.49\textwidth, height=0.32\textheight]{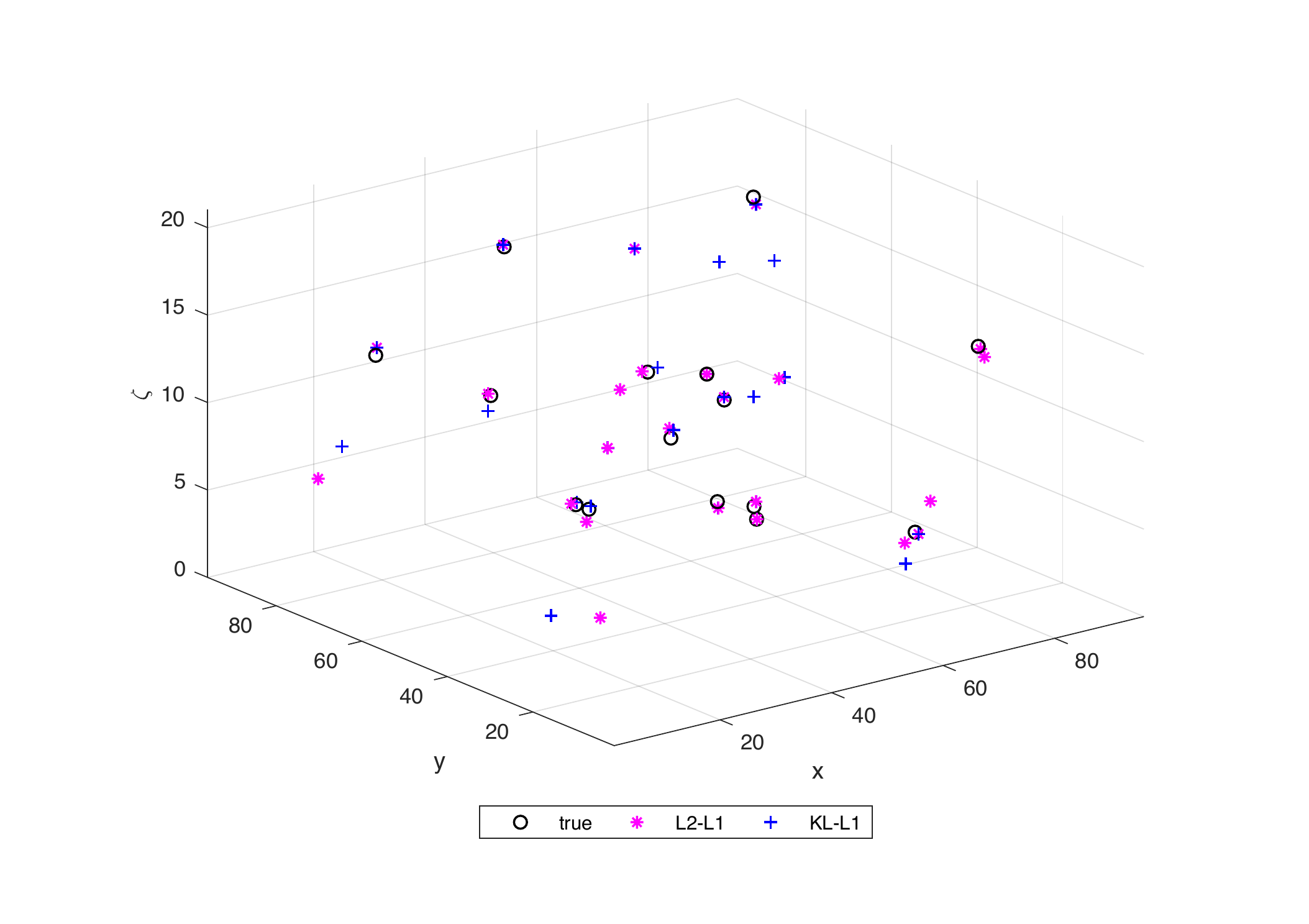}} 
\subfloat[With post-processing (non-convex model)]{\includegraphics[width=0.49\textwidth, height=0.32\textheight]{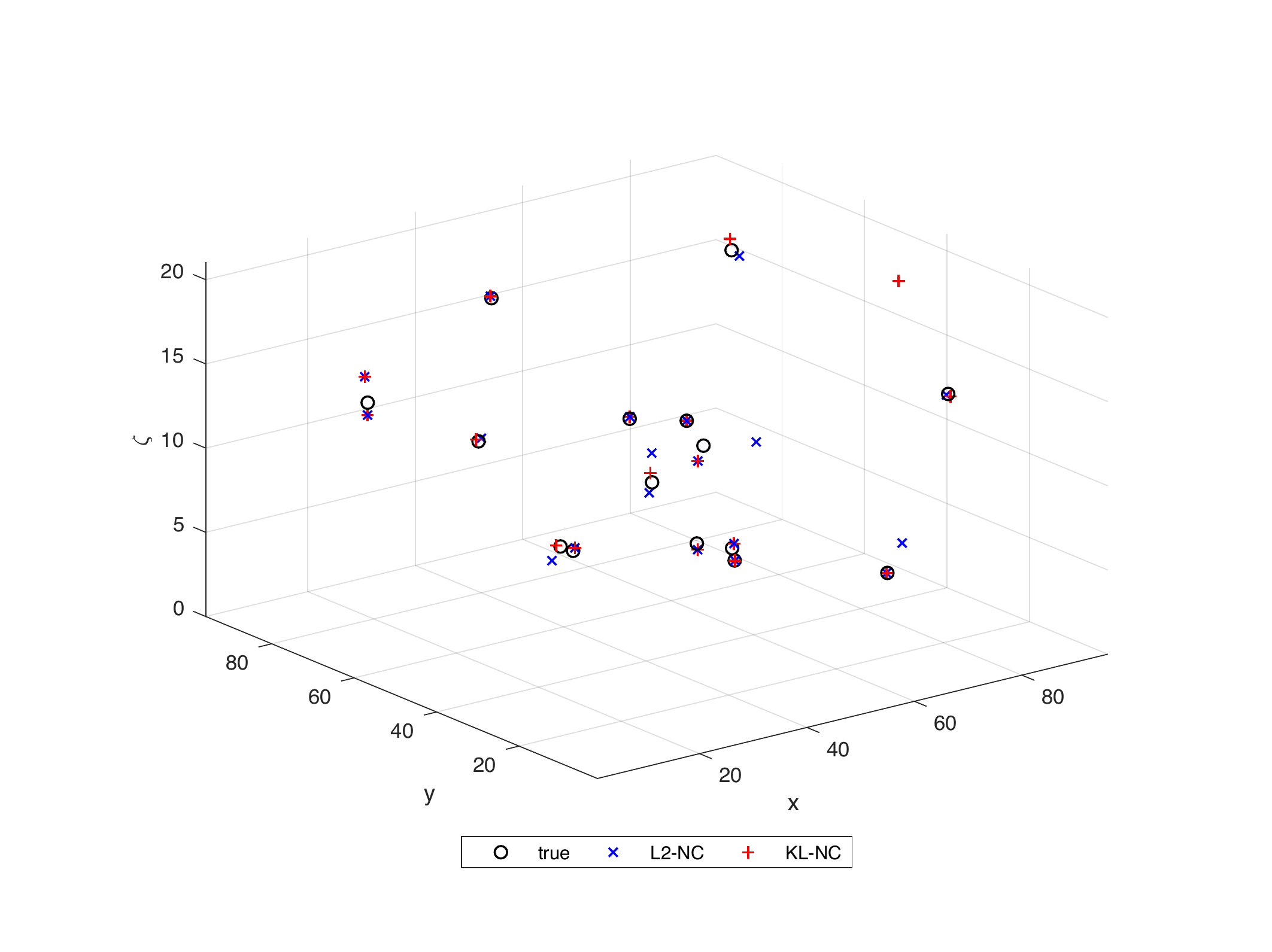}} 
\caption{Solution $\mathcal{X}$ from 4 algorithms (15 point sources). }
\label{fig:compare_distr_15}
\end{figure}
\begin{figure}[htbp]
\centering
\subfloat[Without post-processing ($\ell_1$ models)]{\includegraphics[width=0.49\textwidth, height=0.32\textheight]{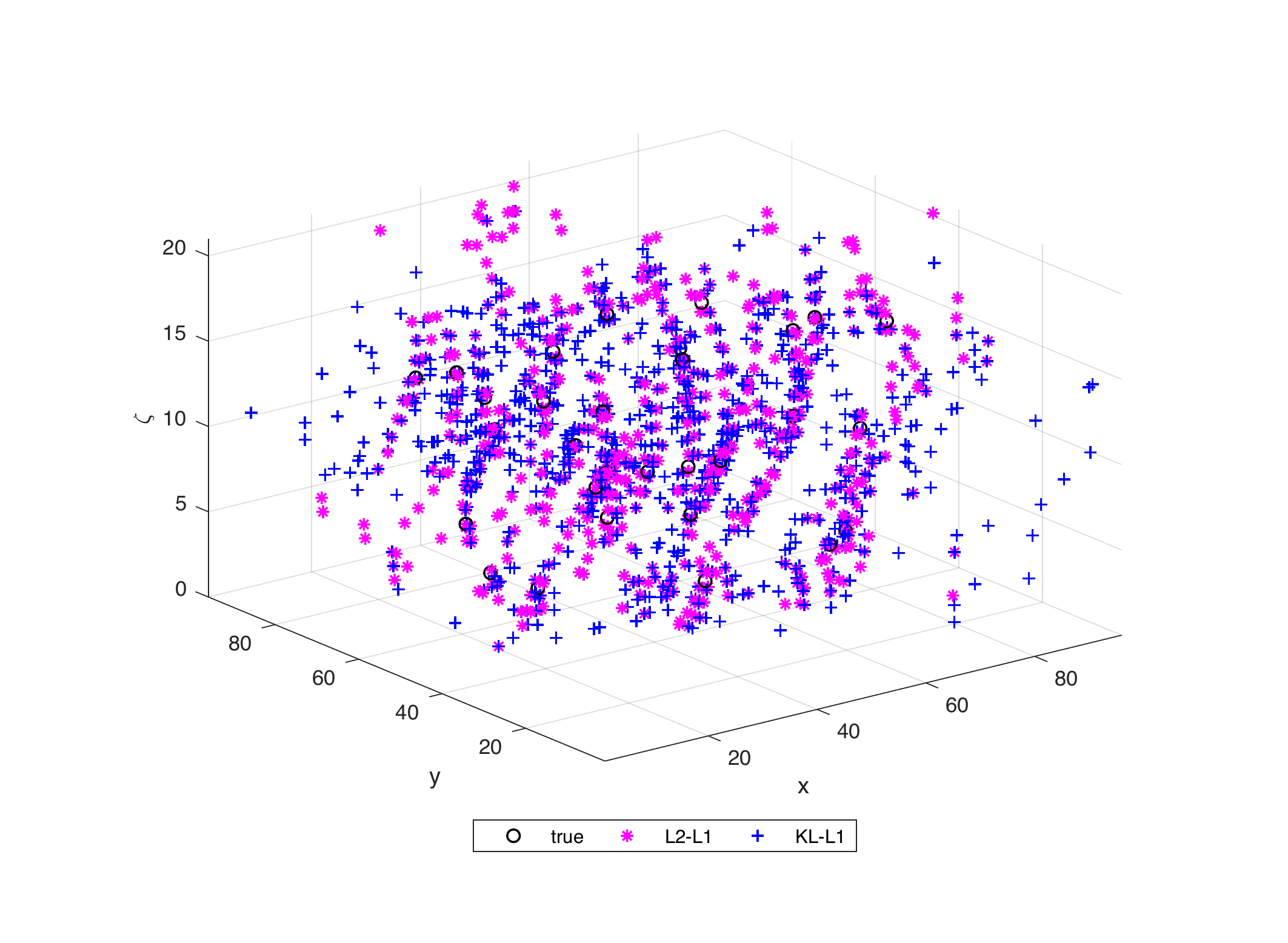}}
\subfloat[Without post-processing (non-convex model)]{\includegraphics[width=0.49\textwidth, height=0.32\textheight]{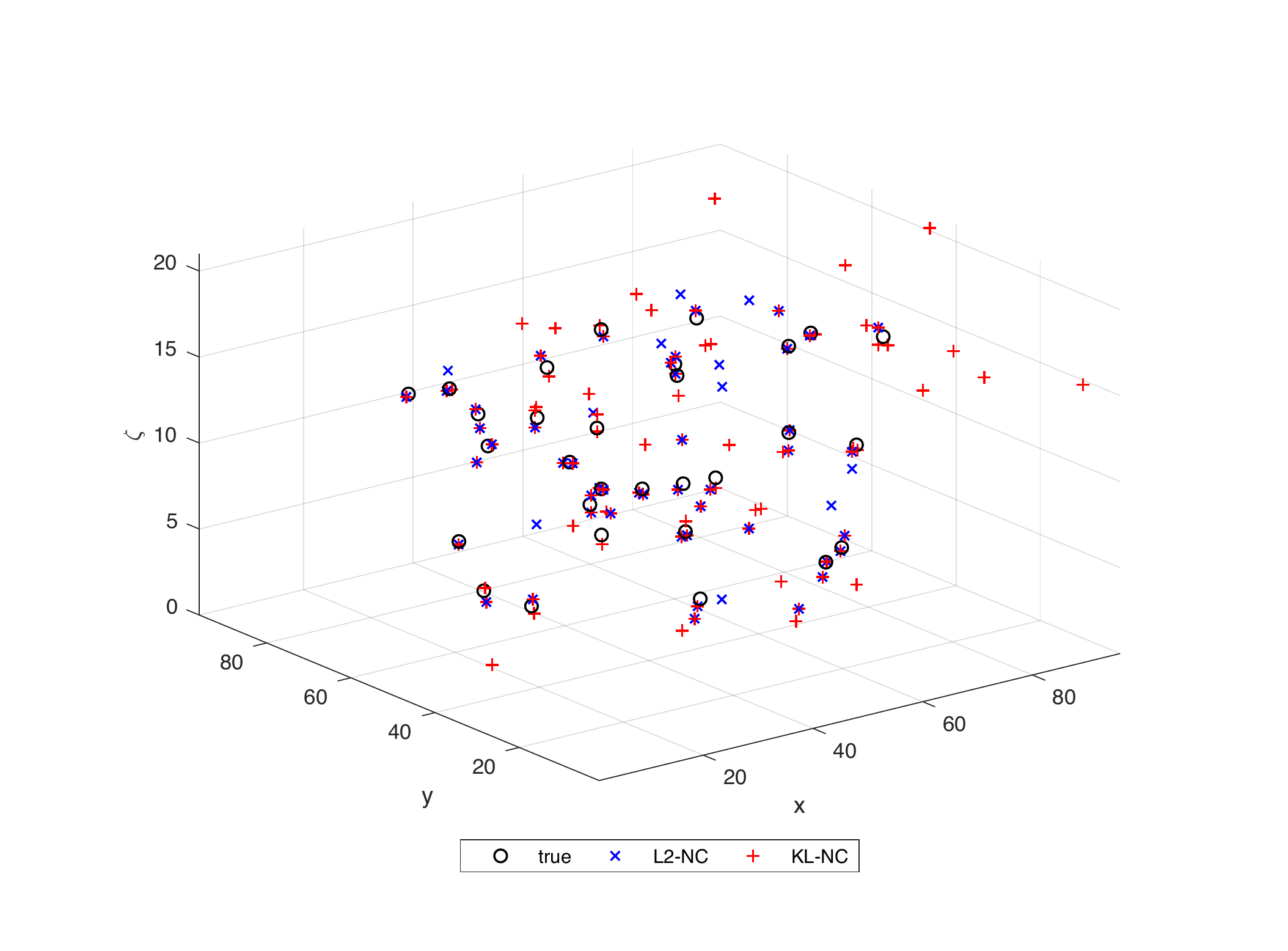}}\\
\subfloat[With post-processing ($\ell_1$ models)]{\includegraphics[width=0.49\textwidth, height=0.32\textheight]{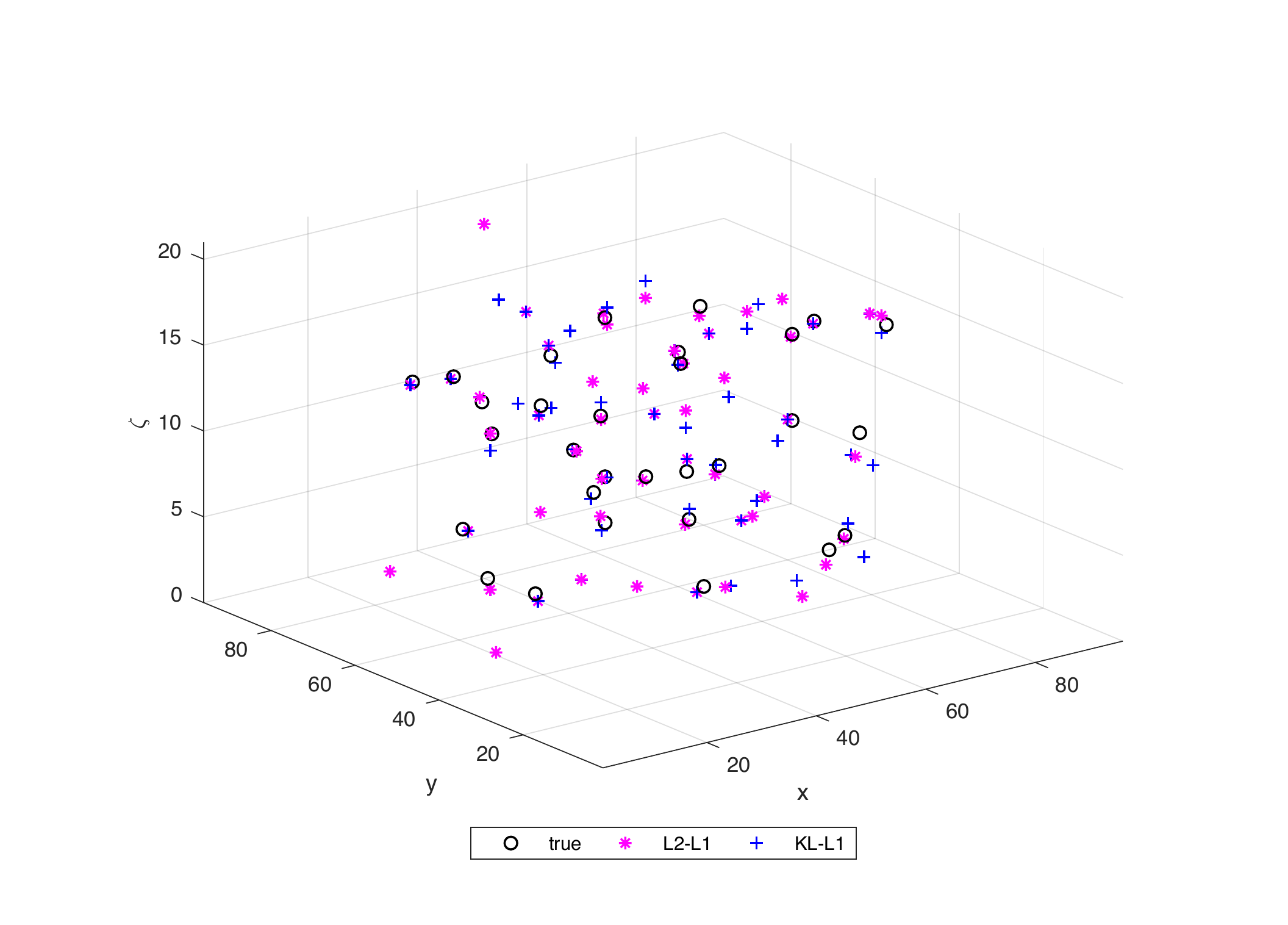}} 
\subfloat[With post-processing (non-convex model)]{\includegraphics[width=0.49\textwidth, height=0.32\textheight]{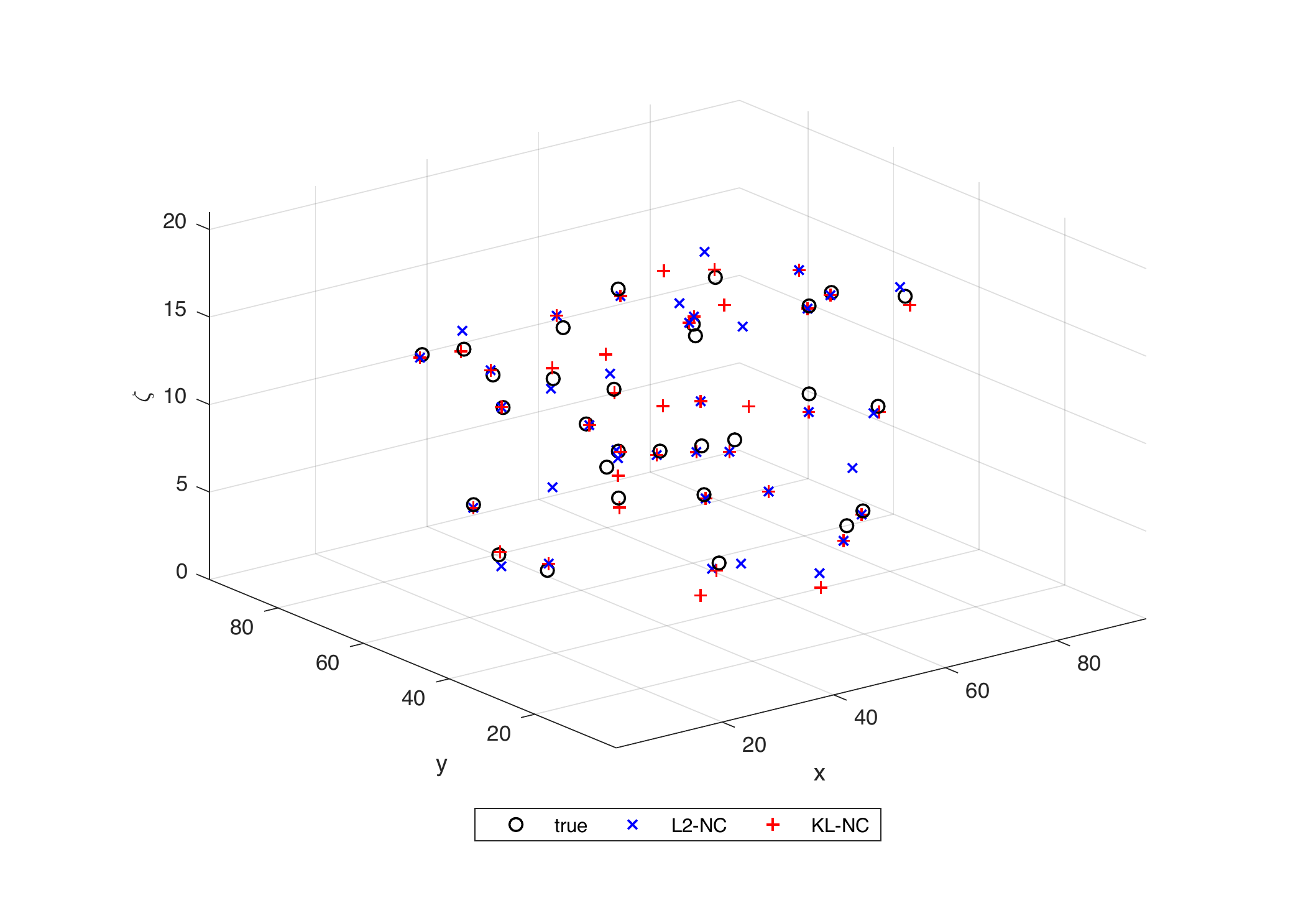}} 
\caption{Solution $\mathcal{X}$  from 4 algorithms (30 point sources). }
\label{fig:compare_distr_30}
\end{figure}

 From  \Cref{fig:compare_distr_15} and \Cref{fig:compare_distr_30}, we see the overfitting of the $\ell_1$ regularization models (KL-$\ell_1$ and $\ell_2$-$\ell_1$). Before post-processing, the solutions of these two algorithms spread out the PSFs a lot and have many false positives. After post-processing, both algorithms are improved, especially KL-$\ell_1$. However, in comparison to the non-convex regularization (KL-NC and $\ell_2$-NC), they still have many more false positives. Among the four algorithms, our approach (KL-NC) performs the best in terms of the recall and precision rates.

We also tested a number of other values of point source density, namely 5, 10, 15, 20, 30 and 40 point sources, and computed the average recall and precision rate of 50
images for each density and for each algorithm; see \Cref{Tab: with_postprocessing,Tab: without_postprocessing}. The results show the superior results  of our method in terms of both recall and precision rates. In \Cref{Tab: with_postprocessing},  the best recall and precision rates in each case are labeled by bold fonts.   As in the above discussion, our non-convex regularization tends to eliminate more false positives which increases the  precision rate. The KL data-fitting term, on the other hand, improves the recall rate as we see by comparing the results of KL-NC with $\ell_2$-NC. Before post-processing, we see that all the algorithms have low precision rates, especially  the two employing the $\ell_1$ regularization model at less than 10\%.

\begin{table}[htbp]
\begin{center}
\caption{Comparison of our KL-NC with $\ell_2$-$\ell_1$,  $\ell_2$-NC and  KL-$\ell_1$. All the results are {\bf after performing post-processing}. }
\begin{tabular}{c|cc|cc|cc|cc}
 \hline  &  \multicolumn{2}{c}{$\ell_2$-$\ell_1$} &  \multicolumn{2}{c}{$\ell_2$-NC}    &\multicolumn{2}{c}{KL-$\ell_1$} & \multicolumn{2}{c}{KL-NC}\\
\hline No.~Sources & Recall & Prec. & Recall & Prec. & Recall & Prec.  & Recall & Prec.    \\
\hline  5 & {\bf 100.00}\% & 68.91\%  & 97.60\% &  89.15\% & 98.93\% & 	58.64\% & {\bf 100.00}\% & {\bf  97.52}\%  \\
\hline 10 & {\bf 99.60}\% &  55.95\%   & 94.80\% &  83.51\% & 99.40\% &  65.24\% & { 99.40}\% & {\bf  93.69}\% \\
\hline 15 & {98.67}\% &  56.28\%  & 92.80\% &  84.77\% & {\bf 98.93}\% &  58.64\% & 98.40\% & {\bf  88.60}\%  \\
\hline 20 & 97.70\% &  56.50\%  & 95.20\% &  80.92\%   & {\bf 98.10}\% &  57.82\% & {97.70}\% & {\bf  87.49}\% \\
\hline 30  & {96.00}\% &  55.74\%  & 93.93\% &  77.77\% & {94.00}\% &  56.22\% & {\bf 96.20}\% & {\bf 79.75}\% \\
\hline 40  & 93.80\% &  52.68\%  & {\bf 95.40}\% &  59.34\% & 93.70\% &  54.29\% & {95.00}\% & {\bf 73.35}\% \\
\hline
\end{tabular}\label{Tab: with_postprocessing}
\medskip
\end{center}
\end{table}

\begin{table}[htbp]
\begin{center}
\caption{Comparison of our KL-NC with $\ell_2$-$\ell_1$,  $\ell_2$-NC and  KL-$\ell_1$.  All the results are {\bf before post-processing}.  }
\begin{tabular}{c|cc|cc|cc|cc}
\hline  &  \multicolumn{2}{c}{$\ell_2$-$\ell_1$} &  \multicolumn{2}{c}{$\ell_2$-NC}    &\multicolumn{2}{c}{KL-$\ell_1$} & \multicolumn{2}{c}{KL-NC}\\
\hline No.~Sources & Recall & Prec. & Recall & Prec. & Recall & Prec.  & Recall & Prec.    \\

\hline  5 & 100.00\% & 9.36\%  & 98.80\% &  58.97\% & 100.00\% & 3.41\% & 100.00\% & 52.27\%  \\
\hline 10  & 100.00\% &  8.48\%   & 96.80\% &  54.50\% & 100.00\% &  2.92\% & 100.00\% & 51.13\%\\
\hline 15  & 100.00\% &  4.32\%  & 96.47\% &  60.05\% & 100.00\% &  3.41\% & 100.00\% & 49.40\% \\
\hline 20  & 100.00\% &  4.41\%  & 97.70\% &  50.27\% & 100.00\% &  3.27\% & 99.30\% & 52.80\% \\
\hline 30  & 100.00\% &  4.47\%  & 98.13\% &  49.25\%  & 99.93\% &  2.71\% & 99.40\% & 35.67\%\\
\hline 40  & 100.00\% &  4.47\%  & 98.95\% &  30.27\% & 100.00\% &  3.75\% & 99.00\% & 30.31\%\\
\hline
\end{tabular}\label{Tab: without_postprocessing}
\medskip
\end{center}
\end{table}

\begin{figure}[htbp]
\centering
\subfloat[$\ell_2$-$\ell_1$]{\includegraphics[width=0.39\textwidth]{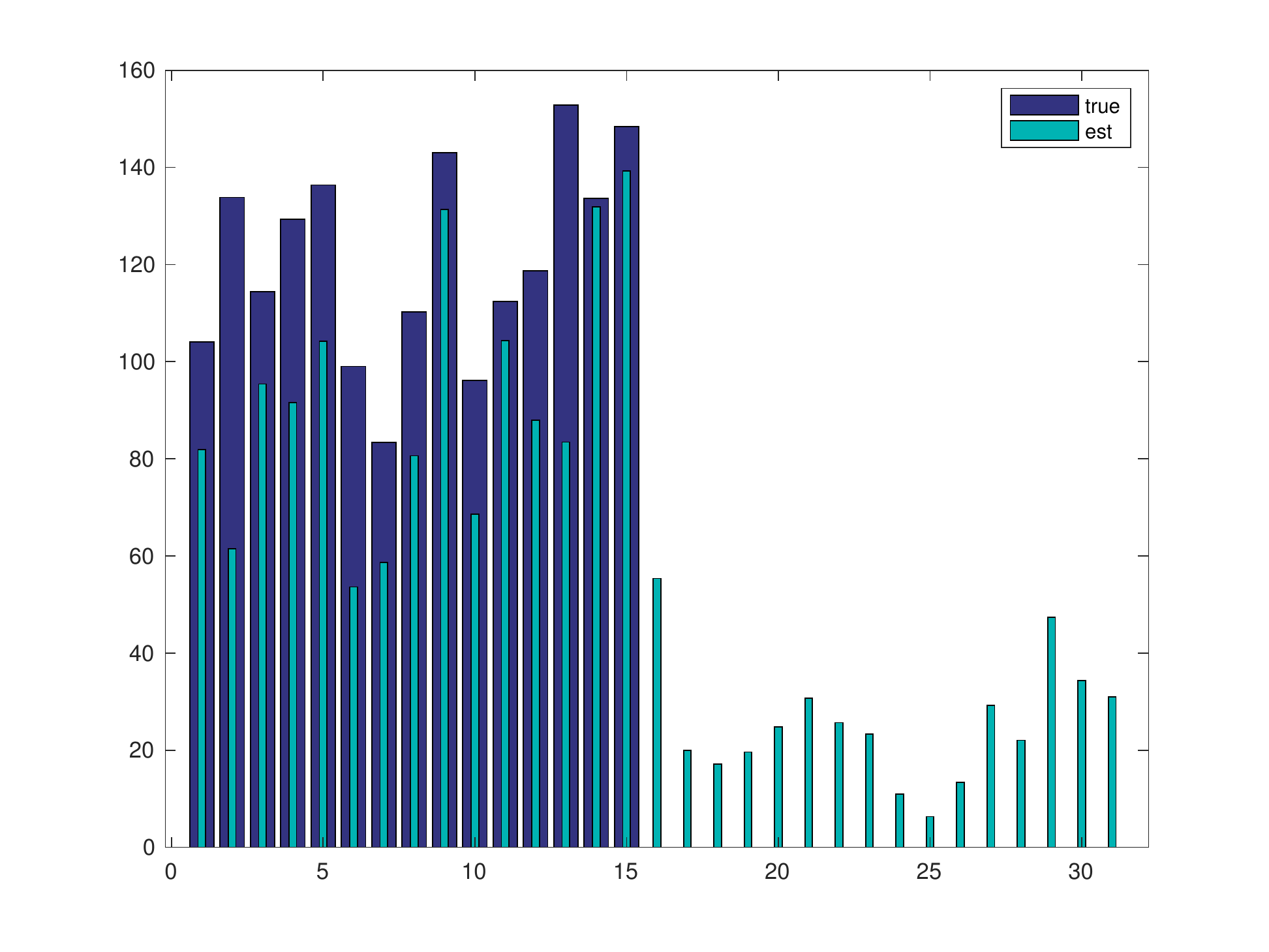}}
\hspace{0.01mm}
\subfloat[$\ell_2$-NC]{\includegraphics[width=0.39\textwidth]{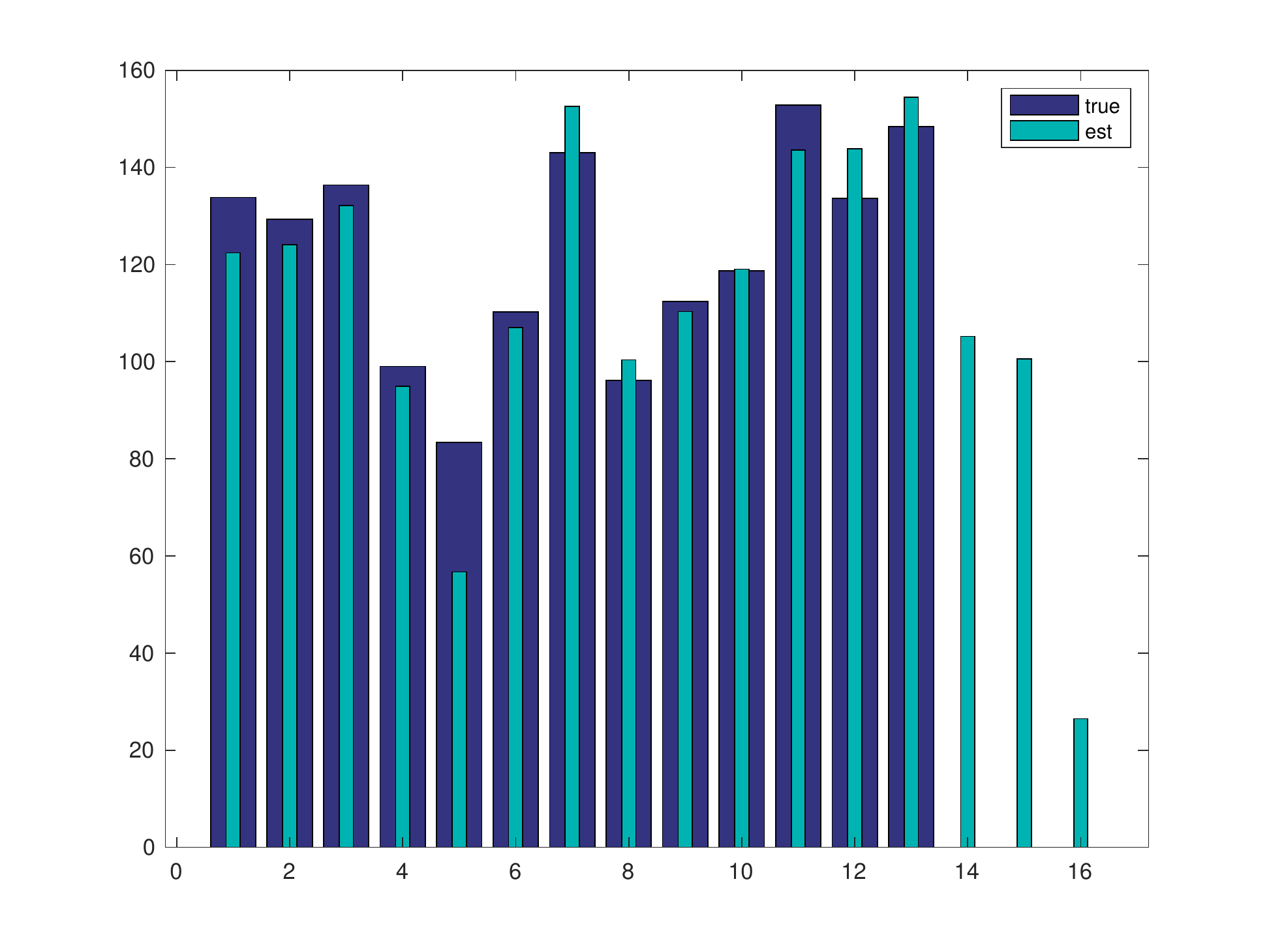}}
\hspace{0.01mm}
\subfloat[KL-$\ell_1$]{\includegraphics[width=0.39\textwidth]{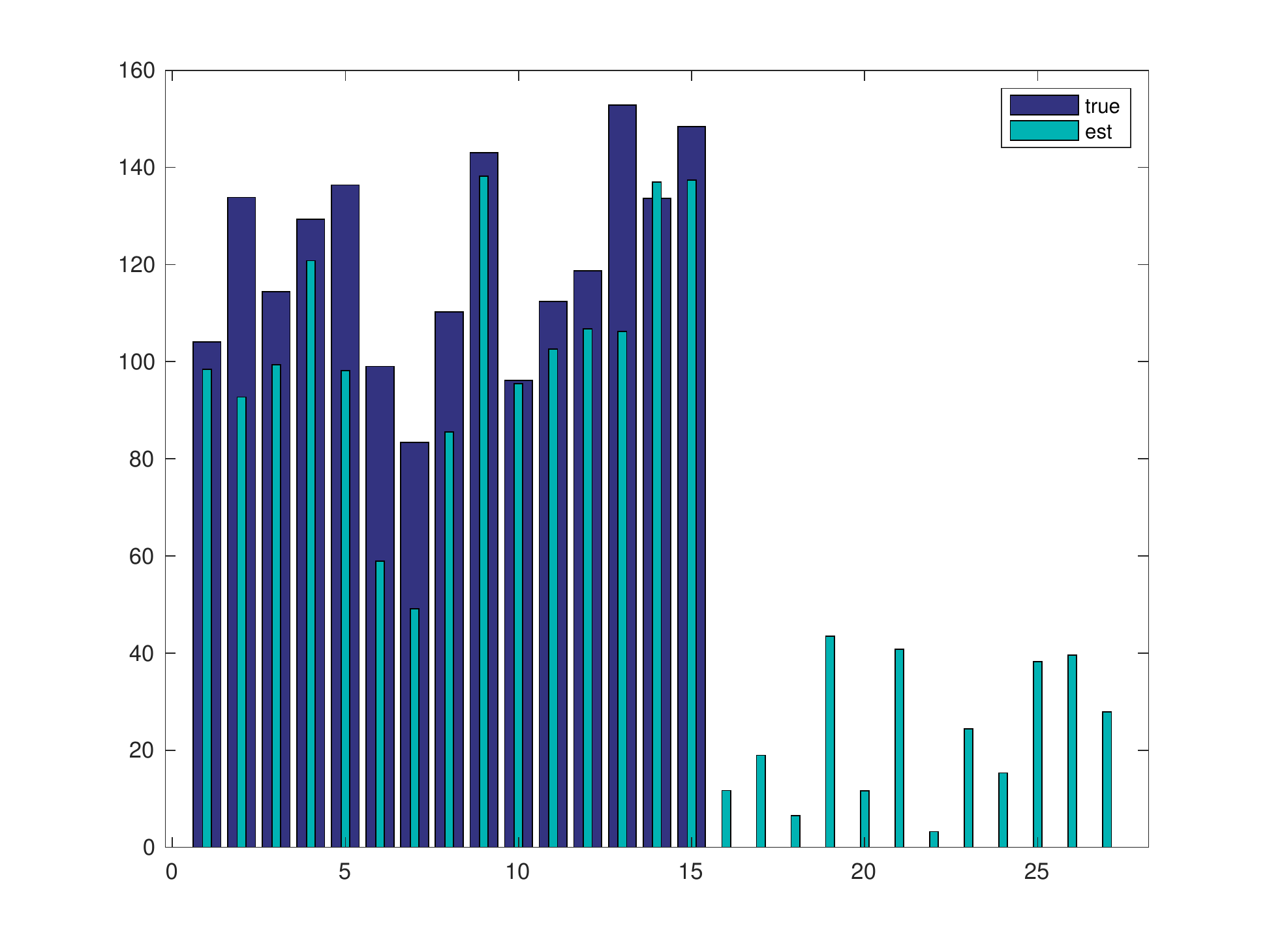}} 
\hspace{0.01mm}
\subfloat[KL-NC]{\includegraphics[width=0.39\textwidth]{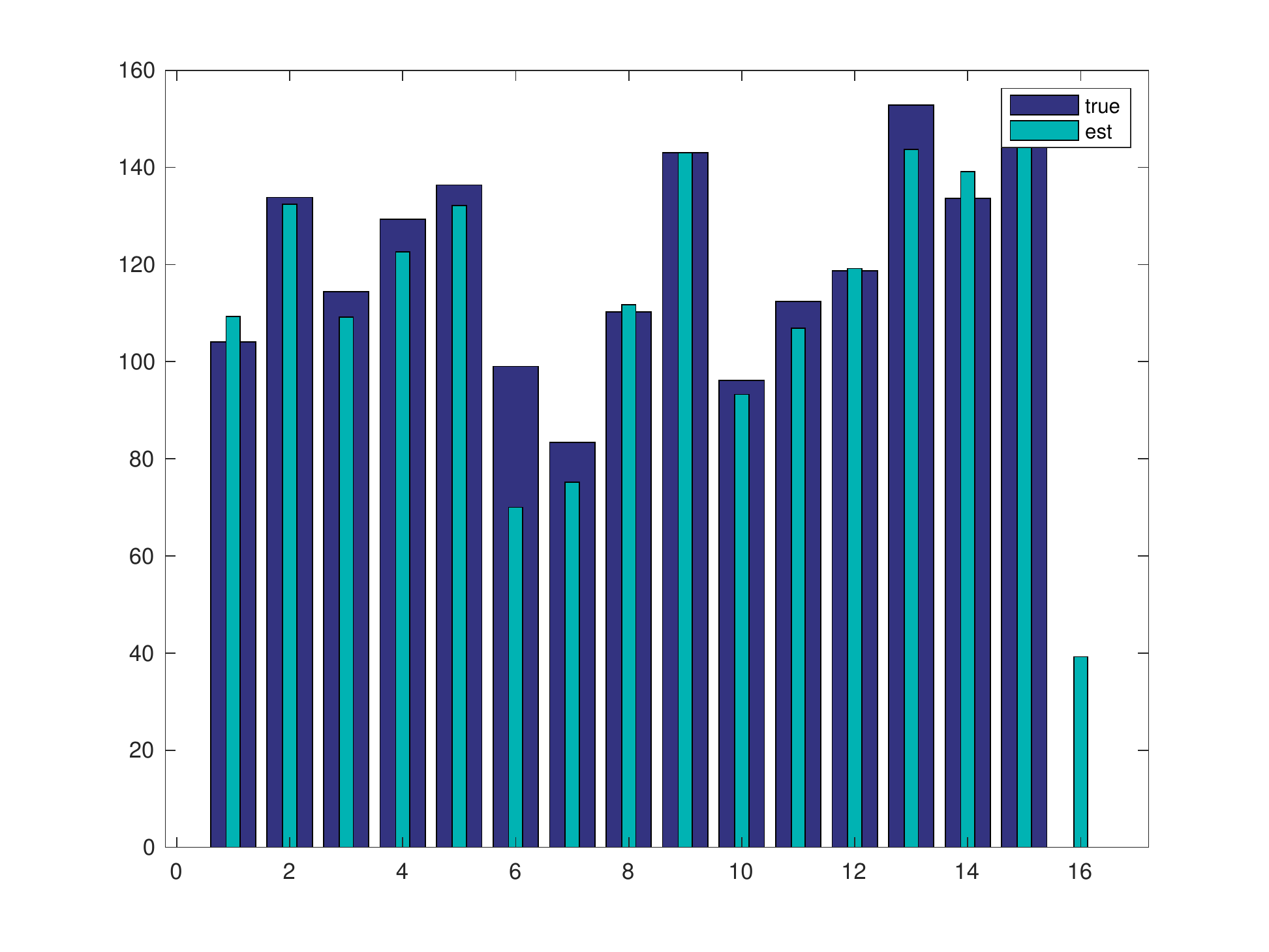}}
\caption{Estimating flux values. The bar graph with no value  ground truth part corresponds to a false positive. }\label{fig:flux}
\end{figure}

		In the following, we compare the results of the estimation of the flux in these four algorithms, considering specifically the case of 15 point sources. In \Cref{fig:flux}, we plot the fluxes of ground truth sources and the fluxes of the estimated point sources for the true positive point sources. For the false positive point sources, we only  show
the estimated fluxes.  Both $\ell_1$ models underestimate the fluxes. The rotating PSF images for false positives carry the energy away from the true positive source fluxes. In non-convex models, we also have similar observations when we have false positives. For example, in \Cref{fig:flux}(b), we see the flux on the fifth bar is  underestimated more than the others. We note that its rotating PSF image is overlapping with that of a false positive.
The more false positives an algorithm recovers the more they will spread out the intensity, leading to more underestimated fluxes for the true positives.

We also tested 50 different observed images for each density, and analyzed the relative error in the estimated flux values, which we define as
\begin{equation*}
\mathrm{error} = \frac{\mathbf{f}_{\mathrm{est}}-\mathbf{f}_{\mathrm{tru}}}{\mathbf{f}_{\mathrm{tru}}},
\end{equation*}
where $(\mathbf{f}_{\mathrm{est}}, \ \mathbf{f}_{\mathrm{tru}})$ is the pair which contains the flux of an identified true positive and the corresponsing  ground truth flux. We plot the histogram of the relative errors on these four models in \Cref{fig:15flux_per,fig:30flux_per}. We still see the advantage of KL-NC over other algorithms in this respect. The distribution of relative errors mostly lies within $[0, \ 0.1]$. For the $\ell_1$ regularization algorithms, the distribution of the relative error spreads out and there are  many cases with error higher than 0.3.  For $\ell_2$-NC, we have the similar result for the 15 point sources case; however, in higher density cases, KL-NC has better results comparing with $\ell_2$-NC.
\begin{figure}[htbp]
\centering
\subfloat[$\ell_2$-$\ell_1$]{\includegraphics[width=0.39\textwidth]{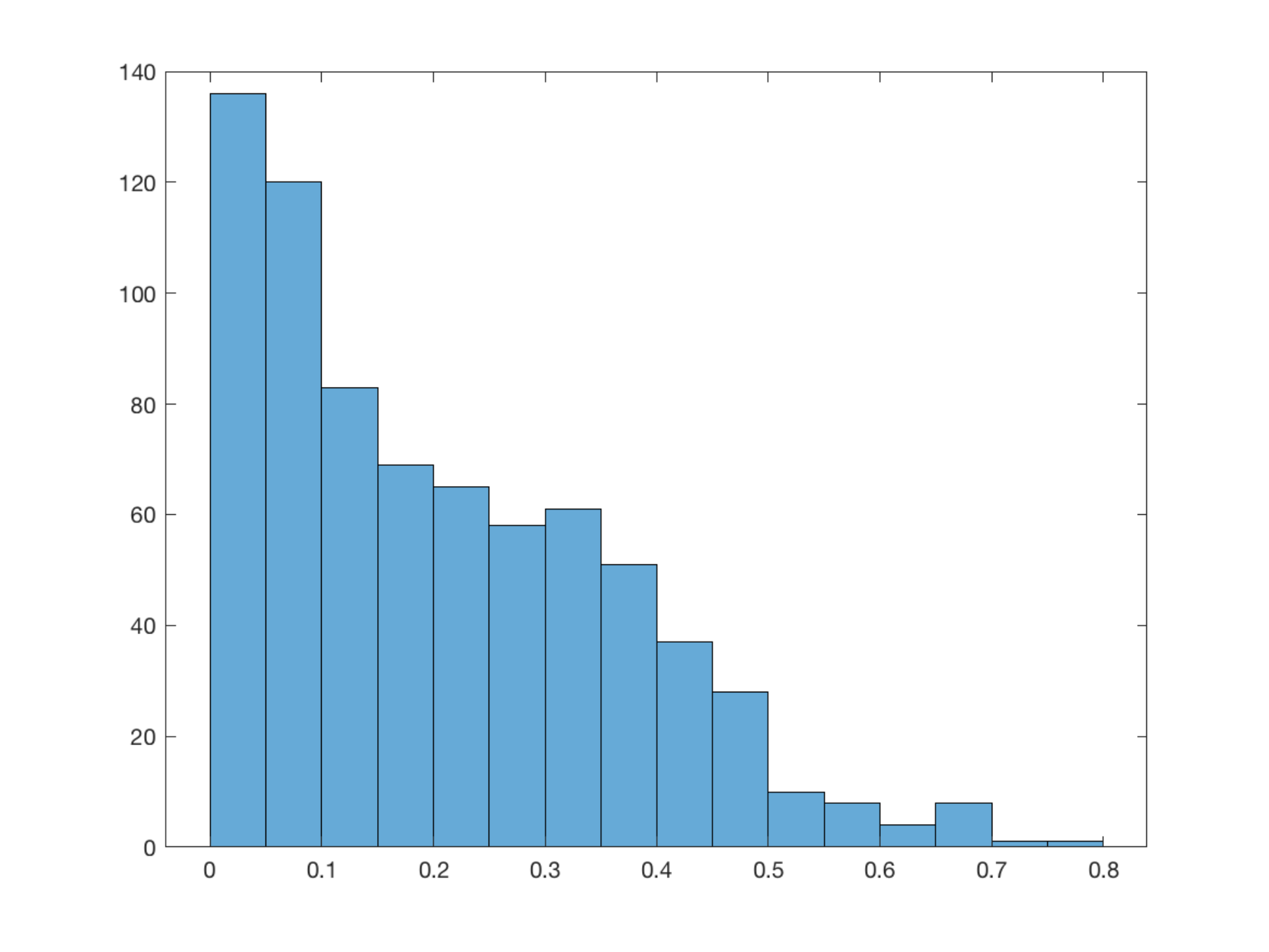}}
\hspace{0.01mm}
\subfloat[$\ell_2$-NC]{\includegraphics[width=0.39\textwidth]{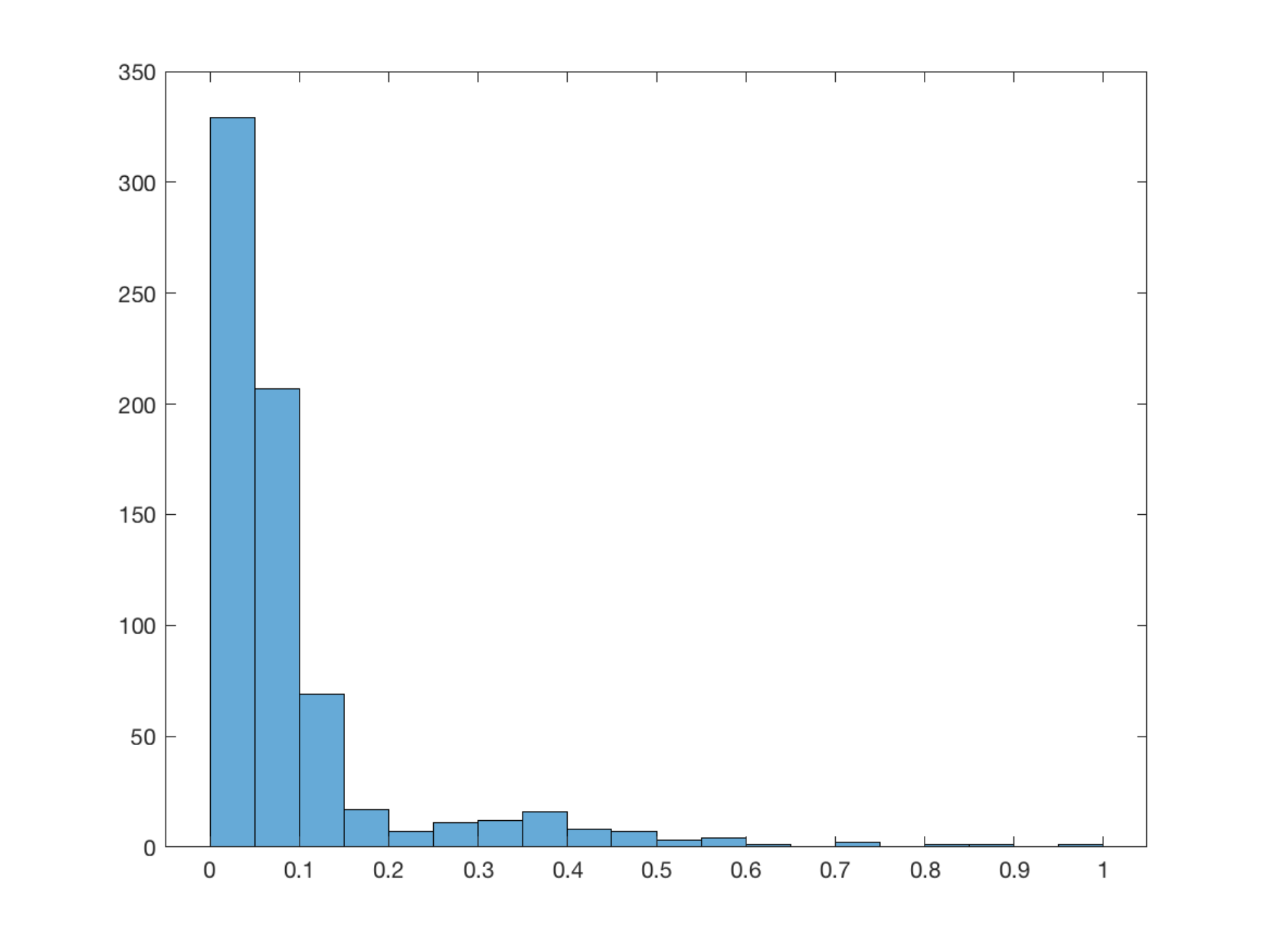}}
\hspace{0.01mm}
\subfloat[KL-$\ell_1$]{\includegraphics[width=0.39\textwidth]{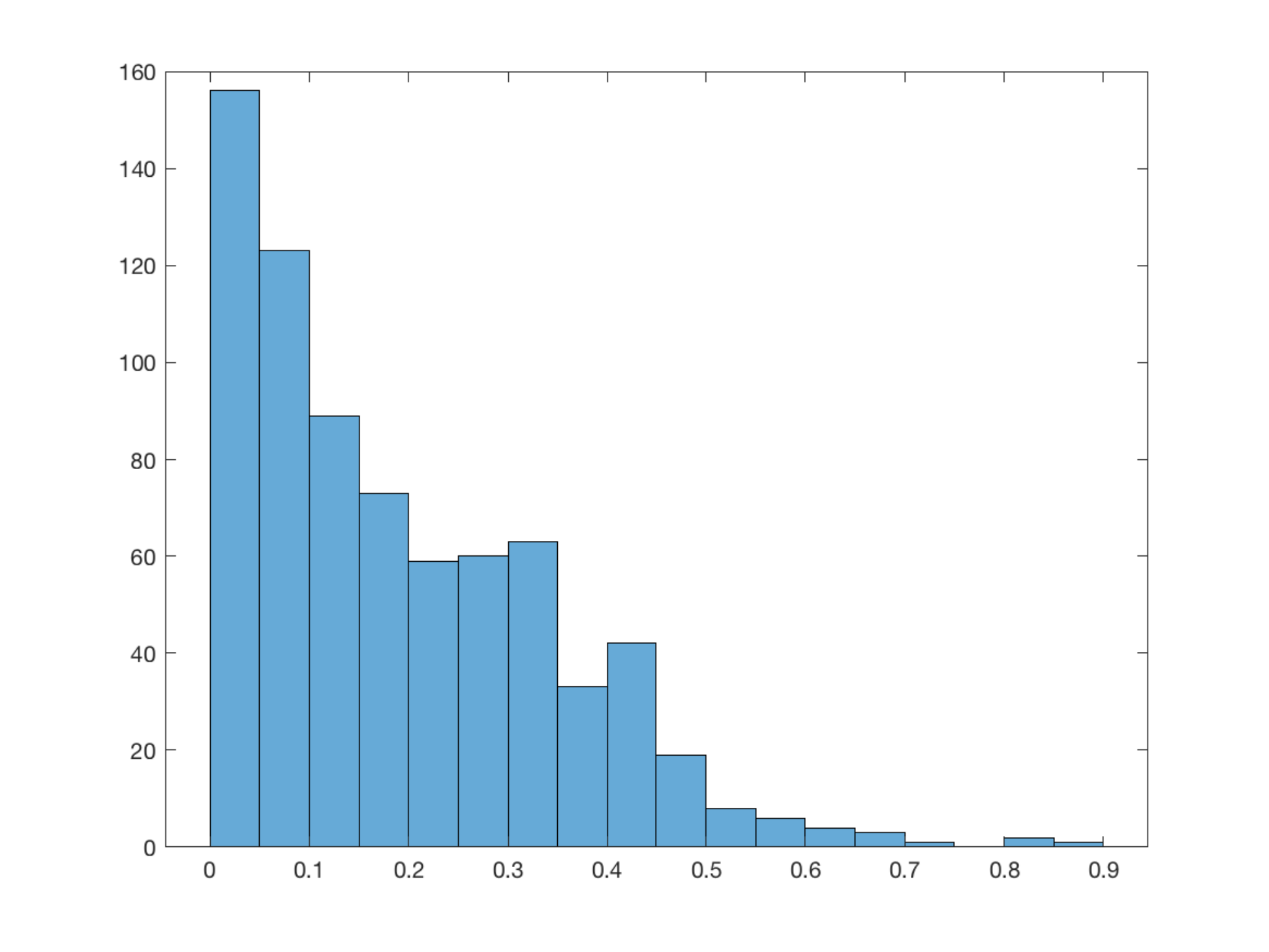}}
\hspace{0.01mm}
\subfloat[KL-NC]{\includegraphics[width=0.39\textwidth]{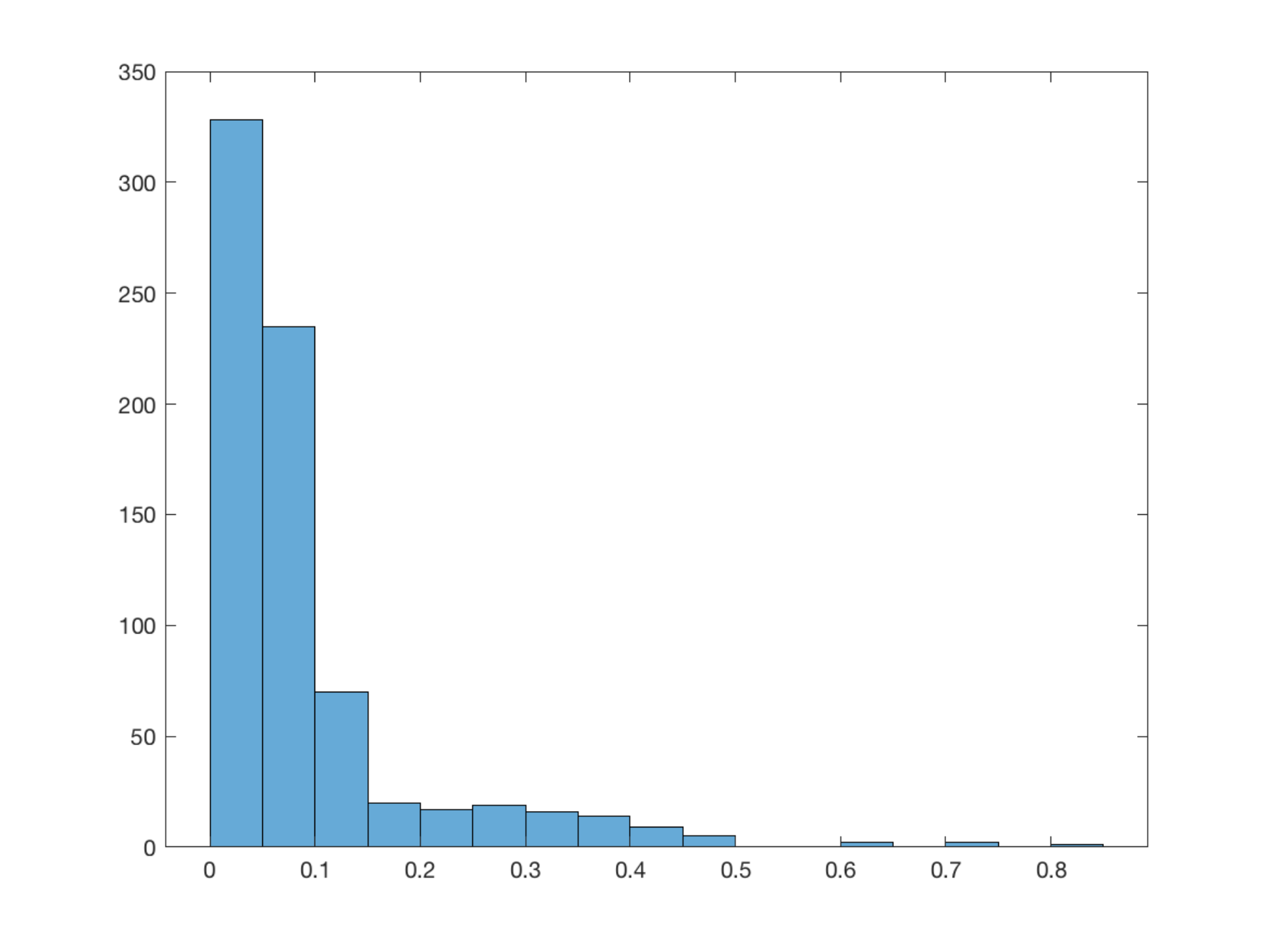}}
\caption{Histogram of the relative errors of the flux values for the 15 point sources case. }
\label{fig:15flux_per}
\end{figure}
\begin{figure}[htbp]
\centering
\subfloat[$\ell_2$-$\ell_1$]{\includegraphics[width=0.39\textwidth]{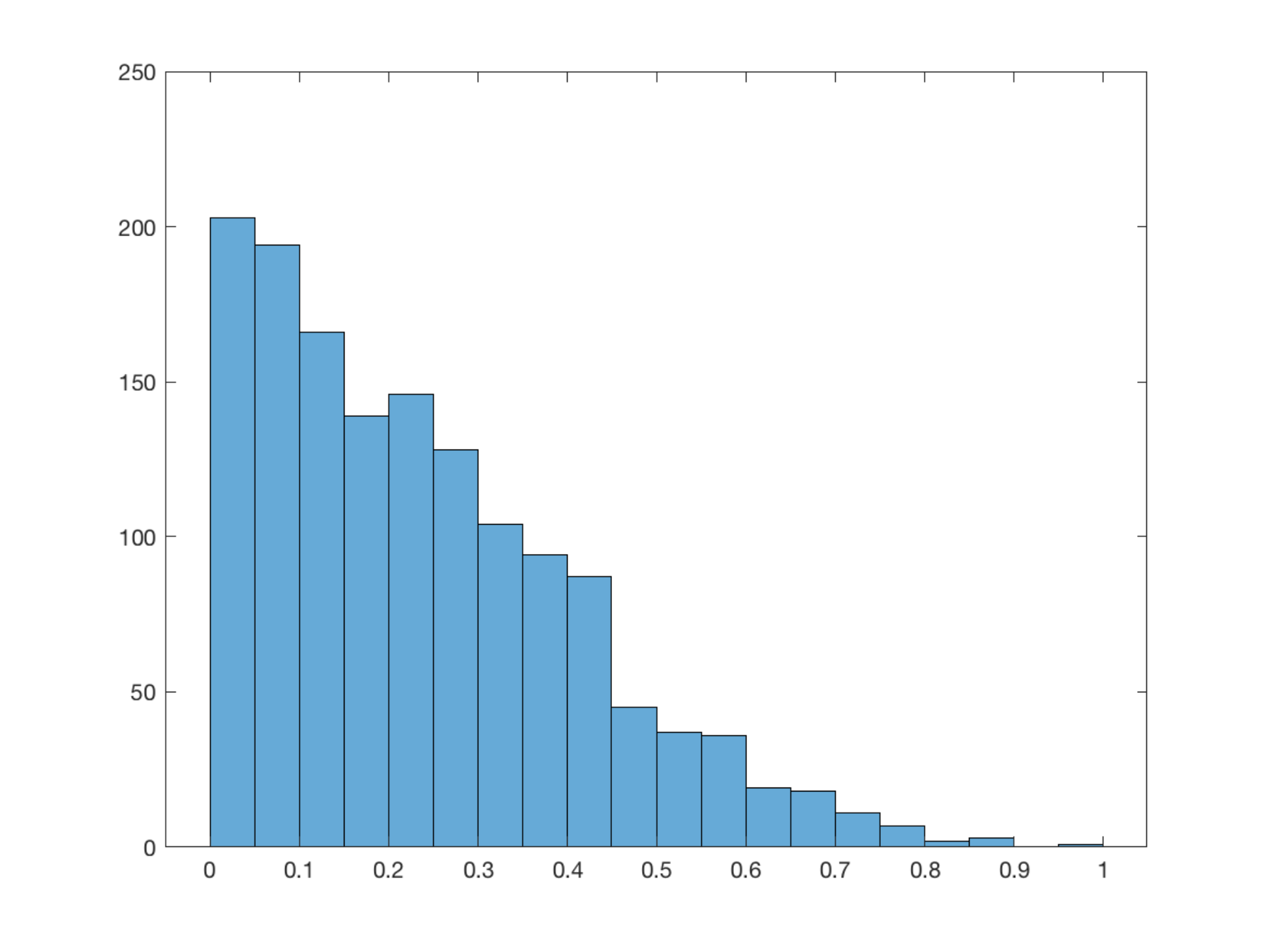}}
\hspace{0.01mm}
\subfloat[$\ell_2$-NC]{\includegraphics[width=0.39\textwidth]{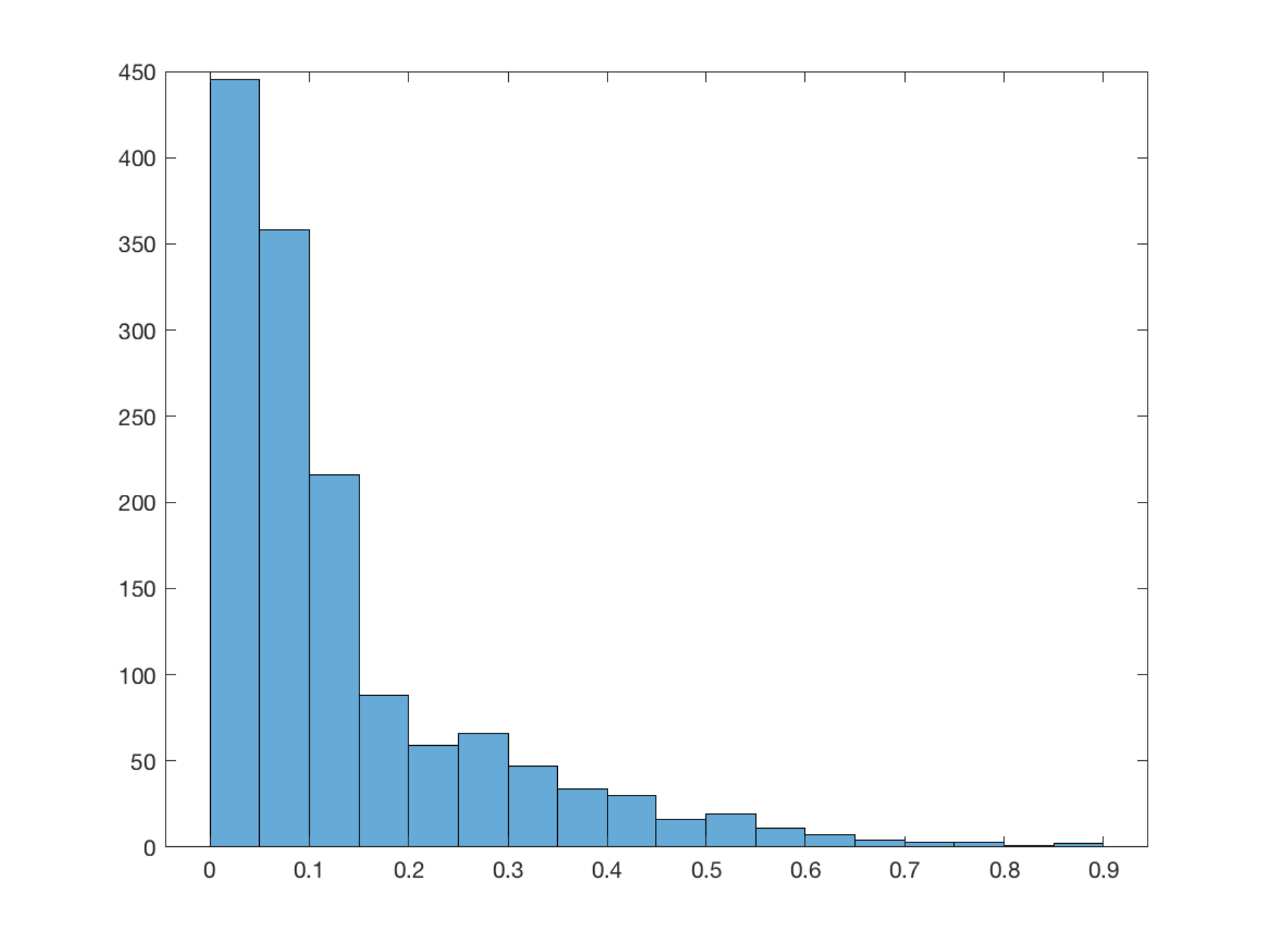}}
\hspace{0.01mm}
\subfloat[KL-$\ell_1$]{\includegraphics[width=0.39\textwidth]{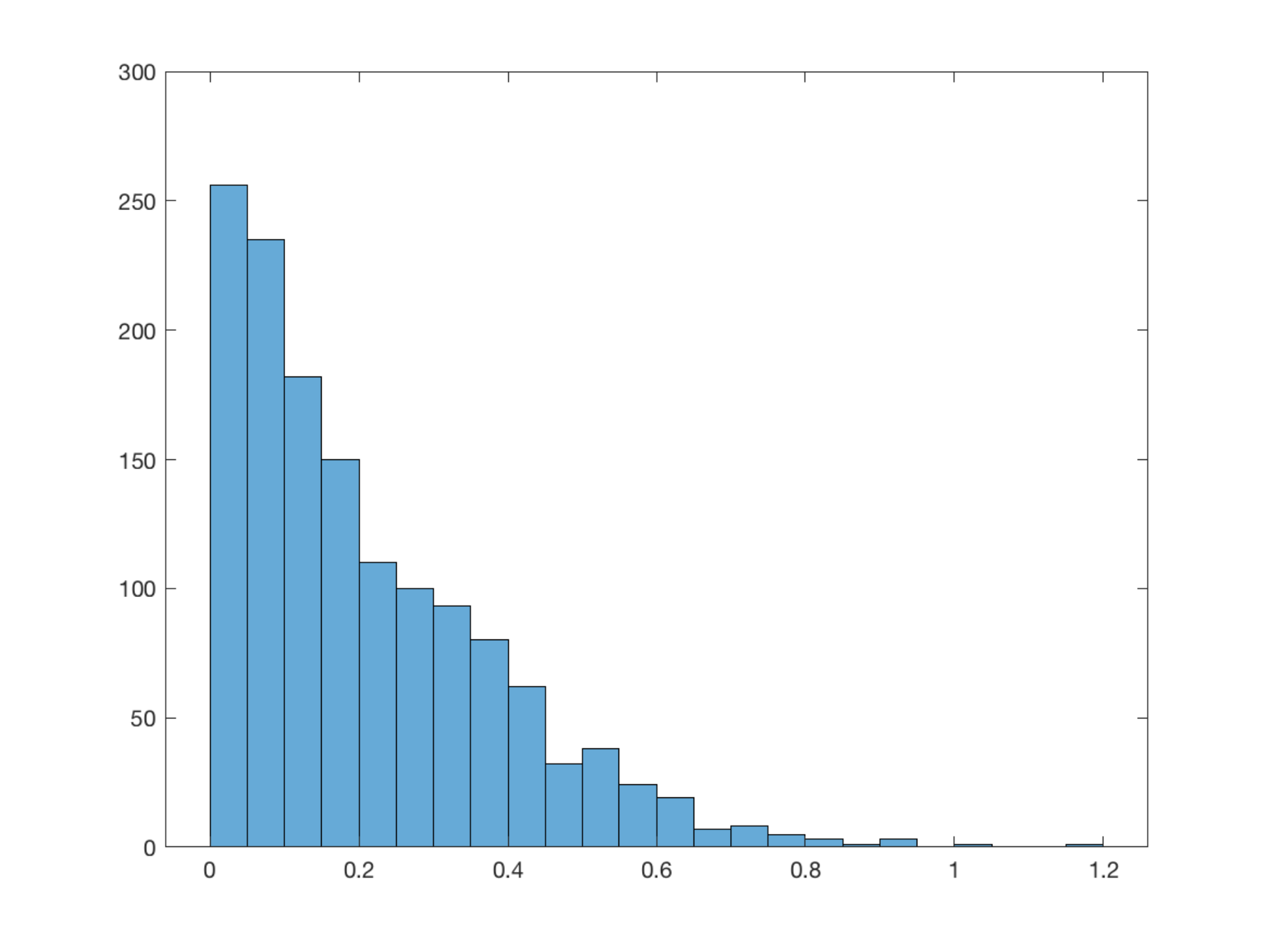}} 
\hspace{0.01mm}
\subfloat[KL-NC]{\includegraphics[width=0.39\textwidth]{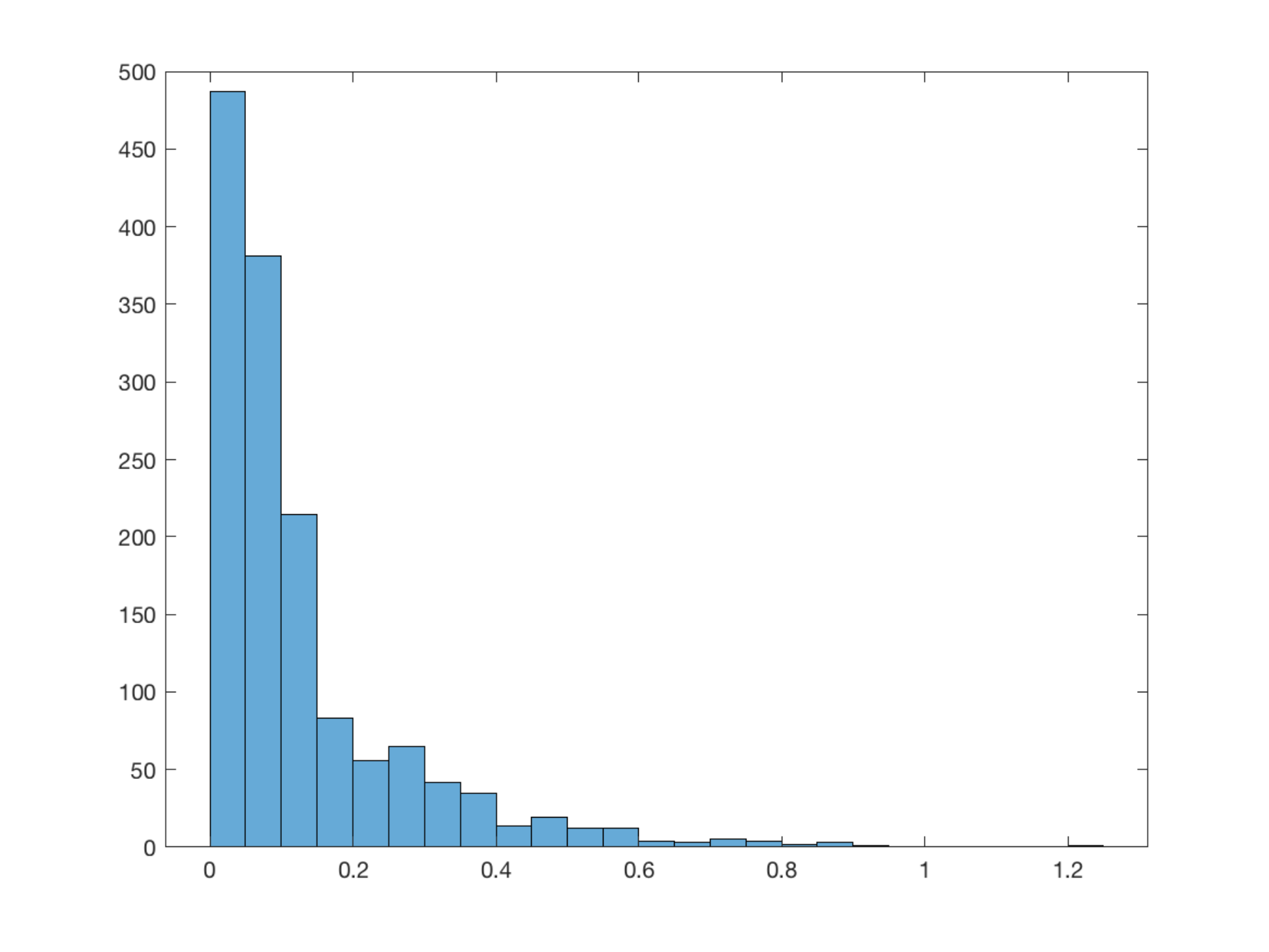}}
\caption{Histogram of relative errors of flux values in the 30 point sources case. } 
\label{fig:30flux_per}
\end{figure}
\subsection{Low number of photons case}
In the previous subsection, we saw the advantages of using our nonconconvex regularization term. However, KL-NC does not seem to possess significant advantages over  the $\ell_2$-NC model, especially in the cases of low source density. This is because only when we have a sufficiently large number of photons,
as was the case in the results discussed earlier, the $\ell_2$ data-fitting term corresponding to
the maximum likelihood estimator for the  Gaussian noise model is well approximated by the KL data-fitting term for the Poisson noise model assumed here. In order to verify this, we study the case of low photon number in this subsection. In \Cref{fig:photons}, we  see the observed image with three different photon numbers per source, namely 500, 1000, and 2000. For the case of 500 photons per source, the image is very dim, and we can barely distinguish the rotating PSF signal from the background noise, so we do not consider this case any further.
For the case of 1000 photons being emitted by each point source, the corresponding image follows a Poisson distribution with a mean of 1000 photons. We take each image to contain 15 point sources. As before, we first randomly generated 20 images, which were used for training parameters. Then we tested 50 different observed images by using the trained parameter values. We obtained {90.00\% for the recall rate and 78.66\%} for the precision rate in the KL-NC algorithm, but only {84.67\% and 59.87\%}, respectively, for the two rates in the $\ell_2$-NC algorithm. For evaluation of the flux estimates, we refer to \Cref{fig:15flux_per1000} from which we see that once again KL-NC gives much better results than $\ell_2$-NC since the bars  corresponding to  the low relative error case are higher. We therefore conclude that the KL data-fitting term  is much better than the $\ell_2$ data-fitting term  when dealing with low numbers of photons.
\begin{figure}[htbp]
\centering
	\resizebox{\textwidth}{0.24\textheight}{\includegraphics{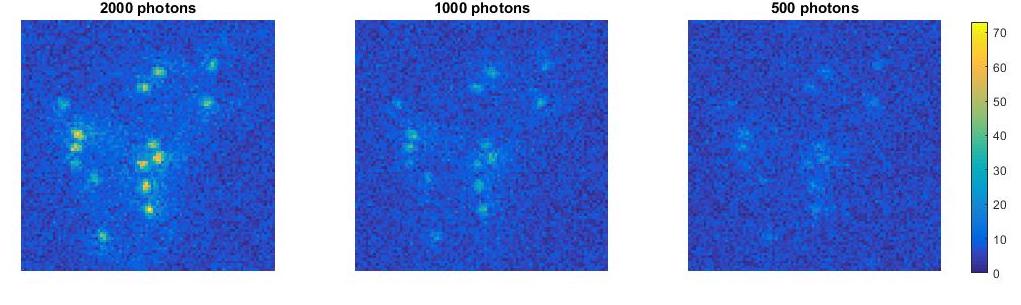}}
\caption{Observed image with different numbers of photons in each point source.  }\label{fig:photons}
\end{figure}
\begin{figure}[h!]
\centering
\subfloat[$\ell_2$-$\ell_1$]{\includegraphics[width=0.39\textwidth]{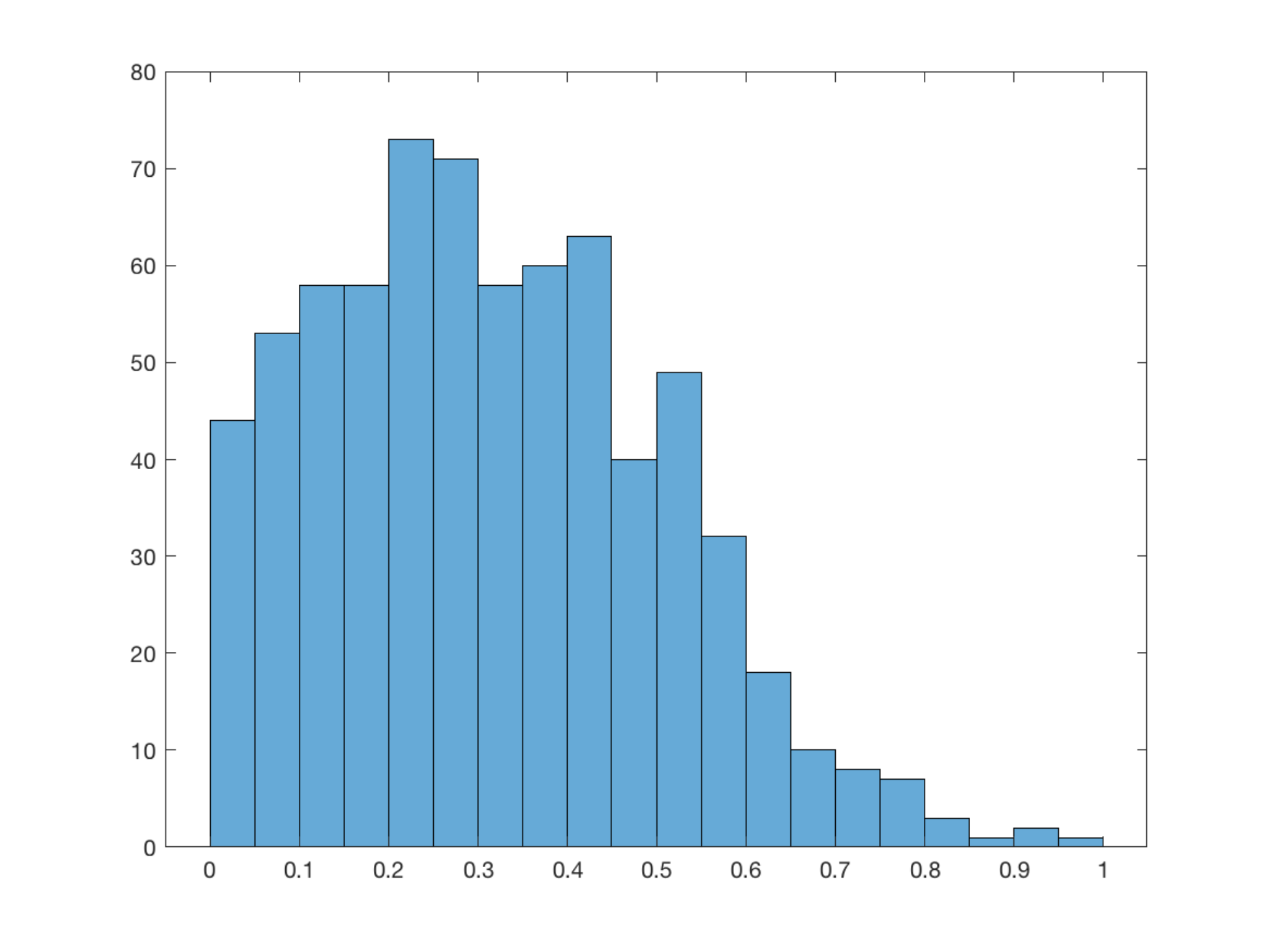}}
\hspace{0.01mm}
\subfloat[$\ell_2$-NC]{\includegraphics[width=0.39\textwidth]{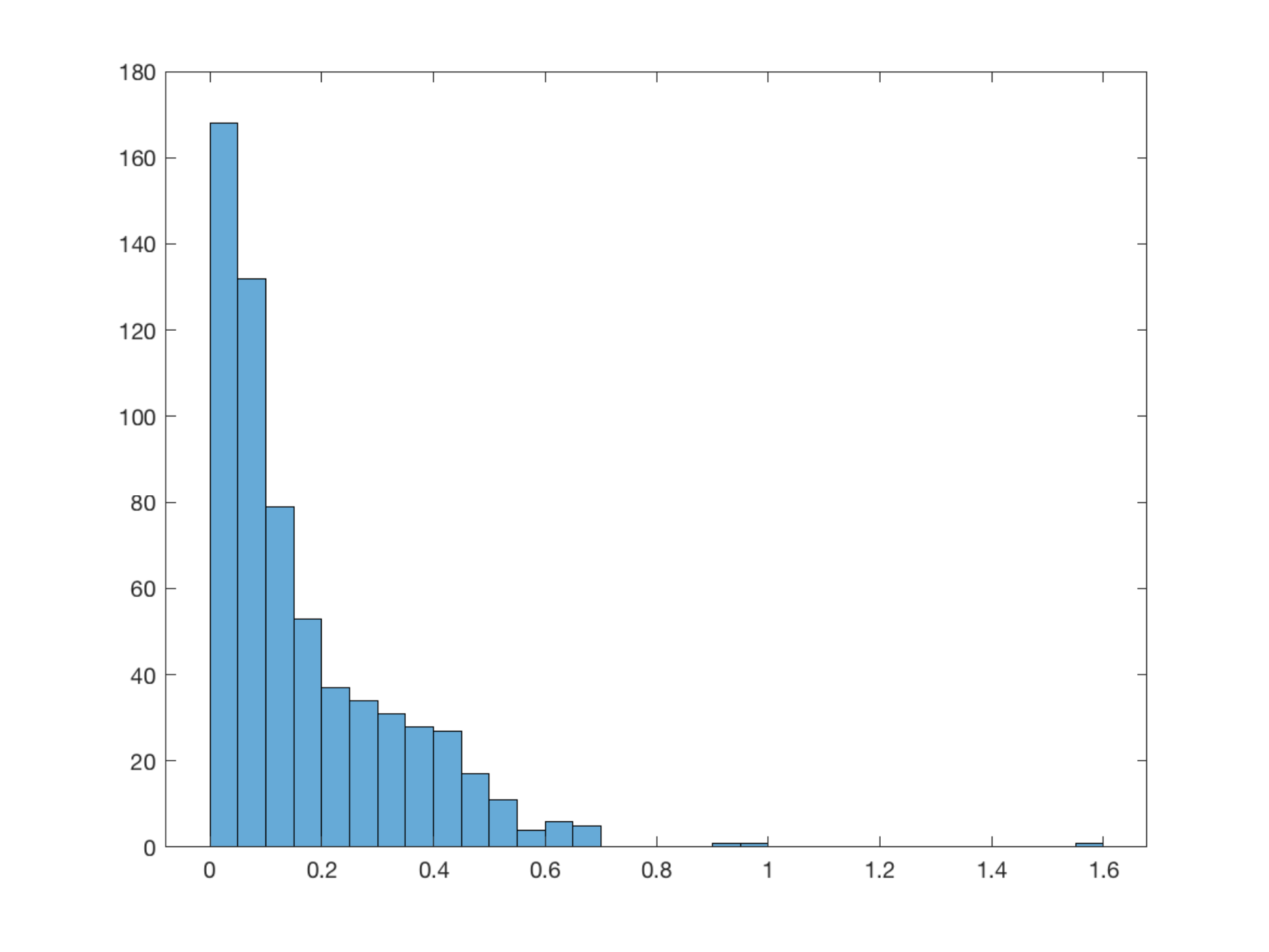}}
\hspace{0.01mm}
\subfloat[KL-$\ell_1$]{\includegraphics[width=0.39\textwidth]{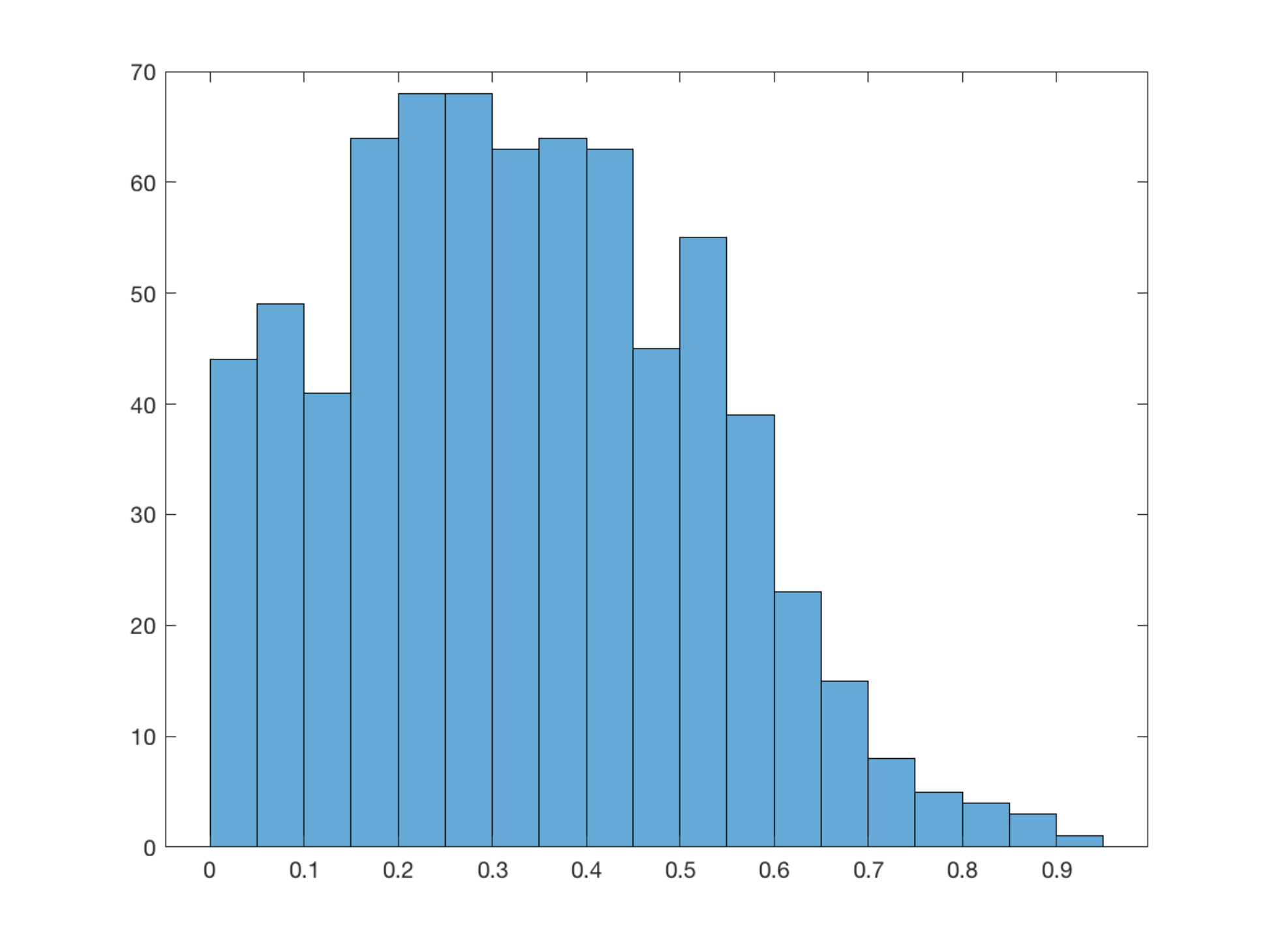}}
\hspace{0.01mm}
\subfloat[KL-NC]{\includegraphics[width=0.39\textwidth]{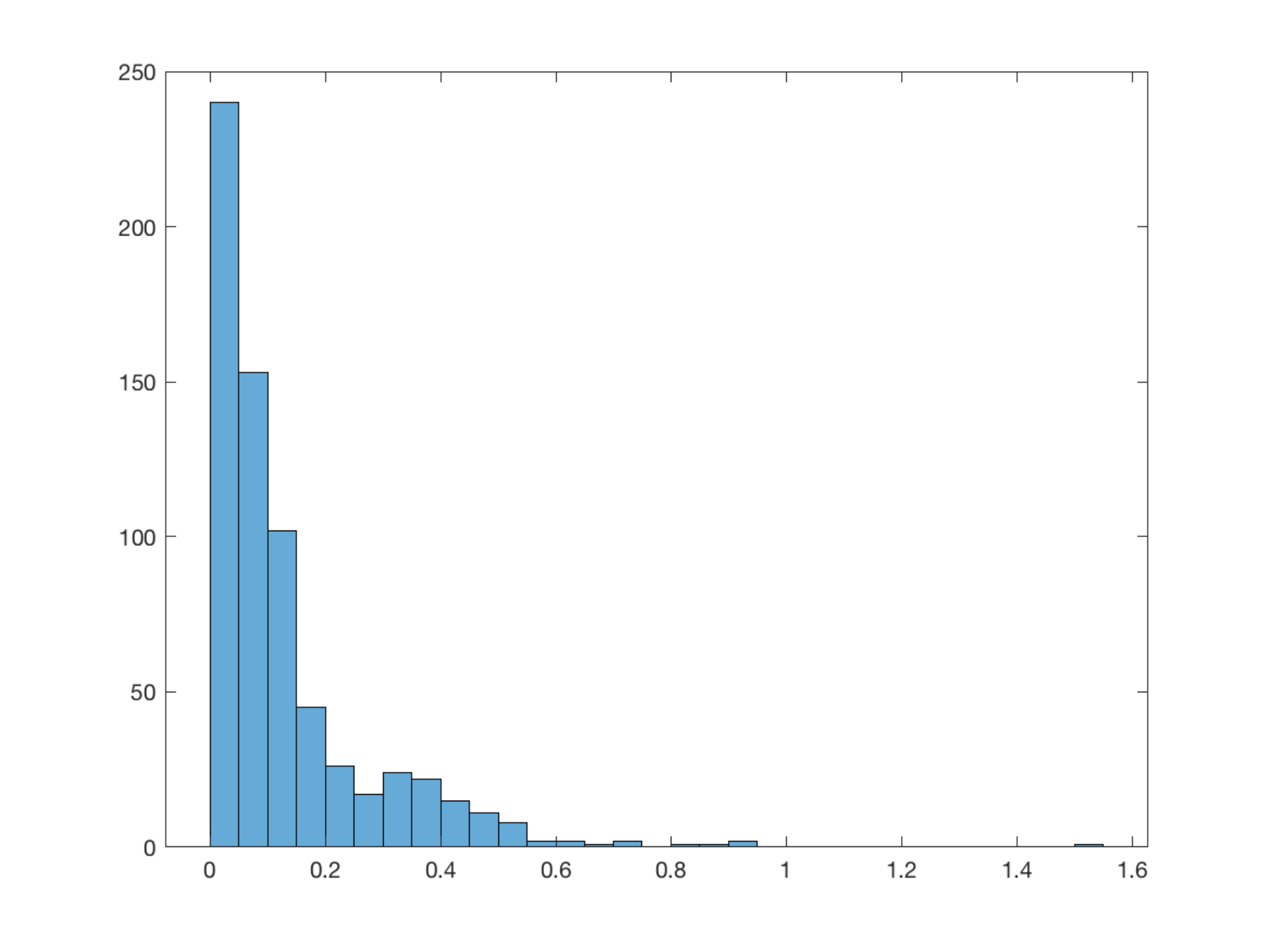}}
\caption{Consider the 1000 photons per point source case. Histogram of the relative errors of the fluxes for the 15 point sources.  }
\label{fig:15flux_per1000}
\end{figure}
\vspace{-0.3cm}
\section{Conclusions  and Current Directions} \label{sec:Conclusions}
In this paper,  we have proposed a non-convex optimization algorithm for 3D localization of a swarm of randomly spaced point sources using a specific rotating PSF which has a single {primary} lobe in the image of each point source. {Rotating PSFs with a single lobe seem to possess several obvious advantages over those that have multiple lobes, such as the double-lobe rotating PSF, e.g. \cite{DH2008pavani,DH2009pavani,Rice2016generalized}, or spread out rapidly with increasing source defocus, such as the standard Airy PSF \cite{Goodman17}, particularly when dealing with relatively high source densities and relatively small photon numbers per source.}  We have employed a post-processing step based both on centroiding the locations of recovered sources that are tightly clustered  and thresholding the recovered flux values  to eliminate obvious false positives from our recovered sources. In addition, we have proposed a new iterative scheme for refining the estimate of the source fluxes after the sources have been localized. 
{In our numerical results, we simulated noisy image data according to the} {physics of the space debris problem at optical wavelengths using a space-based telescope.} {To our knowledge, our algorithm is the first one developed so far for snapshot 3D localization of space debris and tracking of space debris via a rotating PSF approach. 
}

These techniques can be applied to other {rotating and depth-encoding} PSFs for accurate 3D localization and flux recovery of point sources in a scene from its image data under the Poisson noise model.  Applications include not only 3D localization of  space debris, but also super-resolution 3D single-molecule localization microscopy, e.g.  \cite{Book_SR_micro2017,3dsml2017review}. 

{A theoretical analysis of the advantages of our non-convex model over convex models} {and their convergence properties for the rotating PSF problem is planned for study in a future paper. Complete recovery based on the theory of super-resolution is the  direction we will use to further analyze PSF stability and related questions.} Tests of this algorithm based on real data collected using phase masks fabricated for both applications {will be studied in the future. Work involving snapshot multi-spectral imaging, which will permit accurate  material characterization as well as higher 3D resolution and localization of space micro-debris via a sequence of snapshots, is currently under way.}

\bibliographystyle{siamplain}
\bibliography{ref_rPSF}
\end{document}